\documentclass[english,10pt]{article}
\usepackage{amsfonts}
\usepackage[left=2.3cm,right=2.0cm, bottom = 2.3cm, top=2.4cm]{geometry}
\usepackage[utf8]{inputenc}
\usepackage{babel}
\usepackage{slantsc}
\usepackage{array}
\usepackage{amsmath}
\usepackage{amsthm}
\usepackage{amssymb}
\usepackage{enumerate}
\usepackage{bbm}
\usepackage{tikz}
\usepackage{subfigure}
\usepackage{pgfplots}
\usepackage{utopia}
\usepackage[labelsep=period]{caption}
\usepackage{setspace}
\usepackage{url}
\usepackage{textpos}
\usepackage[bookmarks=false]{hyperref}

\usepackage{caption}
\usepackage{wrapfig}
\usepackage{eufrak}

\usepackage{graphicx}

\usetikzlibrary{patterns,snakes}

\newtheorem{thm}{Theorem}
\newtheorem{cor}{Corollary}
\newtheorem{lem}{Lemma}
\newtheorem{claim}{Claim}

\newtheorem{assum}{Assumption}
\newtheorem{prop}{Proposition}

\newtheorem{defn}{Definition}
\newtheorem{exmp}{Example}
\newtheorem{rem}{Remark}

\newcommand{\R}{{\mathbb{R}}}

\newcommand{\G}{{\mathcal{G}}}

\newcommand{\V}{{\mathcal{V}}}

\newcommand{\Sp}{\hspace{0.05cm}}
\newcommand{\B}{\mathbf}
\newcommand{\BB}{\boldsymbol}
\newcommand{\ER}{\Xi_{\mathcal G}}
\newcommand{\PP}{\mathbb{P}}

\definecolor{PennBlue}{RGB}{001,031,091}
\definecolor{PennRed}{RGB}{153,0,0}
\definecolor{NewBlue}{RGB}{001,031,110}
\definecolor{NewRed}{RGB}{200,0,0}
\hypersetup{
pdfborder = {0 0 0},
    colorlinks,
    citecolor=PennRed,
    filecolor=PennRed,
    linkcolor=PennRed,
    urlcolor=PennRed
}
\begin{document}

\title{Time-Delay Origins of  Fundamental Tradeoffs Between \\ Risk of Large Fluctuations and Network Connectivity}
\author{Christoforos Somarakis, Yaser Ghaedsharaf, Nader Motee
\thanks{The authors are with the Department of Mechanical Engineering and Mechanics, Lehigh University, Bethlehem, PA, 18015, USA.   \{\tt\small csomarak,ghaedsharaf,motee\}@lehigh.edu}.}%

\maketitle

\begin{abstract}
For the class of noisy time-delay linear consensus networks, we obtain explicit formulas for risk of large fluctuations of a scalar observable as a function of Laplacian spectrum and its eigenvectors. It is shown that there is an intrinsic tradeoff between risk and effective resistance of the underlying coupling graph of the network. The main implication is that increasing network connectivity, increases the risk of large fluctuations. For vector-valued observables, we obtain computationally tractable lower and upper bounds for joint risk measures. Then, we study behavior of risk measures for networks with specific graph topologies and show how risk scales with network size.  
\end{abstract}
\section{Introduction}
The notion of systemic risk describes fragility in interconnected networks that can result in global or cascading failures due to either relatively small disturbances at the subsystem level or larger and more malicious types of disruptions affecting the whole network. The challenges of dealing with risk spread through many areas of engineering and economy \cite{Tahbaz13,artzner97,artzner99,6248170,fouque13,rockafellar07}. Modern networked control systems bear innate  complexities in their structure and evolving dynamics. There are numerous networks, which are playing major roles in fabrics of our society, that are susceptible to exogenous stochastic uncertainties and  disturbances. Examples include platoon of autonomous vehicles in highways, synchronous power networks with integrated renewable sources, water supply networks, interconnected networks of financial institutions,  and air traffic networks. The recent crisis show specific vulnerabilities of modern networks due to weaknesses in their structures, e.g., air traffic congestion, power outages, the 2008 financial crisis, and other major disruptions. The costly and severe in magnitude outcomes of these network-related crisis motivate risk-oriented analysis and synthesis of networked control systems.

There are some recent lines of research that are close in spirit to the subject of this paper. In \cite{6248170}, the authors introduce novel concepts of resilience  for dynamic flow networks, where they study robustness of such network with respect to perturbations that reduce the flow functions on the links of the network. They also discuss a situation in which imposing finite density constraints on the links may result in cascaded failures. Other related works \cite{Bamieh12,zamp15,15acc,Siami16TACa}  utilize performance measures based on $\mathcal{H}_2$-norm to quantify quality of noise propagation in noisy linear consensus networks. The focus of these works are on linear consensus networks with no time delay. In \cite{Siami18TAC, Siamisparsification15}, the authors introduce a class of spectral systemic performance measures (which includes $\mathcal{H}_p$-norms as its special case) to synthesize networks  through growing and sparsification of links in noisy linear consensus networks in absence of time delay. In all these papers, performance measures  assess holistic and {\it macroscopic} features of the networks, such as total deviation from average or uncertainty volume of the output. In this paper, we show that exogenous uncertainties along with time delay creates unorthodox    behaviors in {\it microscopic} level that cannot be captured by the existing performance measures in literature.  

The class of time-delay linear consensus model \eqref{first-order} has been previously studied within different disciplines  \cite{PhysRevE.92.062816,PhysRevE.86.056114,kuechler92,PhysRevE.90.042135,yasernecsys16,yaserecc16}. In \cite{PhysRevE.92.062816}, the authors obtain expressions for the average size of fluctuations using modal behavior of the network. They also show how extreme fluctuations scale with network size  and argue that  distribution of those nodes with highest fluctuation takes specific forms.  The focus of \cite{PhysRevE.86.056114} is on synchronizability threshold and the scaling behavior of the width of the synchronization for linear consensus systems with uniform and multiple time delays, where their results and analysis are heavily relied on numerical simulations. Reference \cite{PhysRevE.90.042135}  proposes a method to study systems with multiple delays (scalar and multidimensional) in the presence of additive noise. Using Laplace transform, they obtain an integral formula for statistics of the state variables without providing closed form solutions for the integrals. Their results are applied to first and second order consensus networks with multiple delays and complete (all-to-all) graph topology. Extensive numerical solutions to draw conclusions are evident in all these works. In this paper, we adopt a rigorous terminology and develop tools based on notions of risk to quantify and measure the tendency of noisy time-delay linear consensus networks  towards exhibiting undesirable systemic events (e.g., large fluctuations and violating critical constraints). Our analysis is based on the fact that a time-delay linear consensus network can be decomposed into many one-dimensional systems and then analyzed using tools developed in  \cite{kuechler92,yasernecsys16,yaserecc16}. 

Systemic risk as a tool for decision making under uncertainty has been initially employed to study uncertainties in financial markets \cite{artzner97,artzner99,fouque13,rockafellar07,Rockafellar_Royset_2015}. More recently, cascading failures in interconnected financial institutes have been studied \cite{Tahbaz13,fouque13}. In \cite{sommiladnader16, somyasnader17}, the notion of value-at-risk has been applied to calculate probability of large fluctuations due to external disturbances in linear consensus networks with no time delay.  The present work is an outgrowth of \cite{somyasnader17} and \cite{sommiladnader16} in a number of ways. At first, we introduce and analyze two types of risk measures to assess possibility of large fluctuations in networks: one gauge risk in probability and the other one in expectation w.r.t. a utility (cost) function.  By exploiting functional properties of risk measures, such as (quasi-) convexity and monotonicity,  we calculate explicit formulas for them w.r.t. scalar network observables and show that these individual risk measures are spectral functions of Laplacian eigenvalues and eigenvectors.  Furthermore, it is shown that the risk measures are not monotone functions of Laplacian eigenvalues, i.e., increasing connectivity (through adding new feedback interconnections or increasing feedback gains) may increase risk of large fluctuations and sparsification (eliminating existing feedback loops) may decrease risk of large fluctuations. We do not observe this peculiar behavior when there is no time delay.  In addition to the fundamental limits reported to \cite{somyasnader17}, we characterize several Heisenberg-like inequalities and reveal intrinsic tradeoff between the risk measure and network connectivity. Our results assert that in presence of time delay more connectivity results in higher risk of large fluctuations that in turn may lead to violation of critical constraints and network fragility.  Finally, we extend our risk analysis to multiple network observables. In this case, the risk measures take vector values and one needs to calculate joint risk measures. In general, obtaining explicit closed form formulas for the joint risk measures is nearly impossible. However, one can employ numerical algorithms to compute values of such joint risk measures \cite{Genz:2009:CMN:1695822}. To tackle this challenge, we obtain tight lower and upper bounds for the joint risk measures and show that joint risk measures depend on the vector of individual risk measures, where we know how to calculate them explicitly. The big advantage of our bounds is that they are computationally tractable. Finally, we investigate behavior of risk for networks with particular graph topologies and show how the risk measures scale with network size. Several numerical examples support our theoretical findings. The proofs of all technical results,  in this work, are placed in the Appendix.

\allowdisplaybreaks

\section{Preliminaries}\label{sect: prelim}
Throughout the paper, plain and bold lowercase letters are reserved for scalar-valued and vector-valued variables, respectively. The $n$-dimensional Euclidean space with elements $\mathbf z=(\Sp z_1,\Sp \dots \Sp,z_n \Sp)^T$ is denoted by $\mathbb R^n$. The set of standard Euclidean basis for $\R^n$ is $\{\mathbf e_1,\ldots, \mathbf e_n \}$. The nonnegative orthant cone is defined by
\[\R^n_+ = \big\{ \mathbf{z} \in \R^n ~\big|~z_i \geq 0 ~~\textrm{for all}~~i=1,\ldots,n  \big\}.\]
For  $\BB x, \BB z \in \R^n$,  relation $\BB x \preceq \BB z$ holds iff $\BB z - \BB x \in \R^n_+$, and  $\BB x \prec \BB z$ iff  all elements of $\BB z - \BB x$ are strictly positive. The Euclidean vector and induced matrix norms is represented by $\| \cdot \|$, while absolute value operation for vectors is defined element-wise, i.e., $|\mathbf z|=[ \Sp |z_1|,\Sp \ldots \Sp,|z_n| \Sp ]^T$.  The $n \times 1$ vector of all ones is shown by $\mathbf 1_n$ and the $n \times n$ centering matrix is defined as $M_n := I_n - \frac{1}{n} \mathbf 1_n \mathbf 1_n^T$, where $\mathbf m_i$ is the $i$'th column of $M_n$. Representation of a vector in $\mathbb{R}^n$ with respect to linearly independent columns of matrix $Q=[\mathbf q_1~|~\dots~|~\mathbf q_n]$ is shown by $\mathbf z^Q=[\Sp z_1^Q,\Sp \dots \Sp,z_n^Q \Sp]^T$, where $z_i^Q=\mathbf q_i^T\mathbf z$.

\noindent {\it Algebraic Graph Theory:} 
 A weighted graph is the triple $\mathcal G=(\V, \mathcal{E},w)$, where $\V$ is the set of nodes, $\mathcal{E}$ is the set of links (edges),  and $w: \mathcal{E} \rightarrow \mathbb{R}_{+}$ is the weight function of the graph that assigns a real  number to every link.  
  
\begin{assum}\label{assum0}
For every pair of nodes $i,j \in \mathcal{V}$,  we assume: (i)~nonnegativity of link weights, i.e., $w(i,j) \geq 0$ 
(ii)~indirectness of links, i.e., $w(i,j)=w(j,i)$, and  (iii)~simpleness, i.e., $w(j,j)=0$. Moreover, this paper only considers connected graphs. 
\end{assum}

\noindent Let us denote $k_{ij}=w(i,j)$ for every $i,j \in \V$. The Laplacian matrix of graph $\G$ is an $n \times n$ matrix $L=[l_{ij}]$ whose elements are defined as 
\begin{equation}\label{eq: laplacian}
\displaystyle l_{ij} := \left\{\begin{array}{ccc}
-k_{ij} & \textrm{if} & i \neq j \\
 &  &  \\
k_{i1}+\ldots+k_{in}& \textrm{if} & i=j
\end{array}\right..
\end{equation} 
Laplacian matrix of a graph is  symmetric and positive semi-definite. Our connectivity assumption  implies that only the smallest Laplacian eigenvalue is equal to zero and all other ones are strictly positive, i.e., $0=\lambda_1 < \lambda_2 \leq \ldots \leq \lambda_n$. The eigenvector of $L$ corresponding to  $\lambda_k$  is denoted by $\mathbf q_k$. By letting $Q=[\Sp \mathbf q_1~|~\dots~|~\mathbf q_n \Sp]$, it follows that $L=Q \Lambda Q^T $ with $\Lambda=\mathrm{diag}[\Sp 0,\Sp \lambda_2,\Sp \ldots \Sp, \Sp\lambda_n \Sp]$. We normalize the Laplacian eigenvectors such that $Q$ becomes an orthogonal matrix, i.e., $Q^T Q = Q Q^T =I_n$ with  $\mathbf q_1=\frac{1}{\sqrt{n}} \mathbf 1_n$. 

\begin{rem}\label{rem: cntrincid}
Suppose that $E_n$ is the $n\times n$ diagonal matrix that is obtained by setting element of $(1,1)$ of $I_n$ to zero. For an $n\times n$ unitary matrix $Q$, standard graph theory results are:  $M_n=QE_nQ^T$ and $B_n^T B_n=n M_n=n Q E_n Q^T$, where $B_n$ is the $\frac{n(n-1)}{2}\times n $ complete incidence matrix.
\end{rem}

The notion of effective resistance is a widely used metric to measure how strongly a graph is connected in the context of graph theory \cite{klein1993, Mieghem:2011:GSC:1983675,Siami16TACa}. The effective resistance between every pair of nodes  $i$ and $j$ of a graph can be interpreted as the voltage difference between terminals $i$ and $j$ if we replace every link with an electrical resistance (whose value is equal to that link's weight) and let unit current source to be applied between the two nodes; we refer to  \cite{Mieghem11,klein1993} for more information. The total effective resistance is the sum of all effective resistances over all distinct pair of nodes in $\mathcal G$ and is equal to \cite{klein1993} 
\begin{equation}\label{eq: graphr}
\ER =n\sum_{i=2}^n\lambda_i^{-1},
\end{equation} where $\lambda_2,\ldots,\lambda_n$ are the non-zero Laplacian eigenvalues.

\noindent{\it Probability Theory:} 
A complete filtered probability space is the quadruple $(\Omega,\mathcal F, \{ \mathcal F_t\}_{t\geq 0}, \mathbb P)$, where $\{\mathcal F_t\}_{t\geq 0}$ is the filtration of a vector-valued  standard Brownian motion $\{\mathbf w_t\}_{t\geq 0}$, with $\mathbf w_t=\big(w_1(t),\ldots, w_n(t)\big)^T$ constructed out of  $n$ independent scalar Brownian motions. The space of all $\mathcal F$-measurable, $\mathbb R^q$-valued random variables  with finite second moment is denoted by $\mathbb L^2(\mathbb R^q)$. The stochastic process $\mathbf y=\{\mathbf y_t\}_{t\in \mathbb R}$ is $\{\mathcal {F}_{t}\}$-adapted if $\mathbf y_t$ is $\mathcal {F}_{t}$-measurable for  $ t \geq 0$. The multi-variable normal distribution is denoted by $\mathcal N(\boldsymbol \mu, \Sigma)$, where $\boldsymbol \mu \in \mathbb R^q$ is the expected value and  $\Sigma$ is the $q\times q$ positive semi-definite covariance matrix. A random variable $\mathbf y\sim \mathcal N(\boldsymbol \mu, \Sigma)$ attains a density function if and only if $\Sigma$ is positive definite. The error function is denoted as
$$\mathrm{erf}(z)=\frac{1}{\sqrt{\pi}}\int_{-z}^{z}e^{-t^2}\,dt.$$ For fixed $\varepsilon \in (0,1)$, the following quantity is introduced for the sake of convenience:
\begin{equation}\label{eq: sgaussiant} S_{\varepsilon}(\alpha)=\inf\bigg\{\delta \in \mathbb R_+~\bigg|~ \frac{1}{\sqrt{\pi}}\int_{-\delta-2\alpha}^{\delta}e^{-t^2}\,dt\geq 1-\varepsilon \bigg\} \end{equation}
and $S_{\varepsilon}(0)=\inf\big\{\delta \in \mathbb R_+~|~~\mathrm{erf}(\delta)\geq 1-\varepsilon\big\}$. 


\section{Problem Statement}\label{sect: model}

We consider the class of time-delay linear consensus networks that has found wide applications in engineering (e.g., clock synchronization in sensor networks, rendezvous in a space or time, and heading alignment in swarm robotics) and social sciences (e.g., agreement in opinion networks); we refer to \cite{84640184aa38444fb42776d48b3e313d,4118472} for more details. As a motivational application, we discuss a cooperative rendezvous problem where the control objective of a team of agents (e.g., ground vehicles or quadcopters) is to meet simultaneously at a pre-specified location known to all agents\footnote{Rendezvous in space is very similar to rendezvous in time by switching the role of time and location in the above explanation.}. In this rendezvous problem agents do not have a priori knowledge of the meeting time as it may have to be adjusted in response to unexpected emergencies or exogenous uncertainties (see \cite{84640184aa38444fb42776d48b3e313d} for a detailed discussion). Thus, all agents should agree on a rendezvous time by achieving consensus. This is usually done by each agent creating a state variable, say $x_i \in \R$, that represents its estimation of the rendezvous time. Each agent's initial estimate is set to a preferred time at which it would be able to rendezvous with other agents. The rendezvous (more generally consensus) dynamics for each agent evolves in time according to the following stochastic differential equation
\begin{equation}
d{x}_i(t) = u_i(t) dt + b \Sp d w_i(t) \label{TI-consensus-algorithm} 
\end{equation}
for all $i=1,\ldots,n$. Each agent's control input is $u_i \in \R$. The source of uncertainty is diffused in the network as additive stochastic noise,  the magnitude of which is uniformly scaled by the diffusion  $b\in \mathbb R$. The impact of uncertain environments on dynamics of agents are modeled by independent Brownian motions $w_1, \ldots, w_n$. In real-world situations, each agent experiences a time delay in accessing, computing, or sharing its own state information with itself and other neighboring agents \cite{84640184aa38444fb42776d48b3e313d}. It is assumed\footnote{This assumption has been widely used by other researchers as it allows analytical derivations of formulas. For rendezvous in time, it is a common practice in robotics labs to use identical communication modules for all agents, which results in a uniform communication time delay.  Moreover, in other related applications such as  heading alignments, rendezvous in space, and  velocity control of agents using  a Motion Capture (MoCap) system  to observe their spatial locations in indoor labs \cite{LUPASHIN201441,5569026,sipahi08}, all agents experience an identical time delay to access data through MoCap system.}  that all agents experience an identical time delay that is equal to a nonnegative number $\tau$. The control inputs are determined via a negotiation process by forming a linear consensus network over a communication graph using the following feedback law
\begin{equation}
u_i(t) ~=~\sum_{j=1}^{n} k_{ij} \big(x_j(t-\tau) - x_i(t-\tau)\big),\label{feedback-law}
\end{equation}
where $k_{ij}$ are non-negative feedback gains. Let us denote the state vector by $\mathbf x_t = \big(x_1(t),  \ldots,  x_n(t)\big)^{\rm T}$ and the vector of exogenous disturbance by $\mathbf w_t = \big(w_1(t),  \ldots,  w_n(t)\big)^{\rm T}$. The dynamics of the resulting closed-loop network  can be cast as a linear consensus network that is governed by the stochastic delay differential equation
\begin{equation}\label{first-order}
d\mathbf x_t=-L\Sp \mathbf x_{t-\tau} \Sp dt+ B \Sp d\mathbf w_t 
\end{equation}
for all $t\geq 0$, where the initial function $\mathbf x_t= \boldsymbol \phi(t)$ is statistically independent of $\mathbf w_0$ for all $t \in [-\tau, 0]$, and $B=b\, I_n$. The underlying coupling structure of the consensus  network \eqref{first-order} is a graph $\G=(\V,e, w)$ that satisfies Assumption \ref{assum0} and whose Laplacian matrix is $L$.  The underlying communication graph is considered to be time-invariant with a time-independent  Laplacian matrix as the network of agents aim to reach consensus on  a rendezvous time before each agent can perform motion planning to get to the meeting location.  Upon reaching consensus, a properly designed internal feedback control mechanism steers each agent towards the rendezvous location.  
  
\begin{assum}\label{assum1} The time delay satisfies $\tau<\frac{\pi}{2\lambda_n}$. 
\end{assum}

When there is no noise, i.e., $b=0$, it is already known \cite{Olfati04} that under Assumptions  \ref{assum0} and \ref{assum1}, states of all agents converge to average of all initial states $\frac{1}{n}\mathbf 1_n^T \boldsymbol \phi(0)$; whereas in presence of input noise,  state variables  fluctuate around the network average $\frac{1}{n}\mathbf 1_n^T \mathbf x(t)$. In order to quantify the quality of consensus (e.g., rendezvous) and its fragility features, we consider the vector of observables for network \eqref{first-order} 
\begin{equation}\label{output}
\mathbf y_t \Sp = \Sp C \mathbf x_t
\end{equation}
in which $C$ is a generic $q\times n$ output matrix and $\B y_t =\Sp (y_1(t),\ldots,y_q(t)\Sp )^T$. Some relevant examples of network observables are discussed in Section \ref{subs: statisticsofoutput}.  Assumption \ref{assum0} implies that one of the modes of network \eqref{first-order} is marginally stable with eigenvector $\mathbf 1_n$, while the rest are exponentially stable. The marginally stable mode, which corresponds to the zero eigenvalue of $L$, must be unobservable from the output \eqref{output} so that in the long run, $\mathbf y_t$ stays bounded. This can be guaranteed by imposing that $\mathbf 1_n$ belongs to the null space of $C$, i.e. 
$$\mathbf 1_n \in \ker(C).$$
According to \eqref{first-order}, when there is no input noise, the aforementioned condition implies $\B y_t \rightarrow 0$ as $t$ tends to infinity. Consequently, the presence of exogenous noise input makes the network observables to fluctuate around zero.  This implies that agents will not be able to agree upon an exact consensus state (e.g., rendezvous time). A practical resolution is to allow a tolerance interval for agents to concur. 

\begin{defn}\label{delta-consensus}
For observable \eqref{output} and some $\BB \delta \in \R^q_+$, dynamical network \eqref{first-order} can tolerate some degrees of disagreement and reach $\boldsymbol \delta$-consensus if the following stochastic event
\begin{equation}
\lim_{t\rightarrow\infty}|\B y_t| \preceq \BB \delta \label{constraint-1}
\end{equation}
 holds with high probability\footnote{ High probability means a probability larger than a predefined cut-off number close to one.}. 
\end{defn}

\noindent The notion of  $\boldsymbol \delta$-consensus means that all agents have agreement on all points in $
\big\{ \mathbf x \in \R^n ~\big|~ |C \mathbf x|\preceq \BB \delta \big\}$. We interpret this definition using our rendezvous example, where $C=M_n$. Suppose that event \eqref{constraint-1} holds with $\boldsymbol \delta=\delta \mathbf 1_n$. Then, the network of agents will achieve a $\delta$-consensus over the rendezvous time in the following sense: In steady state, the $i$'th agent is assured that by $x_i(t)+2\delta$ units of time, all other agents will arrive and meet each other in that time interval with high probability. Some alternative undesirable situations may also happen that we refer to as systemic events.
 


\begin{defn}\label{def-sys-event}
For a given vector $\BB \delta \in \R^q_+$,  network \eqref{first-order} with observable \eqref{output} will be prone to a systemic event if the probability of happening the following stochastic  event 
\begin{equation}
|y_1(t)| > \delta_1 ~\vee ~\ldots~\vee~ |y_q(t)| > \delta_q\label{constraint-2}
\end{equation}
in steady-state, is nonzero, where $\vee$ is the disjunction operator.  
\end{defn}

\noindent In the rendezvous example, when the network of agents undergo systemic event \eqref{constraint-2} with some probability greater than $\varepsilon >0$, it would be possible that the agents fail to reach $\boldsymbol \delta$-consensus and meet simultaneously during the rendezvous time interval.

The objective of this paper is to identify how time-delay makes dynamical network \eqref{first-order} vulnerable to systemic events \eqref{constraint-2} and quantify probability of such undesired stochastic events that may lead to network fragility, i.e., not being able to reach $\boldsymbol \delta$-consensus. 

\section{The State and Output Dynamics}\label{sec:stat}
 The dynamical network \eqref{first-order} satisfies all the well-posedness conditions of the classical theory of stochastic differential equations and it generates well-defined stochastic processes $\{\mathbf x_t\}_{t\geq -\tau}$ and $\{\mathbf y_t\}_{t\geq -\tau}$; we refer to \cite{evans13,Mohammed84} for more details. In this section, we discuss  statistics of these stochastic processes and derive results that will be useful in derivations of our results.

\subsection{Internal Stability} 
 The solution of the unperturbed system, i.e.,  $b=0$, can be described using the transition matrix of the system, which is the matrix solution  of  differential equation
\[ \dot{\Phi}_L(t)=-L \Phi_L(t)\]
with initial conditions $\Phi_L(t)=0$ if $t\in [-\tau,0)$ and $\Phi_L(0)=I_n$. According to \cite{Olfati04}, we have the following limit behavior as $t\rightarrow+\infty$  
\[\Phi_L(t)\longrightarrow  \frac{1}{n}\mathbf 1_n\mathbf 1_n^T\]
exponentially fast if and only if $\tau$ satisfies Assumption \ref{assum1}.   From Assumption \ref{assum0}, the transition matrix can be decomposed as 
\[\Phi_L(t)=Q \Sp \Phi(t)\Sp Q^T=Q \Sp \mathrm{diag}[\varphi_1(t), \ldots, \varphi_n(t)] \Sp Q^T,\hspace{0.3in}\] 
for $t\geq -\tau$. Namely, $\Phi(\cdot)$ is the principal solution of $\dot {\mathbf x}=-\Lambda \mathbf x(t-\tau)$ and $\Lambda$ the diagonal matrix of Laplacian eigenvalues. We remark that the diagonal elements of $\Phi(t)$ do not come in closed form for $\tau>0$, with the exception of $\varphi_1(t)\equiv 1$. Nevertheless, under Assumption \ref{assum1}, it can be shown that   
\begin{equation}\label{eq: functionf}
\int_0^{\infty}\varphi_i(t)\,dt=\tau f(\lambda_i\tau)~~\text{where}~~f(x)=\frac{1}{2x}\frac{\cos(x)}{1-\sin(x)},
\end{equation} is a function defined in $(0,\frac{\pi}{2})$ that takes values in $\mathbb R_+$. The plot of $f$ is provided in Figure \ref{fig1: functionf}. The rigorous derivation of formula \eqref{eq: functionf} is held in Appendix \ref{append: dde}.
\begin{figure}
  \vspace{0pt}
  \begin{center}
    \includegraphics[scale=0.8]{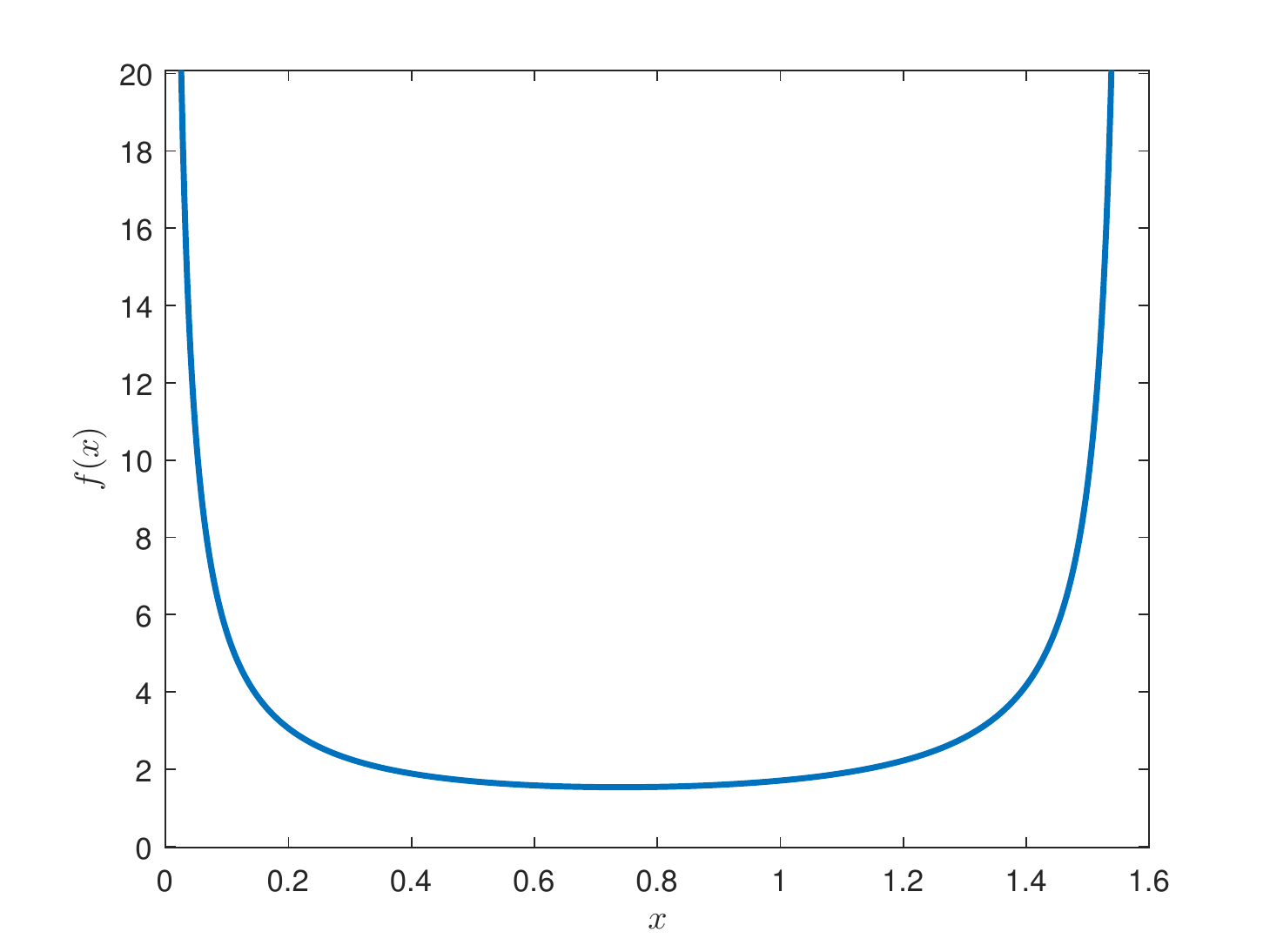}
  \end{center}
  \vspace{-10pt}
\caption{\small Plot of $f(x)=\frac{1}{2x}\frac{\cos(x)}{1-\sin(x)}$ over its domain $(0,\frac{\pi}{2})$ that guarantees internal stability of network \eqref{first-order}.}\label{fig1: functionf}
  \vspace{-10pt}
\end{figure}

\subsection{The Statistics of the Solution}
The solution of the perturbed network \eqref{first-order} is an $\{\mathcal F_t\}_{t\geq 0}$-measurable stochastic process $\{\mathbf x_t \}_{t\geq -\tau}$ given by 
\begin{equation*}
\mathbf x_t=\mathbf v_t+\int_0^{t}\Phi_L(t-s)B\,d\mathbf w_s
\end{equation*} in which  
\[\mathbf v_t=\Phi_L(t)\boldsymbol \phi(0)-L\int_{-\tau}^{0}\Phi_L(t-s-\tau)\boldsymbol\phi(s)\,ds.\] 
Moreover, it follows that $\mathbf x_t\sim \mathcal N(\BB \mu_t,\Sigma_t)$ with vector-valued mean $\BB \mu_t  =\mathbb{E}\big[\mathbf v_t\big]$ and covariance matrix
\begin{equation}\label{eq: generalaverageparameters}
\Sigma_t  = b^2 \int_0^{t}\Phi_L(s)\Phi_L^T(s)\,ds.
\end{equation}

\subsection{The Statistics of the Output}
Consequently, the statistics of the output can be quantified as
\begin{equation}\label{eq: outputdist}\mathbf y_t\sim \mathcal N(C \BB \mu_t,C \Sigma_t C^T).\end{equation} 
The following result characterizes output matrices for which $\mathbf y_t\in \mathbb L^2(\mathbb R^q),~t\geq 0$.

\begin{lem}\label{lem: boundofcov}  
$\mathbf 1_n \in \ker(C)$ if and only if 
\[ \sup_{t \geq 0}~\big\|C\Sigma_{t}C^T\big\|<\infty.\] 
\end{lem}

The result of Lemma \ref{lem: boundofcov} implies that if $C\mathbf q_1=0$, then one gets $\overline{\mathbf y}=\lim_{t \rightarrow \infty} \mathbf y_t \in \mathbb L^2(\mathbb R^q)$ as both $\overline{\Sigma}=\lim_{t \rightarrow \infty} C\Sigma_{t}C^T$ and $\overline{\boldsymbol \mu}=\lim_{t \rightarrow \infty} C\boldsymbol \mu_t$ are well-defined. Moreover, $\overline{\mathbf y}\sim \mathcal N(0, \overline{\Sigma})$  if and only if $\mathbf 1_n \in \ker(C)$.

\section{Systemic Risk Measures}\label{sec: riskmeasures} 


There are several known quantifiers in the context of finance and economy to measure risk of a systemic event. In this paper, we utilize two notions of risk measures, which are relevant to the context of dynamical systems, to quantify safety margins of network \eqref{first-order} before  it experiences systemic event \eqref{constraint-2}  and reaches a $\delta$-consensus based on Definitions \ref{delta-consensus} and \ref{def-sys-event}.

\subsection{Risk of Large Fluctuations in Probability }
When network observables admit a well defined probability measure, we may be able to calculate risk of a systemic event that is expressed in terms of violation of a set of critical constraints; for example, stochastic event \eqref{constraint-2} in the rendezvous example. The value-at-risk measure, also referred to as risk in probability, determines maximum allowable fluctuation, which we also refer to it as safety margin, before network experiences a systemic event with a pre-specified  probability. In other words, this measure quantifies the smallest lower bound $\delta$ for systemic event $z(\omega)\in \mathbb L^2(\mathbb R)$ before probability of event ``$z(\omega)$ surpassing $\delta$'' exceeds a given threshold. 

\begin{defn}\label{eq: risk}
For a given probability threshold $\varepsilon\in [0,1)$, the value-at-risk or risk in probability is an operator $\mathcal R_{\varepsilon}: \mathbb L^2(\mathbb R)\rightarrow \mathbb R$ that is defined by 
\begin{equation*}
\mathcal R_{\varepsilon}(z)=\inf \Big\{\delta \in \mathbb R~\Big|~\mathbb P\big(z(\omega)>\delta\big)< \varepsilon \Big\}.
\end{equation*}  If the right hand-side has no solution, we set $\mathcal R_{\varepsilon}(z)=+\infty$.
\end{defn}
\noindent The value-at-risk measure enjoys several functional properties that will help us later on to derive fundamental limits on the best achievable risk value and tradeoffs between risk and network connectivity.  

\vspace{0.1in}
\begin{prop}\label{prop: cohrisk}
The measure $\mathcal R_{\varepsilon}: \mathbb L^2(\mathbb R)\rightarrow \mathbb R$ in Definition \ref{eq: risk} satisfies the properties:
 \begin{itemize}
\item[(i)] $\mathcal R_{\varepsilon}(z)=z$ if $z$ is independent of $\omega \in \Omega$.
\item[(ii)] $\mathcal R_{\varepsilon}(z+m)=\mathcal R_{\varepsilon}(z)+m$ for all $m\in \mathbb R$.
\item[(iii)]If  $z_1\leq z_2$, then $\mathcal R_{\varepsilon}(z_1)\leq \mathcal R_{\varepsilon}(z_2)$ 
\item[(iv)] $\mathcal R_{\varepsilon}(\alpha z)=\alpha \mathcal R_{\varepsilon}(z)$ for $\alpha>0$.
\item[(v)] $\mathcal R_{\varepsilon}(z)$ is a quasi-convex function of $z$.
\end{itemize}
\end{prop}

\noindent Property (i) implies that if a $z$ includes no uncertainty, the risk measure remains invariant.  Property (ii) means that $\mathcal R_{\varepsilon}$ is not affected when a fixed constant is added to the random variable. Property (iii) explains that if a random variable $z_1$ takes smaller values than $z_2$ almost surely, then the risk of the former is smaller than the risk of the latter. Next, (iv) is a positive homogeneity property and implies that if we scale a random variable (i.e., scale the uncertainty with it) then the risk will scale accordingly. Finally, property (v) is a weak convexity property of the risk measure and indicates that the risk of a convex combination of two random variables cannot get worse than the maximum of the risk of these random variables. 

\subsection{Risk of Large Fluctuations in Expectation} 

The second type of risk measures that are relevant to the context of dynamical systems is quantification of risk through statistical expectation of a utility or cost function. In several applications in economics, game theory, and decision and control theory, we evaluate violation of critical constraints against some meaningful utility or cost functions. These utility functions reflect  perception of a decision maker towards a potential  outcome when there is uncertainty  \cite{follmer11}.
In this setting, utility function can be interpreted as a distance function to measure how far a random variable is from a predefined threshold. This interpretation is  particularly useful in the context of dynamical systems as one can employ energy or entropy based utility functions to quantify risk.  

\begin{defn}\label{eq: riskmom}
Suppose that  utility function $v: \R\rightarrow \R$ is  convex and monotonically increasing. For a given threshold $\varepsilon$, the risk in expectation is an operator $\mathcal T_{\varepsilon}:\mathbb L^2(\mathbb R)\rightarrow \mathbb R$ that is defined by 
\begin{equation*}
\mathcal T_\varepsilon(z)=\inf\Big\{\delta \in \mathbb R ~\Big|~ \mathbb{E}\big[v \big(z(\omega)-\delta \big) \big ]< v(\varepsilon)\Big\}.
\end{equation*}  If the right hand-side has no  solution, we set $\mathcal T_{\varepsilon}(z)=+\infty$.
\end{defn}

The risk in expectation quantifies possibility of large fluctuations by evaluating expected value of a predefined utility or cost function. A trivial choice for utility function is $v(z)=z$, where the value of risk is the smallest number for which inequality $\mathbb{E}[z(\omega)-\delta] \leq \varepsilon$ holds. For more general utility functions, we can still relate to this simple inequality by applying Jensen's inequality and monotonicity property of the utility function. It is straightforward to show that  $\mathcal T_\varepsilon$ is positive if and only if $\mathbb E\big[v(z)\big]> v(\varepsilon)$ (we refer to \cite{follmer11} for more details). This helps us to define the following acceptance set
\begin{equation}\label{eq: acceptance}
\mathbb A= \Big\{~z \in \mathbb L^2(\mathbb R)~\Big| ~\mathcal T_{\varepsilon}(z) < 0 ~\Big\}
\end{equation}
that contains all random variables in $\mathbb L^2(\mathbb R)$ whose expected cost $\mathbb E[v(z)]$ is less than cost threshold $v(\varepsilon)$. In the following, we discuss two important and relevant utility functions for risk analysis in dynamical systems. 

%
%

\vspace{0.1in}

\noindent {\it Quadratic Utility:} For scalar random variables, the quadratic utility function is $v(z)=z^2$ over nonnegative real numbers. The value of risk in this case represents the smallest value $\delta$ for which random variable $z$ stays in  a ball with center $\delta$ and radius $v(\varepsilon)$. Depending on the application, quadratic form of the output \eqref{output} relates to the potential or kinetic energies of the system \cite{Siami16TACa}. 

\vspace{0.1in}

\noindent {\it Exponential Utility:} For scalar random variables, the exponential utility function is $v(z)=e^{\beta z}$ for some parameter $\beta>0$. This function is monotonically increasing over all real numbers and close in spirit to the entropic risk measure \cite{follmer11}. 

\vspace{0.1in}
\begin{prop}\label{prop: riskmom} The risk in expectation $\mathcal T_{\varepsilon}:\mathbb L^2(\mathbb R)\rightarrow \mathbb R$ as in Definition \ref{eq: riskmom}  satisfies the following properties:
\begin{itemize}
\item[(i)] $\mathcal T_{\varepsilon}(z+m)=T_{\varepsilon}(z)+m$ for all $m\in \mathbb R$. 
\item[(ii)] $\mathcal T_{\varepsilon}(z)$ is a convex function of $z\in\mathbb A$.
\end{itemize}
\begin{itemize}
\item[(iii)] $z_1\leq z_2$ implies $\mathcal T_{\varepsilon}(z_1)\leq \mathcal T_{\varepsilon}(z_2)$.
\item[(iv)] For all nonnegative  $z\in \mathbb A$, it follows that 
\[ \mathcal T_{\varepsilon}(\alpha z) \geq \alpha \mathcal T_{\varepsilon}(z) +(\alpha-1)\varepsilon \]
for all $\alpha\geq 1$, and for all nonpositive $z\in \mathbb A$, we have
\[ \mathcal T_{\varepsilon}(\alpha z) \leq \mathcal \alpha \mathcal T_{\varepsilon}(z)-(\alpha-1)\varepsilon\]
for all $\alpha \in [0,1]$.  
\end{itemize}
\end{prop}  
\noindent Property (i) is the equivalent of Property (ii) in Proposition \ref{prop: cohrisk}. It explains that adding fixed numbers (with no uncertainty) does not affect the value of risk in expectation.  The convexity property of Property (ii) holds within the acceptance set \eqref{eq: acceptance} and can be generalized to any acceptance set that is convex \cite{follmer11}. Property (iii) means that risk in expectation is monotonically increasing. According to property (iv), the risk measure enjoys a weak positive homogeneity property that holds over a subset of $\mathbb A$.

%

\section{Risk Assessment of a Single Event}\label{subs: statisticsofoutput}
Suppose that the network observable is scalar
\begin{equation}\label{scalar-y}
y_t=\mathbf c^T\mathbf x_t
\end{equation}
with $\mathbf c \in \R^n$. According to \eqref{eq: outputdist}, it follows that $y_t\sim \mathcal N(\mu_t,\sigma_t^2)$ with  
\begin{eqnarray}
\mu_t \hspace{-0.2cm} &  = & \hspace{-0.2cm}\mathbb{E}\bigg[\sum_{k=1}^n c_k^Q\big[\varphi_k(t)\phi_k^Q(0)+\int_{-\tau}^{0}\varphi_k(t+s)\phi_k^Q(s)\,ds\big]\bigg] \label{eq: meanvaluescalar}\\
\sigma_t^2 \hspace{-0.2cm} &=& \hspace{-0.2cm} b^2\sum_{k=1}^n (c_k^Q)^2 \int_0^{t}\varphi_k^2(s)\,ds, \label{eq: standarddevscalar}
\end{eqnarray} 
where $\mathbf c^Q=[c_1^Q,\dots,c_k^Q]^T$  and $\boldsymbol \phi^Q(t)=[\phi_1^Q,\dots,\phi_k^Q]^T,t\in [-\tau,0]$ are the corresponding vectors $\mathbf c$ and $\boldsymbol\phi(t)$ represented in terms of the basis spanned by the columns of $Q=[\mathbf q_1 \Sp|\Sp\dots \Sp|\Sp \mathbf q_n]$, for $\mathbf q_k$ is the $k$'th Laplacian eigenvector, and $\{\varphi_k\}_{k\in \V}$ the diagonal elements of principal matrix $\Phi$.


 \begin{lem}\label{prop: uniqueness}  For every time instant $t$, functions $\mu_t=\mu_t(Q)$ and  $\sigma_t=\sigma_2(Q)$, as they are defined by \eqref{eq: meanvaluescalar} and \eqref{eq: standarddevscalar}, are eigenspace invariant, i.e., for all orthonormal matrices $Q_1$ and $Q_2$ for which $Q_1 \Lambda Q_1^T=Q_2 \Lambda Q_2^T$, we have 
\[\mu_t(Q_1)=\mu_t(Q_2)~~~\textrm{and}~~~ \sigma_t(Q_1)=\sigma_t(Q_2).\]
\end{lem} 

\begin{thm}\label{thm: main0} Suppose that stochastic process $\mathbf x=\{\mathbf x_t\}_{t\in [-\tau,T]}$ is the solution of  \eqref{first-order} and the network observable is scalar as in \eqref{scalar-y}. Then, the risk of large  fluctuations in probability is given by 
\begin{equation}\label{eq: riskprob}
\mathcal R_{\varepsilon}(|y_t|)=\sqrt{2} \sigma_t S_{\varepsilon}\bigg(\frac{\mu_t}{\sqrt{2} \sigma_t}\bigg)+\mu_{t},
\end{equation}
where $S_{\varepsilon}(\cdot)$ is defined by  \eqref{eq: sgaussiant}, and the risk of output fluctuations violating a utility threshold can be calculated as follows:

\noindent (i) when $u(z)=z^2$, 
 \begin{equation}\label{eq: riskmeanquad}
\mathcal T_{\varepsilon}(|y_t|)=\begin{cases} \mu_{|y_t|}-\sqrt{\varepsilon^2-\sigma_{|y_t|}^2}&  \textrm{if}~~~ \varepsilon\geq \sigma_{|y_t|} \\
+\infty & \textrm{otherwise}
\end{cases}
\end{equation}

\noindent (ii) when $u(z)=e^{\beta z}$ for some $\beta>0$, 
\begin{equation}\label{eq: exprisk}\begin{split}
\mathcal T_{\varepsilon}(|y_t|)=\frac{\beta \sigma_t^2}{2}&+\frac{\ln\big(\kappa(\mu_t,\sigma_t)/2 \big)}{\beta}-\varepsilon
\end{split}
\end{equation}
%
in which 
\begin{eqnarray*}
\mu_{|y_t|}&=&\sigma_t\sqrt{\frac{2}{\pi}}e^{-\frac{\mu_t^2}{2\sigma_t^2}}-\mu_t~\mathrm{erf}\bigg(\frac{-\mu_t}{\sqrt{2\sigma_t^2}}\bigg)\\
\sigma_{|y_t|}^2&=&\mu_{t}^2+\sigma_t^2-\mu_{|y_t|}^2,\\
\kappa(\mu_t,\sigma_t) &=&\bigg[1-\mathrm{erf}\bigg(-\frac{\mu_t}{\sqrt{2}\sigma_t}-\frac{\beta\sigma_t}{\sqrt{2}}\bigg)\bigg]e^{\beta \mu_t}\\
& &+\bigg[1+\mathrm{erf}\bigg(-\frac{\mu_t}{\sqrt{2}\sigma_t}+\frac{\beta\sigma_t}{\sqrt{2}}\bigg)\bigg]e^{-\beta \mu_t}
\end{eqnarray*} 
and $\mu_t$ and $\sigma_t$ are defined by  \eqref{eq: meanvaluescalar}-\eqref{eq: standarddevscalar}. 
\end{thm}



In case of no time delay, i.e., $\tau=0$, the eigensolutions of the nominal system take simple forms $\varphi_k(t)=e^{-\lambda_k t}$ for $k=2,\ldots,n$. This allows us to derive a more explicit form for the variance: 
$$\sigma_t^2=b^2\sum_{k=2}^n~(c_k^Q)^2~\frac{1-e^{-2\lambda_k t}}{2\lambda_k}.$$  When $\tau>0$ and $t<\infty$, the risk expressions in Theorem \ref{thm: main0} do not admit more tractable representations  in general. However, towards the end of this section, we show  that in steady state (as $t \rightarrow \infty$) one can obtain more explicit formulas for systemic risk measures; we also refer to \cite{sommiladnader16} for more details.

The steady state behavior of systemic risk measures can be investigated by letting $t$ tend to infinity. According to Lemma \ref{lem: boundofcov}, if output matrix $C$ is orthogonal to the vector of all ones, then $\overline{y}=\lim_{t \rightarrow \infty} y_t$ will belong to $\mathbb L^2(\mathbb R)$, i.e. with finite variance, denoted as $\overline{\sigma}$. If $\mathbf c^T \mathbbm{1}_n =0$, then $\mu_{t}$ vanishes exponentially fast. In steady state, the transient effect of initial conditions disappears and other origins of large volatility and fluctuation in network observables reveal themselves, e.g., communication  topology, time delay, and the diffusion coefficient $b$. The steady-state risk analysis uncovers how exogenous noise propagates throughout the network, is amplified by various factors, and results in rapid and unpredictable variations in network observables.     


\begin{thm}\label{thm: main1}
Suppose that stochastic process $\mathbf x=\{\mathbf x_t\}_{t \geq -\tau}$ is the solution of  \eqref{first-order} and the network observable is scalar as in \eqref{scalar-y} and the output matrix satisfies $\mathbf c^T  \mathbbm{1}_n=0$. Then, the risk of large  fluctuations in probability in steady-state is given by 
\begin{equation}\label{eq: riskprob}
\mathcal R_{\varepsilon}(|\overline{y}|)=\sqrt{2} \Sp S_{\varepsilon}(0) \Sp \overline{\sigma} 
\end{equation} 
in which 
 \begin{equation}\label{variance}
\overline{\sigma} ^2 \Sp = \Sp \frac{1}{2} \Sp b^2 \Sp \sum_{k=2}^n  \Sp \frac{\cos(\lambda_k\tau)}{\lambda_k  \big(1-\sin(\lambda_k\tau)\big)}\Sp (c_k^Q)^2
\end{equation} 
where $\lambda_k$ is the $k$'th non-zero Laplacian eigenvalue. The risk of output fluctuations violating a utility threshold in steady-state w.r.t.:

\noindent (i) quadratic utility $v(x)=x^2$ is given by 
\begin{equation}\label{eq: riskmeanquadss}
\mathcal T_{\varepsilon}(|\overline{y}|)=\begin{cases}\sqrt{\frac{2}{\pi}}\overline{\sigma}  -\sqrt{\varepsilon^2-\big(1-\frac{2}{\pi}\big)\overline{\sigma}^2} & \textrm{if}~~ \varepsilon\geq \overline{\sigma}  \sqrt{1-\frac{2}{\pi}}\\
+\infty & \text{otherwise}\end{cases}
\end{equation}

\noindent (ii) exponential utility $v(x)=e^{\beta x}$ for some $\beta>0$ can be expressed by 
 \begin{equation}\label{eq: expriskss}
 \begin{split}
\mathcal T_{\varepsilon}(|\overline{y}|)=\frac{\beta \overline{\sigma}^2}{2}&+\frac{\ln\big(1+\mathrm{erf}(\frac{\beta\overline{\sigma}}{2})\big)}{\beta}-\varepsilon.
\end{split}
\end{equation}
\end{thm}


The expression for variance \eqref{variance} can be equivalently written in the following compact form\footnote{The Moore-Penrose pseudo-inverse of matrix $L$ is denoted by $L ^{\dagger}$. For a matrix $X \in \mathbb{R}^{n \times n}$, the matrix-valued functions of matrices $\cos (X)$ and $\sin(X)$ are defined as $\cos(X) = \sum_{k=0}^{\infty}\frac{(-1)^k}{(2k)!}X^{2k}$ and $\sin(X) = \sum_{k=0}^{\infty}\frac{(-1)^k }{(2k+1)!}X^{2k+1}$.}  
\begin{equation}\label{trace-risk}
\overline{\sigma} ^2 \Sp = \Sp \frac{1}{2} \Sp b^2 \Sp \mathrm{Tr} \Big[L_o \Sp L^{\dagger} \cos(\tau L)\big(M_{n}-\sin(\tau L)\big)^{\dagger}\Big]
\end{equation} 
where $L_o = \mathbf c\mathbf c^T$. In the following, we discuss some of the important and useful examples of scalar network observables.


\noindent {\it (a) Deviation of a single state from the average:} Suppose that we are interested in observing state of agent $i$. The deviation from average of $x_i$ can be obtained by choosing the $i$'th column of the centering matrix $M_n$ as  the output matrix, i.e., $\mathbf c=\mathbf m_i$. A direction calculation reveals that  
\[ c_k^Q = \left\{\begin{array}{ccc}
0 & \textrm{if} & k=1 \\
\mathbf m_i^T \mathbf q_k & \textrm{if} & k=2,\ldots,n
\end{array}\right..\]


\noindent {\it (b) Pairwise deviation:} When one is interested in characterizing volatility of the disagreement between two specific fixed agents, say agents with labels $i$ and $j$, we choose output matrix $\mathbf c=\mathbf e_i- \mathbf e_j $. In this case, 
\[ 
c_k^Q = \left\{\begin{array}{ccc}
0 & \textrm{if} & k=1 \\
(\mathbf e_i-\mathbf e_j)^T\mathbf q_k & \textrm{if} & k=2,\ldots,n 
\end{array}\right..
\]
%

\noindent {\it (c) Deviation from average of immediate neighbors:} Comparing state of each agent with average of its own adjacent neighbors is another meaningful observable for network \eqref{first-order}.  In this case, the components of the output matrix are defined as
\[ 
c_k^Q = \left\{\begin{array}{ccc}
0 & \textrm{if} & k=1 \\
\displaystyle \mathbf q_k^T \left(\mathbf e_i-\frac{1}{n_i}\sum_{j \sim i}\mathbf e_j\right)  & \textrm{if} & k=2,\ldots,n 
\end{array}\right.,
\]
where $\{ j \Sp|\Sp j \sim i\}$ is equal to set $\{j \Sp|\Sp \{i,j\}\in \mathcal{E}\}$ and $n_i$ is its cardinality.

In  all these scenarios, the value of systemic risk measure $\mathcal T_{\varepsilon}(|y_t|)$ and $\mathcal R_{\varepsilon}(|y_t|)$ can be calculated using Theorem \ref{thm: main0} with statistical parameters defined in \eqref{eq: meanvaluescalar} and \eqref{eq: standarddevscalar}. Moreover, Theorem \ref{thm: main1} shows that value of the systemic risk measure depend on Laplacian spectrum. 


\begin{rem}
 When the output matrix is $\mathbf c=\frac{1}{n}\mathbbm 1_n$, the network observable becomes the average of all states. It follows that $c_1^Q=\frac{1}{\sqrt{n}}$ and $c_k^Q=0$ for every $k>1$. In this case, the mean and variance of the scalar observable reduces to
\begin{equation}\label{eq: avparam} \mu_t=\frac{1}{\sqrt{n}}\bigg[\phi_1^Q(0)+\int_{-\tau}^{0}\phi_{1}^Q(s)\,ds \bigg]~~~\textrm{and}~~~\sigma_t^2=\frac{b^2}{n^2}t.
 \end{equation}
Using these formulas, one can compute the value of transient systemic risk measures $\mathcal T_{\varepsilon}(|y_t|)$ and $\mathcal R_{\varepsilon}(|y_t|)$ according to Theorem \ref{thm: main0}. Since network \eqref{first-order} has a marginally stable mode corresponding to $\lambda_1=0$, the variance $\sigma_{t}^2$ diverges as $t$ grows (see also Lemma \ref{lem: boundofcov}). This implies that all states  fluctuate around their average and eventually diverge. This phenomenon is so-called ``flocking to default''  in finance and economy literature; we refer to  \cite[Ch. 17]{fouque13} for a discussion. Furthermore, one also observes that the stastistical  parameters in \eqref{eq: avparam} are independent of the network topology and  time delay.
\end{rem}

We conclude this section by investigating effects of time delay on systemic risk measures.  

\begin{thm}\label{thm: monotonicity}
Suppose that the conditions of Theorem \ref{thm: main1} are satisfied. Then, systemic risk measures  \eqref{eq: riskprob}, \eqref{eq: riskmeanquadss}, \eqref{eq: expriskss} are strictly increasing functions of time delay.
\end{thm}

\begin{rem}
The class of systemic risk measures in this paper exhibit some idiosyncratic behavior in presence of time delay. It can be easily shown that the variance \eqref{variance} is a convex function of eigenvalues. Systemic risks are in turn increasing functions of variance, which is not monotonically decreasing with respect to connectivity. This imposes counter-intuitive challenges in design of optimal networks with respect to this class of risk measures. For instance, increasing connectivity may deteriorate systemic risk of network \eqref{first-order} and sparsification may improve it. Figure \ref{fig1: functionf} depicts behavior of each spectral term in expression of variance \eqref{variance}.  The unorthodox  behavior of risk and its interplay with connectivity and performance in time-delay  networks is highly counterintuitive and requires careful investigations. In Section \ref{sec:tradeoffs}, we characterize an intrinsic tradeoff that explains this behavior.   
\end{rem}

\section{Joint Risk Assessment of Multiple Events}

The extension of risk for network \eqref{first-order} with vector valued output observables   
\begin{equation}\label{observables}
\mathbf y_t=C \mathbf x_t = \left[ \Sp \mathbf c_1^T\mathbf x_t \Sp, \Sp \ldots \Sp,\Sp    \mathbf c_q^T\mathbf x_t \Sp \right]^T 
\end{equation}
is a problem with different approaches, each of which highlights the problem from a different perspective. In this section we introduce the most important multi-value risk measures for the steady state distribution of 
 $|\mathbf y_t|=\big[ \Sp |y_1|,\Sp \dots \Sp,|y_q| \Sp \big]^T$  with  $y_k=\mathbf c_k^T\mathbf x_t$.
 We recall from Section \ref{sec:stat} that in steady-state $$\overline{\mathbf y}\sim \mathcal N({\mathbf 0}, \overline{\Sigma})$$  if and only if $\mathbbm 1_n$ is in the kernel of output matrix $C$.

When different types of observations are used for risk analysis in a dynamical network, systemic risk measure becomes a vector, where every element of that vector corresponds to the value of risk of an event that has a similar nature to the corresponding observation. For example, suppose that our goal is to evaluate risk using vector of observables $\mathbf y_t=[ \Sp y_1(t),y_2(t),y_3(t) \Sp]^T$, where  stochastic variable $y_1$ measures deviation of a single state from the average of all states, $y_2$ shows pairwise deviation between two given agents, and $y_3$ gives deviation of a single state from average of its immediate neighbors. 

\section{Joint Risk Assessment of Multiple Events}

The extension of risk for network \eqref{first-order} with vector valued output observables   
\begin{equation}\label{observables}
\mathbf y_t=C \mathbf x_t = \left[ \Sp \mathbf c_1^T\mathbf x_t \Sp, \Sp \ldots \Sp,\Sp    \mathbf c_q^T\mathbf x_t \Sp \right]^T 
\end{equation}
is a problem with different approaches, each of which highlights the problem from a different perspective. In this section we introduce the most important multi-value risk measures for the steady state distribution of 
 $|\mathbf y_t|=\big[ \Sp |y_1|,\Sp \dots \Sp,|y_q| \Sp \big]^T$  with  $y_k=\mathbf c_k^T\mathbf x_t$.
 We recall from Section \ref{sec:stat} that in steady-state $$\overline{\mathbf y}\sim \mathcal N({\mathbf 0}, \overline{\Sigma})$$  if and only if $\mathbbm 1_n$ is in the kernel of output matrix $C$.

When different types of observations are used for risk analysis in a dynamical network, systemic risk measure becomes a vector, where every element of that vector corresponds to the value of risk of an event that has a similar nature to the corresponding observation. For example, suppose that our goal is to evaluate risk using vector of observables $\mathbf y_t=[ \Sp y_1(t),y_2(t),y_3(t) \Sp]^T$, where  stochastic variable $y_1$ measures deviation of a single state from the average of all states, $y_2$ shows pairwise deviation between two given agents, and $y_3$ gives deviation of a single state from average of its immediate neighbors.

\subsection{Risk of Large Fluctuations in Probability}

The joint risk in probability measure is a vector-valued operator  $\boldsymbol{\mathcal R}_{\varepsilon}: \mathbb L^2(\mathbb R^q)\rightarrow \mathbb R^q$ that is defined by \begin{equation}\label{eq: riskflucmmonec}
\boldsymbol{\mathcal R}_{\varepsilon}(|\overline{\mathbf y}|)  =  \inf\Big\{\boldsymbol \delta \in \R_{+}^q~\Big|~\mathbb P\big(\Sp|\overline{\mathbf y} | \preceq \boldsymbol \delta \Sp \big)\geq 1-\varepsilon\Big\} 
\end{equation}  
in which infimum operates on elements of $\boldsymbol \delta$.  
Obtaining a closed-form expression for the joint risk in probability is particularly difficult. To see this, we observe that \eqref{eq: riskflucmmonec} is equivalent to a joint chance constraint optimization problem \cite{doi:10.1287/opre.1090.0712}.The solution set is not a singleton. In fact, the set of all optimal vectors $\boldsymbol \delta$ constitute a Pareto set. Moreover, problems of this kind involve the quantification of joint dependent events   are typically difficult to solve.   One way to compute solution of such  multi-objective optimization problems is via Monte Carlo sampling \cite{ruszcynsky02,doi:10.1137/050622328}. Robust Optimization techniques propose analytic approximations to this problem by decomposing the joint constraint problem into a problem with individual chance constraints \cite{doi:10.1287/opre.1090.0712}. By applying the latter approach to \eqref{eq: riskflucmmonec}, we obtain the next result. Let us define the vector of individual risk measures by 
\[ \mathfrak{R}_{\varepsilon}(|\overline{\mathbf y}|) = \left[\Sp {\mathcal R}_{\varepsilon}(|\overline{y}_1|),~\ldots~, {\mathcal R}_{\varepsilon}(|\overline{y}_q|) \Sp \right]^T. \]  

\begin{thm}\label{thm: doubleinequality2}
For dynamical network \eqref{first-order} with vector of observables \eqref{observables},  the joint vector-valued systemic risk measure \eqref{eq: riskflucmmonec} satisfies the following inclusion 
\begin{equation}
\boldsymbol{\mathcal R}_{\varepsilon}(|\overline{\mathbf y}|)  ~\subseteq~ \mathbb W_{\mathfrak{R}_\varepsilon (|\overline{\mathbf y}|)} 
\end{equation}
where 
\begin{equation}\label{epsilon-q}
\hspace{-0.21cm}\mathbb W_{\mathfrak{R}_\varepsilon (|\overline{\mathbf y}|)}=\left\{ \left. \left[\begin{array}{c}
\delta_1 \\
\vdots \\
\delta_q
\end{array}\right]  ~\right|\begin{array}{c} 
\vspace{0.06cm}
\mathcal R_{\varepsilon}(|\overline{y}_i|) \leq \delta_i \leq \frac{S_{\varepsilon_i}(0)}{ S_{\varepsilon}(0)} \mathcal R_{\varepsilon}(|\overline{y}_i|) \\
\hspace{-0.2cm}\textrm{for all}~\varepsilon_i \in (0,1) ~\textrm{that satisfy:}  \\
\hspace{-0cm}\varepsilon_1+\dots+\varepsilon_q = \varepsilon 
\end{array}  \right\}. 
\end{equation}
%
\end{thm}

\begin{figure}\center
\includegraphics[scale=0.8]{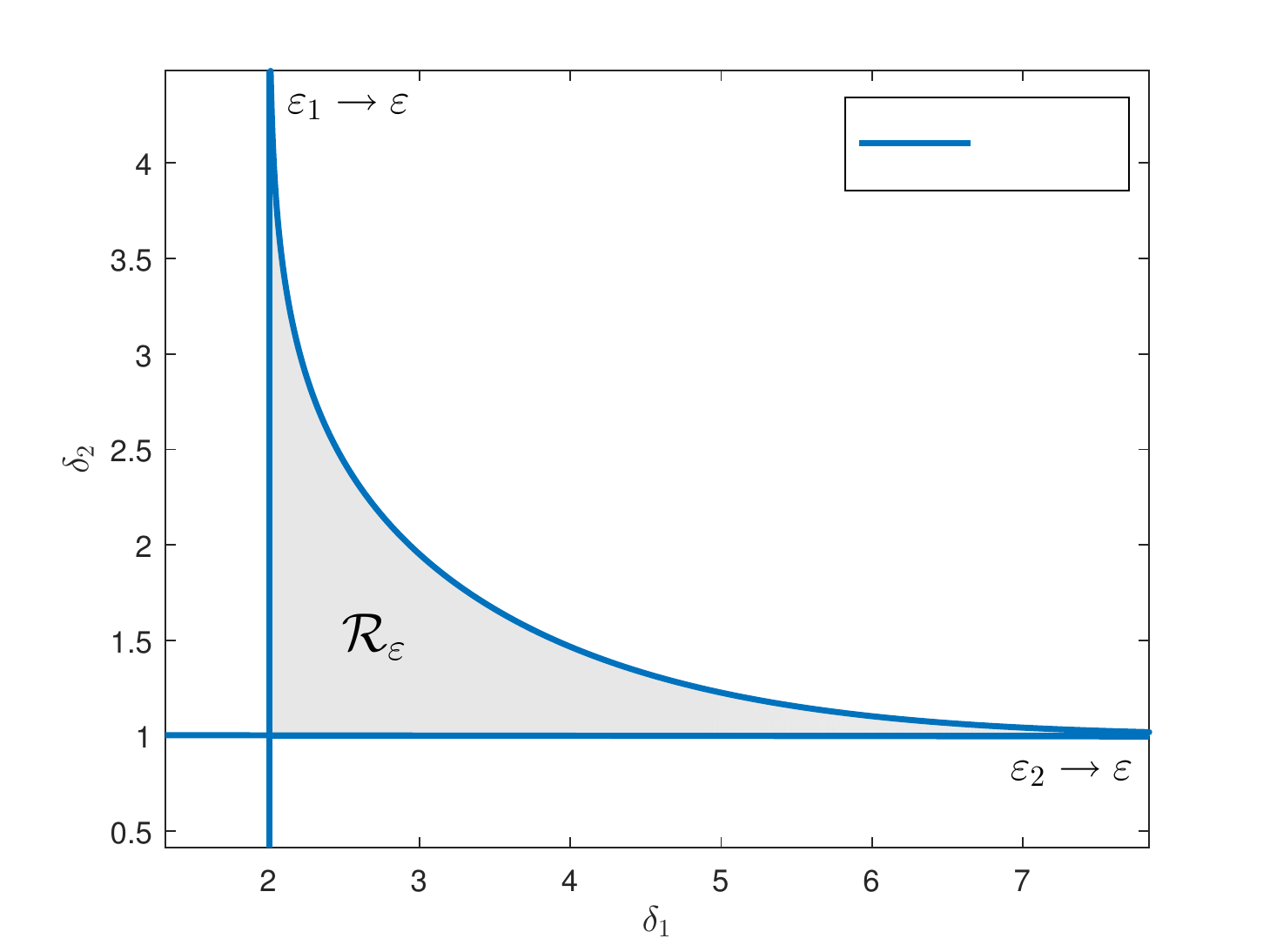}
\caption{The geometry of $\mathbb W_{\mathfrak{R}_{\varepsilon}(|\overline{\mathbf y}|)}$ is depicted in two dimensions. 
Since $\varepsilon_1+\varepsilon_2=\varepsilon$, if $\varepsilon_1 \rightarrow \varepsilon^{-}$, then 
$\varepsilon_2 \rightarrow 0^+$, and vice versa. This shows that $\mathbb W_{\mathfrak{R}_{\varepsilon}(|\overline{\mathbf y}|)}$ is bounded. 
}\label{fig: illustration}
\end{figure}

Figure \ref{fig: illustration} illustrates geometric shape of  $\mathbb W_{\mathfrak{R}_{\varepsilon}(|\overline{\mathbf y}|)}$ for $q=2$. The set of solutions that constitute $\boldsymbol{\mathcal R}_{\varepsilon}(|\overline{\mathbf y}|)$ lies within a bounded set that is characterized by $\mathfrak{R}_{\varepsilon}(|\overline{\mathbf y}|)$.

When all output observables show evolution of similar stochastic events, a single safety margin can be employed to assess systemic risk throughout the network. For instance, all observables are similar or homogenous if every $y_k$, for $k=1,\ldots,q$, measures either deviation of a single state from the average of all states, pairwise difference between two agents, or deviation of a single state from average of its immediate neighbors.  In this case,  systemic risk measure becomes scalar-valued according to 
\begin{equation}\label{eq: riskv}
\mathcal R_{\varepsilon}(|\overline{\mathbf y}|)=\inf\Big\{\delta \in \R_+~\Big|~\mathbb P\big(\Sp|\overline{\mathbf y}| \preceq \delta \mathbbm{1}_q \Sp \big)\geq 1-\varepsilon\Big\}.
\end{equation}
The next result is a direct consequence of Theorem \ref{thm: doubleinequality2}.
\begin{cor}
The scalar-valued joint systemic risk measure \eqref{eq: riskv} for dynamical network \eqref{first-order} with observables \eqref{observables} is bounded by 
\begin{equation}
\max_{1 \leq i \leq q} \Sp \mathcal{R}_{\varepsilon}(|\overline{y}_i|) ~\leq~ {\mathcal R}_{\varepsilon}(|\overline{\mathbf y}|)  ~\leq~  \Sp  \min_{1 \leq i \leq q} \Sp \upsilon_i \Sp \mathcal{R}_{\varepsilon}(|\overline{y}_i|) \label{ineq_vec_risk_1}
\end{equation}
where $\upsilon_i=S_{\varepsilon}(0)^{-1} S_{\varepsilon_i}(0)$ for all positive $\varepsilon_i$'s  that satisfy \eqref{epsilon-q}. 
\end{cor}

\begin{textblock*}{\textwidth}(11.3cm,-14.2cm) \small
$\partial \mathbb W_{\mathfrak{R}_{\varepsilon}}$
\end{textblock*}


\subsection{Risk of Large Fluctuations in Expectation}

Suppose that function $v: \R^q \rightarrow \R$ is  convex and monotonically increasing on the cone of positive orthant, i.e., if $\BB x \preceq \BB y$, then $v(\BB x) \leq u(\BB y)$. The joint risk in expectation  is an operator  $\boldsymbol{\mathcal T}_{\varepsilon}: \mathbb L^2(\mathbb R^q)\rightarrow \mathbb R^q$ that is defined by     
\begin{equation}\label{eq: riskmomv}
\boldsymbol{\mathcal T}_\varepsilon(|\overline{\mathbf y}|)=\inf\Big\{\boldsymbol \delta \in \mathbb R^q ~\Big|~ \mathbb{E}\big[u(|\overline{\mathbf y}|-\boldsymbol \delta) \big] \leq  v(\varepsilon)\Big\} ,
\end{equation} 
where  $\boldsymbol \delta=[\Sp \delta_1,\dots,\delta_q \Sp]^T$ is the vector of safety margins.

\begin{thm}\label{thm: multobjvolrisk}
Let us assume that $\mathbbm{1}_n \in \ker(C)$ for dynamical network \eqref{first-order} with observables \eqref{observables}. The constraint $\mathbb{E}\big[v(|\overline{\mathbf y}|-\boldsymbol \delta) \big] \leq  v(\varepsilon)$ in \eqref{eq: riskmomv} with $v(\BB z)=\BB z^T \BB z$, is feasible  if and only if $r > 0$ in which
$$r=\varepsilon^2-\bigg(1-\frac{2}{\pi}\bigg)\sum_{i=1}^q \Sp \overline{\sigma}_i^2.$$ 
Moreover, the joint risk in expectation assumes a set value that is given by 
\begin{equation*}
\boldsymbol{\mathcal T}_{\varepsilon}(|\overline{\mathbf y}|) =\left\{ \Sp \sqrt{\frac{2}{\pi}}\Sp \overline{\BB \sigma}+\mathbf z~\bigg|~\BB z\preceq \BB 0:~\|\mathbf z\|=\sqrt{r} \Sp \right\} \subset \R^q,
\end{equation*}
where $\overline{\BB \sigma}=[\Sp \overline{\sigma}_1,\ldots, \overline{\sigma}_q \Sp]^T$.
\end{thm}

\begin{thm}\label{thm: multobjvolrisk2}
Suppose that assumptions of Theorem \ref{thm: multobjvolrisk} hold. For all safety margins $\varepsilon_1,\ldots,\varepsilon_q >0$ that satisfy 
\[\varepsilon_1^2 + \ldots + \varepsilon_q^2 = \varepsilon^2 \] 
we have that
\[ \big[\mathcal{T}_{\varepsilon_1}(|\overline{y}_1|) \Sp, \Sp \ldots \Sp , \Sp \mathcal{T}_{\varepsilon_q}(|\overline{y}_q|)  \big]^T  ~\in~  \BB{\mathcal T}_{\varepsilon}(|\overline{\mathbf y}|). \]
\end{thm}

When all observed stochastic events are homogeneous, the joint risk in expectation \eqref{eq: riskmomv} reduces  to a scalar as follows   
\begin{equation}\label{eq: riskmomv2}
\mathcal T_\varepsilon(|\overline{\mathbf y}|)=\inf\Big\{\delta \in \mathbb R ~\Big|~ \mathbb{E}\big[v(|\overline{\mathbf y}|- \delta \mathbbm 1_q)\big] \leq v(\varepsilon)\Big\}.
\end{equation} 
The risk measure $\mathcal T_\varepsilon(|\overline{\mathbf y}|)$ is the smallest number $\delta$ that can be added to every element of  vector $|\overline{\mathbf y}|$ such that the cost of the resulting vector stays below the predetermined safety threshold.   


\begin{cor}Let us assume that $\mathbbm{1}_n \in \ker(C)$ for dynamical network \eqref{first-order} with observable \eqref{observables}. Let $\varepsilon^2>\sum_{j} \overline{\sigma}_j^2-\frac{2}{\pi}\big(\sum_j \overline{\sigma}_j \big)^2$, for $\overline{\sigma}_j$ the standard deviation of $\overline{y}_j$. Then, the risk in expectation  as in \eqref{eq: riskmomv2} with quadratic utility function satisfies
\begin{equation*}
 \mathcal T_{\varepsilon}\big(|\overline{\mathbf y}|\big)=\sqrt{\frac{2}{\pi}}\sum_{j=1}^q \overline{\sigma}_{j}-\sqrt{\varepsilon^2-\sum_{j=1}^{q} \overline{\sigma}^2_j+\frac{2}{\pi}\bigg(\sum_{j=1}^{q} \overline{\sigma}_j \bigg)^2}.
 \end{equation*}  

\end{cor}


For vector-valued risk measure in expectation w.r.t.  exponential utility function, we may use utility function 
\[v(\BB z)=e^{\beta (z_1+\ldots+z_q)}.\]
The definition of risk measure in this case is given by 
\begin{equation}\label{eq: riskvectorexponential}
\mathcal T_\varepsilon(|\mathbf{\overline y}|)=\inf\Big\{ \boldsymbol \delta\in \mathbb R^q ~\Big|~  \mathbb E \big[e^{\beta \sum_{i}(|\overline{y}_i|-\delta_i)} \big] \leq e^{\beta \varepsilon}\Big\}.
\end{equation}  Following a simple continuity argument, it can be shown that the vector of risk  $\mathcal T_\varepsilon(|\mathbf{\overline y}|)=\boldsymbol \delta=(\delta_1,\dots,\delta_q)$ satisfies
\begin{equation}\label{eq: riskvectorexponential2}
\sum_{i=1}^q \delta_i=\frac{1}{\beta}\ln\left( \mathbb E\Big[e^{\beta \big(\sum_{i}(|\overline{y}_i|-\frac{1}{q}\varepsilon \big)}\Big]\right)
\end{equation} 
The calculation of the right hand side of \eqref{eq: riskvectorexponential2} becomes significantly harder as it involves  calculation of functions of joint random variables, which is beyond the scope of this paper. 

\begin{rem} The vector-valued risk metrics, enjoy similar functional properties with the scalar risk metrics of \S \ref{sec: riskmeasures}. Indeed, Propositions \ref{prop: cohrisk} and \ref{prop: riskmom} can be appropriately modified and extended to vector-valued risk, along the lines of \cite{Jouini2004}. The details are omitted due to space limitation.
\end{rem}

\section{Fundamental Limits and Tradeoffs}\label{sec:tradeoffs}
The presence of time-delay induces fundamental limits and tradeoffs on the least achievable risk values for the class of dynamical network \eqref{first-order} in steady-state. We characterize some Heisenberg-like inequalities and reveal inherent interplay between the systemic risk measure and a measure of network connectivity. 


The first limit occurs due to stability condition of Assumption \ref{assum1} that can be quantified as  the following lower bound on the total effective resistance of the coupling graph of the network 
\begin{equation}\label{eff-limit}
 \ER > \frac{2 n(n-1)\tau}{\pi}.
\end{equation}
This implies that  in  time-delay network \eqref{first-order} the feedback structure (represented by the weighted graph $\G$) may not have an overall strong couplings beyond the above limit, e.g., feedback loops may not have very large gains and/or dense topology especially when all gains are equal (which corresponds to unweighted $\G$).

\begin{figure}\center
\includegraphics[scale=0.8]{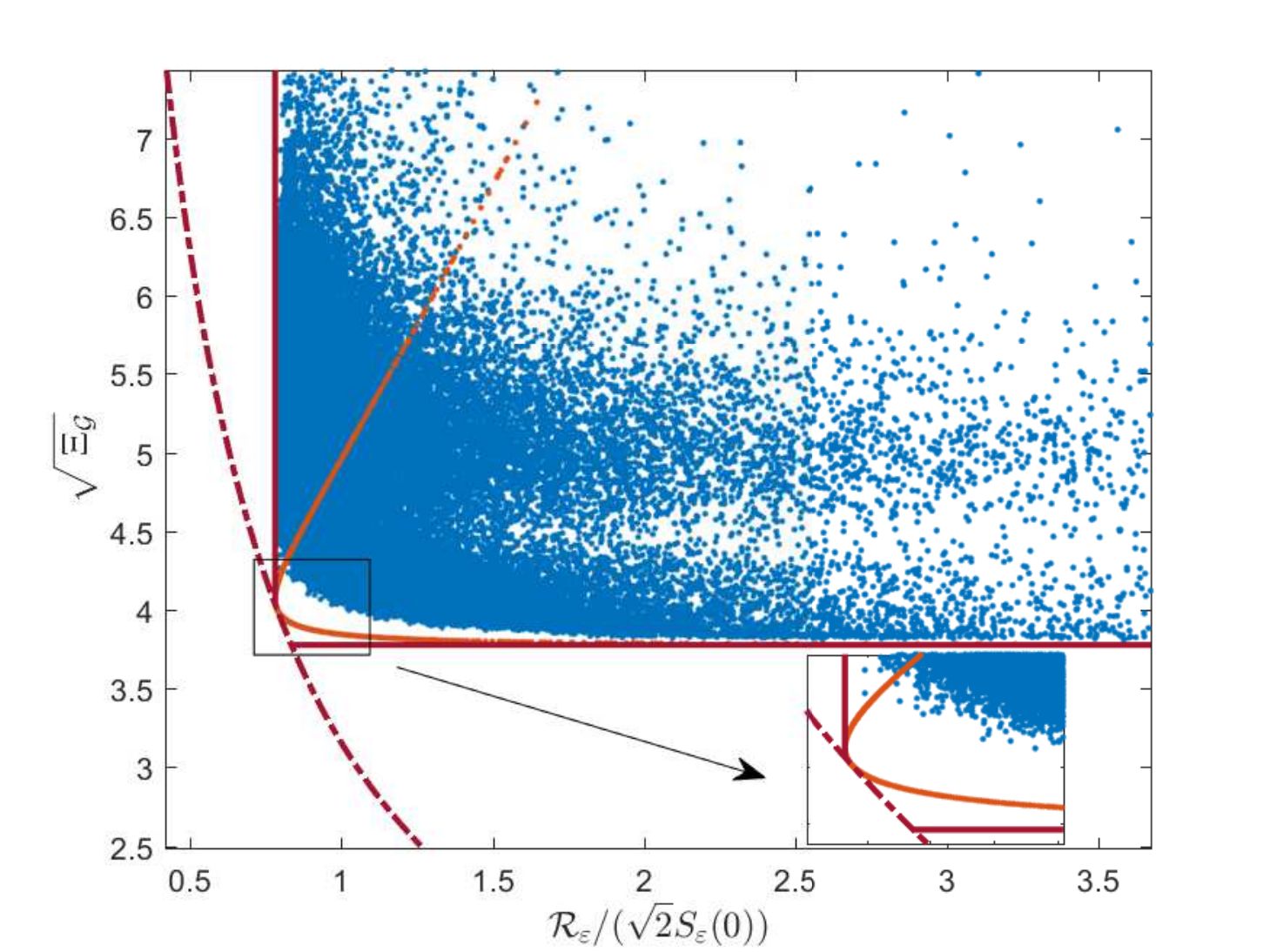}
\caption{Fundamental and trade-off bound curves for the systemic risk. Blue dots are randomly generated graphs over $n=8$ agents, for fixed $\tau=0.4$ and output vector $\mathbf c$. The red vertical and horizontal lines are the hard limits, and the dashed line is the trade-off principle. The orange curve illustrates a particular family of graphs that achieve the trade-off curve, asserting that all three bounds are sharp.}\label{fig: trade-off}
\end{figure}

\begin{thm}\label{thm: tradeoff}
Let observable of dynamical network \eqref{first-order} be scalar as in \eqref{scalar-y}, i.e., $q=1$. Then, there is a hard limit on the least achievable value for  risk of large fluctuations  as follows 
\begin{equation}\label{limit-risk}
\mathcal R_{\varepsilon}(|\overline{y}|) \Sp \geq \Sp   \kappa_* \Sp \sqrt{\tau} 
\end{equation}
where
\begin{equation*}\label{eq: fl}
\kappa_*= \|\mathbf c\| \Sp \frac{ b \Sp S_{\varepsilon}(0)} {\sqrt{1-\sin(z^+)}}
\end{equation*} 
and $z^+$ is the positive solution of equation $\cos(z)=z$.  Furthermore, the least achievable value for risk of high volatility cannot be improved beyond some lower bound that is characterized by:
\begin{itemize}
\item[(a)] for quadratic utility function: \begin{equation*}
\mathcal T_\varepsilon (|\overline y|) \geq \sqrt{\frac{2}{\pi}}\sigma_*-\sqrt{\varepsilon^2-\bigg(1-\frac{2}{\pi}\bigg)\sigma_*^2}
\end{equation*} 
\item[(b)] for exponential utility function: \begin{equation*}\begin{split}
\mathcal T_{\varepsilon}(|\overline y|)\geq \frac{\beta \sigma_*^2}{2}&+\frac{\ln\big(1+\mathrm{erf}(\frac{\beta\sigma_*}{2})\big)}{\beta}-\varepsilon
\end{split}
\end{equation*}  
\end{itemize}
where
\begin{equation}\label{eq: fl}
\begin{split} 
\sigma_*=b \Sp \|\mathbf c\| \Sp \sqrt{\frac{\tau }{2\big(1-\sin(z^+)\big)}}.
\end{split}
\end{equation} 
\end{thm}



The lower limit can be achieved by graphs with Laplacian spectrum that satisfy $\lambda_i\tau=z^+$ for all $i>1$, i.e., a complete graph over $n$  nodes with coupling weights $k_{ij}\approx \frac{0.7391}{n\tau}$. Inequalities \eqref{eff-limit} and \eqref{limit-risk}  represent {\it hard limits} as their lower bounds are independent of the coupling structure  of the network; by fixing $\varepsilon$ they only depend on the size of the network and time-delay.  When there is no time-delay, i.e., $\tau=0$, hard limits disappear. Because in no time delay case, one may simultaneously improve risk and connectivity by arbitrarily increasing the feedback gains. This shows that time-delay is the main reason  for emergence of these fundamental limits on risk and connectivity. 

\begin{thm}\label{thm: tradeoff2}
Let observable of dynamical network \eqref{first-order} be scalar, i.e., $q=1$. There is an inherent fundamental tradeoff between least achievable values for risk of large fluctuations and network interconnectivity that can be quantified as follows
\begin{equation}\label{tradeoff}
\mathcal R_{\varepsilon}(|\overline{y}|) \Sp \cdot \Sp  \sqrt{\ER} \Sp >  \Sp \vartheta_* \Sp \sqrt{n\tau}
\end{equation} 
where $\ER$ is the total effective resistance of the coupling graph of the network and 
\begin{equation*}\label{eq: fll}
\vartheta_* \Sp = \Sp \|\mathbf c\| \Sp b  \Sp S_{\varepsilon}(0) \Sp \sqrt{p^\dag} 
\end{equation*} 
with $p^\dag=\min\{p_2^*, \ldots, p_n^*\}$ in which positive numbers $p_k^*$ is the minimum of $$p_k(x)=\bigg[(k-1)+\frac{2(n-k)}{\pi} x\bigg]\frac{\cos(x)}{x^2(1-\sin(x))}$$
over $x\in (0,\pi/2)$ for all $k=2,\dots,n$.  Moreover, fundamental tradeoffs emerge between the least achievable values for volatility risk and network connectivity in the following sense 
\begin{equation*}
\big(\mathcal T_{\varepsilon}(|\overline{y}|)+\varepsilon\big) \sqrt{\Xi_{\mathcal G}} \Sp > \Sp \Delta(\tau)
\end{equation*} 
in which 
\begin{equation*}
\Delta(\tau)=\bigg(\sqrt{\frac{2}{\pi}}+\frac{1}{2\varepsilon}\bigg(1-\frac{2}{\pi}\bigg)\sigma_*\bigg)\varrho_*
\end{equation*} 
for the quadratic utility function and 
\begin{equation*}
\Delta(\tau) = \frac{\beta}{2}\sigma_*\varrho_*+ \frac{\mathrm{erf}(\beta\sigma_*/2)}{\beta} \Sp \frac{n(n-1)\tau}{\pi}
\end{equation*}  
for the exponential utility function, where 
\begin{equation}\label{eq: fll}
\varrho_* \Sp = \Sp \|\mathbf c\| \Sp | b | \Sp \sqrt{\frac{n\tau}{2} p^\dag}.
\end{equation}  
\end{thm}

\begin{table*}[t]
\center
 \begin{tabular}{||c || l ||} 
 \hline
 Graph &  Risk of Large Fluctuations in Probability for Scalar Events   \\ [0.5ex] 
 \hline\hline
 $\mathbf K_n$ & $ \mathcal R_{\varepsilon}(|\overline {y}_i|)= \sqrt{2}S_\varepsilon(0)\sqrt{\tau\left(1-\frac{1}{n}\right) f(n\tau)} ~~~~\textrm{for}~~ i=1,\dots,n$   \\ 
 \hline
 $\mathbf{W}_{n+1}$ & $ \mathcal R_{\varepsilon}(|\overline {y}|_i)= \sqrt{2}S_\varepsilon(0) \times \begin{cases}  \sqrt{n\tau f\big((n+1)\tau\big)} & ~~\textrm{if}~~i=1 \\
\sqrt{\frac{\tau}{n}~ \displaystyle \sum_{k=2}^{n+1} f\left(3-2\cos \left(\frac{2\pi(k-1)\tau}{n}\right)\right)} &~~\textrm{if}~~ i=2,\dots,n 
\end{cases}
$ \\
 \hline
 $\mathbf K_{n_1,n_2}$ & $ \mathcal R_{\varepsilon}(|\overline {y}|_i)= \sqrt{2}S_\varepsilon(0) \times \begin{cases}  \sqrt{\big(1-\frac{1}{n_2}\big)\tau f(n_2\tau)+\frac{n_2}{n_1 n}\tau f(n\tau)} & ~~\textrm{if}~~ i\in \G_1 \\
\sqrt{\big(1-\frac{1}{n_1}\big)\tau f(n_1\tau)+\frac{n_1}{n_2 n}\tau f(n\tau)} & ~~\textrm{if}~~ i\in \G_2 
\end{cases}
$  \\
 \hline
 $\mathbf{P}_n$ & $ \mathcal R_{\varepsilon}(|\overline {y}|_i)= \sqrt{2}S_\varepsilon(0)\displaystyle  \sqrt{ \frac{2\tau}{n}\sum_{k=2}^{n}\cos^2\left( \frac{\pi (n-k+1)(2i-1)}{2n}\right)
 f\left(2\left[1-\cos \left(\frac{\pi (k-1)}{n}\right)\right]\tau\right)},~ i=1,\dots,n  $   \\
 \hline
 $\mathbf{R}_n$ & $ \mathcal R_{\varepsilon}(|\overline {y}|_i)= \sqrt{2}S_\varepsilon(0)\sqrt{\tau\displaystyle  \sum_{k=2}^{n} f\left(2\left[1-\cos \left(\frac{2\pi(k-1)}{n}\right)\right]   \tau\right) },~ i=1,\dots,n  $  \\ [1ex] 
 \hline
\end{tabular}\caption{Explicit formulas for systemic risk measure of individual events. The function $f(x)$ is defined in \eqref{eq: functionf}.  For a detailed analysis on the derivation of the formulas we refer to Appendix C.  }\label{table: flucrisktopgraph}
\vspace{-0.5cm}
\end{table*}

Inequality \eqref{tradeoff} implies that network \eqref{first-order} will be prone to higher levels of risk if the connectivity enhances, e.g., by adding new coupling links and/or increasing the  feedback gains. The origin of this intrinsic tradeoff is time-delay; inequality \eqref{tradeoff} would be trivial  if $\tau=0$. We observe that by increasing $\tau$, while preserving stability, interplay between systemic risk and network connectivity becomes more apparent. The key observation is that by fixing the network size, time delay and $\varepsilon$ inequalities \eqref{limit-risk} and \eqref{tradeoff} remain true for all dynamical networks \eqref{first-order} with arbitrary graph topologies and output matrices $c$; in other words, these fundamental limits and tradeoffs are {\it universal}. In Figure \ref{fig: trade-off}, we verify tightness of our bounds using extensive simulation results. The blue dots are numerically calculated based on the following procedure: (i)  randomly generate a connected graph that satisfies stability condition  in Assumption \ref{assum1} , (ii) randomly generate an output vector $c$, and (iii) compute the scaled pair $(\sqrt{\Xi_{\mathcal G}}, \mathcal R_{\varepsilon}(\bar y))/(\sqrt{2}S_{\varepsilon}(0))$. The  orange color curve in Figure \ref{fig: trade-off} illustrates a particular family of complete graphs with parametrized identical coupling strengths that is specifically constructed to  highlight  sharpness of our  bounds in \eqref{limit-risk} and \eqref{tradeoff}. 


\begin{thm}\label{thm: tradeoffflucriskmv}
Let us consider dynamical network \eqref{first-order} with vector of observables \eqref{observables}.  There is an intrinsic  fundamental tradeoff between the least achievable values for the joint vector-valued risk measure on $\R^q_+$ and network interconnectivity that can be characterized by inequality 
\begin{equation}\label{tradeoff-2}
\boldsymbol{\mathcal R}_{\varepsilon}(|\overline{\mathbf y}|) \Sp \cdot \Sp  \sqrt{\ER} \Sp  \succ  \Sp \BB \Theta_* \Sp \sqrt{n\tau}
\end{equation} 
in which $\BB \Theta_* =  b \Sp S_{\varepsilon}(0) \Sp \sqrt{p^\dag}~ \big[ \Sp \|\mathbf c_1\|, \Sp \ldots \Sp, \|\mathbf c_q\| \Sp \big]^T.$
\end{thm}

Proof of  this theorem is based on results of Theorems \ref{thm: doubleinequality2} and \ref{thm: tradeoff2} and is omitted.

\begin{cor}\label{cor: multilimitstradeoffs}
Suppose that $\mathbbm{1}_n \in \ker(C)$ for dynamical network \eqref{first-order} with output matrix $C=\big[\Sp \mathbf c_1, \ldots, \mathbf c_q \Sp \big]^T$. The vector-valued joint risk in expectation \eqref{eq: riskmomv} cannot be improved beyond some hard limits that are characterized as lower bounds in the following inequalities:

\noindent $(i)$ For quadratic utility function  $v(\BB z)=\BB z^T\BB z$, it follows that 
\begin{equation*}
 \boldsymbol{\mathcal T}_{\varepsilon}(|\mathbf{\overline{y}}|) ~\succeq~ \sqrt{\frac{2}{\pi}}\Sp \boldsymbol{\overline{\sigma}}^* ~-~\mathbbm 1_q ~\sqrt{\varepsilon^2-\bigg(1-\frac{2}{\pi}\bigg) \Sp \|\boldsymbol{ \overline{\sigma}}^* \|^2}
\end{equation*}
where  $\boldsymbol{ \overline{\sigma}}^*=[\Sp \sigma_1^*, \ldots, \sigma_q^*]^T$ with
$$
\sigma_i^*=\|\mathbf c_i\| \Sp b \Sp  \sqrt{\frac{\tau }{2\big(1-\sin(z^+)\big)}}, 
$$ and $z^+$ is the positive solution of equation $\cos(z)=z$.

\noindent $(ii)$ For exponential utility function $v(\BB z)=e^{\beta (z_1 + \ldots + z_q)}$, it holds that
\begin{equation*}
\mathbbm 1_q^T \boldsymbol{\mathcal T}_{\varepsilon}(|\mathbf{\overline{y}}|) ~\geq~   b ~\sqrt{\frac{\tau}{\pi\big(1-\sin(z^+)\big)}}~\sum_{i=1}^q\|\mathbf c_i\|-\varepsilon.
\end{equation*}
\end{cor}


\section{Exploiting Graph Topology} 

We show that for some graph topologies, the value of systemic risk measure can be calculated explicitly. Our focus will be only on risk in probability with respect to deviation of a single state from the average, where similar results can be obtained for other types of observations as well as risk in expectation. In order to exploit graph structure and determine how risk scales with network size, we consider the class of dynamical network \eqref{first-order} with unweighted coupling graphs. The results of this section are summarized in Table \ref{table: flucrisktopgraph}.

\vspace{0.1in}

\noindent {\it Complete graph.} Due to the perfect symmetry of a complete graph $\mathbf K_n$, all elements of the vector of joint fluctuation risk measure are identical and the risk measure reduces to \eqref{eq: riskv}. The network is marginally stable if $\tau < \frac{\pi}{2n}$, where this region shrinks as $n$ grows.  For two given linear consensus networks \eqref{first-order} with coupling graphs $\mathbf{K}_{n_1}$ and $\mathbf{K}_{n_2}$ and $n_1 < n_2$, there exists a critical time delay $\tau^* < \frac{\pi}{2 n_2}$ such that $\mathcal R^{(2)}_{\varepsilon}(|\overline{\BB y}|) \leq \mathcal R_{\varepsilon}^{(1)}(|\overline {\BB y}|)$ for all $\tau \leq \tau^*$, and $\mathcal R^{(1)}_{\varepsilon}(|\overline{\BB y}|) < \mathcal R_{\varepsilon}^{(2)}(|\overline {\BB y}|)$ for all $\tau^* < \tau < \frac{\pi}{2 n_2}$. The phenomenon is graphically illustrated on the left plot of Figure \ref{fig: riskcomplete}. In a more simple language, networks with larger complete graphs are safer than smaller ones, and as time delay increases there is a transition point beyond which networks with smaller complete graphs are safer.

\vspace{0.1in}

\noindent {\it Wheel graph.} Let us label the central node of the wheel graph $\mathbf{W}_{n+1}$ by $1$ and the remaining nodes on the circumference by $2,\ldots,n+1$. The largest Laplacian eigenvalue is $\lambda_{n+1}=n+1$, which implies that network is marginally stable if $\tau  < \frac{\pi}{2(n+1)}$.  The symmetric structure of  $\mathbf{W}_{n+1}$ implies that $\mathcal R_{\varepsilon}(|\overline{y}_2|)=\ldots=\mathcal R_{\varepsilon}(|\overline{y}_{n+1}|)$. According to explicit expressions for risk measures in Table \ref{table: flucrisktopgraph}, we can verify that\footnote{It can be shown that  
\begin{equation*}\begin{split}
\frac{\mathcal{R}_{\varepsilon}(|\overline y|_i)}{\mathcal{R}_{\varepsilon}(|\overline y|_1)}&<\frac{1}{n^2}+\frac{n^2-1}{n^2} \frac{\cos(\lambda^*\tau)}{1-\sin(\lambda^*\tau)}\frac{1-\sin((n+1)\tau)}{\cos((n+1)\tau)} \leq 1
\end{split}
\end{equation*} for all  $i=2,\ldots,n+1$, where  $\lambda^*$ is the maximizer of $\frac{\cos(\lambda\tau)}{1-\sin(\lambda\tau)}$ over $ \lambda \in \big\{3-2\cos(2\pi(k-1)/n)~|~k=2,\ldots,n \big\}$.} $\mathcal{R}_{\varepsilon}(|\overline y|_1) < \mathcal{R}_{\varepsilon}(|\overline y|_i)$ for all $i=2,\ldots,n+1$. Thus, the central node in a dynamical network with  wheel graph is riskier than the surrounding nodes.

\begin{figure}
\includegraphics[scale=0.55]{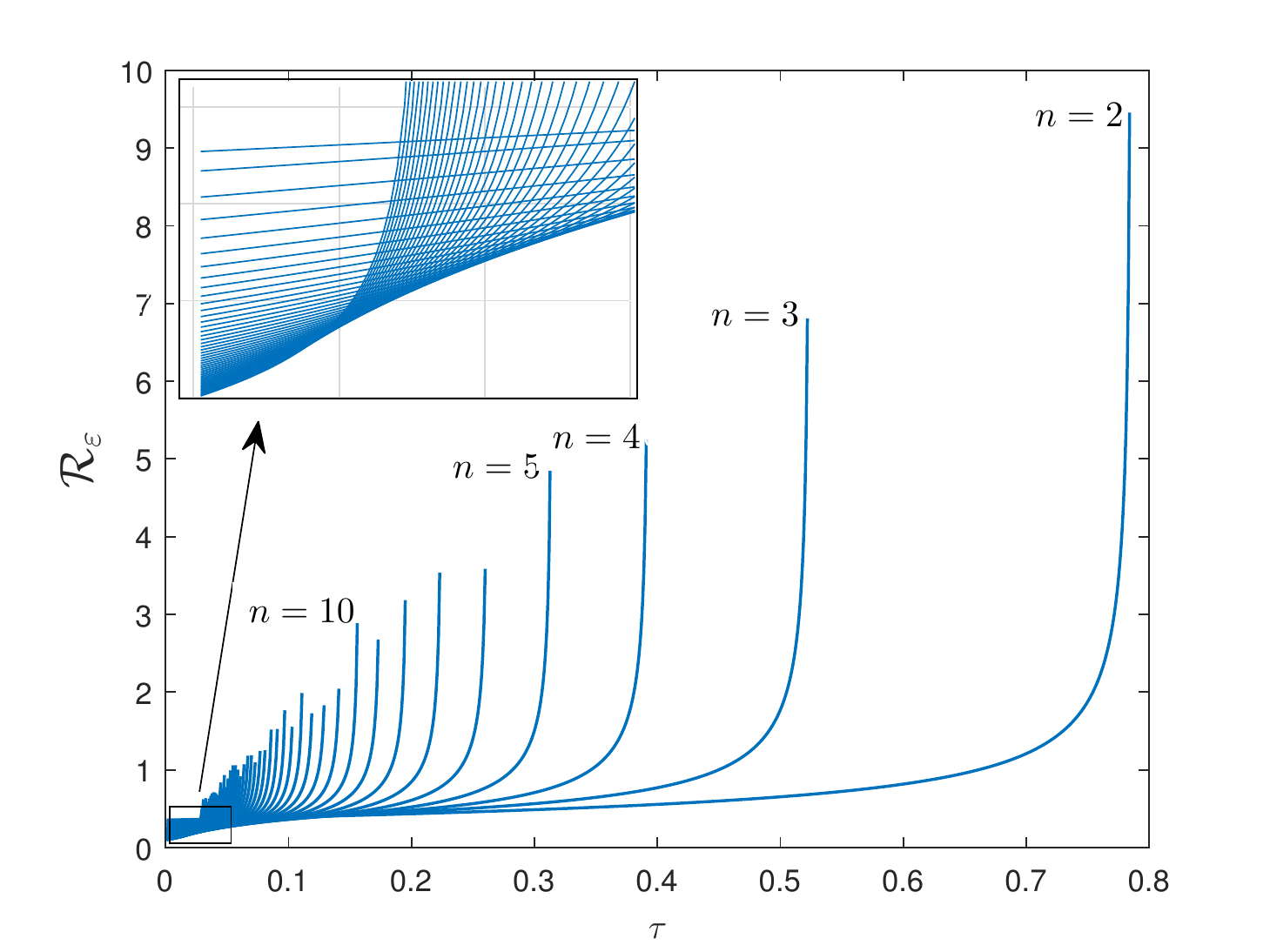}\hspace{-0.2in}
\includegraphics[scale=0.55]{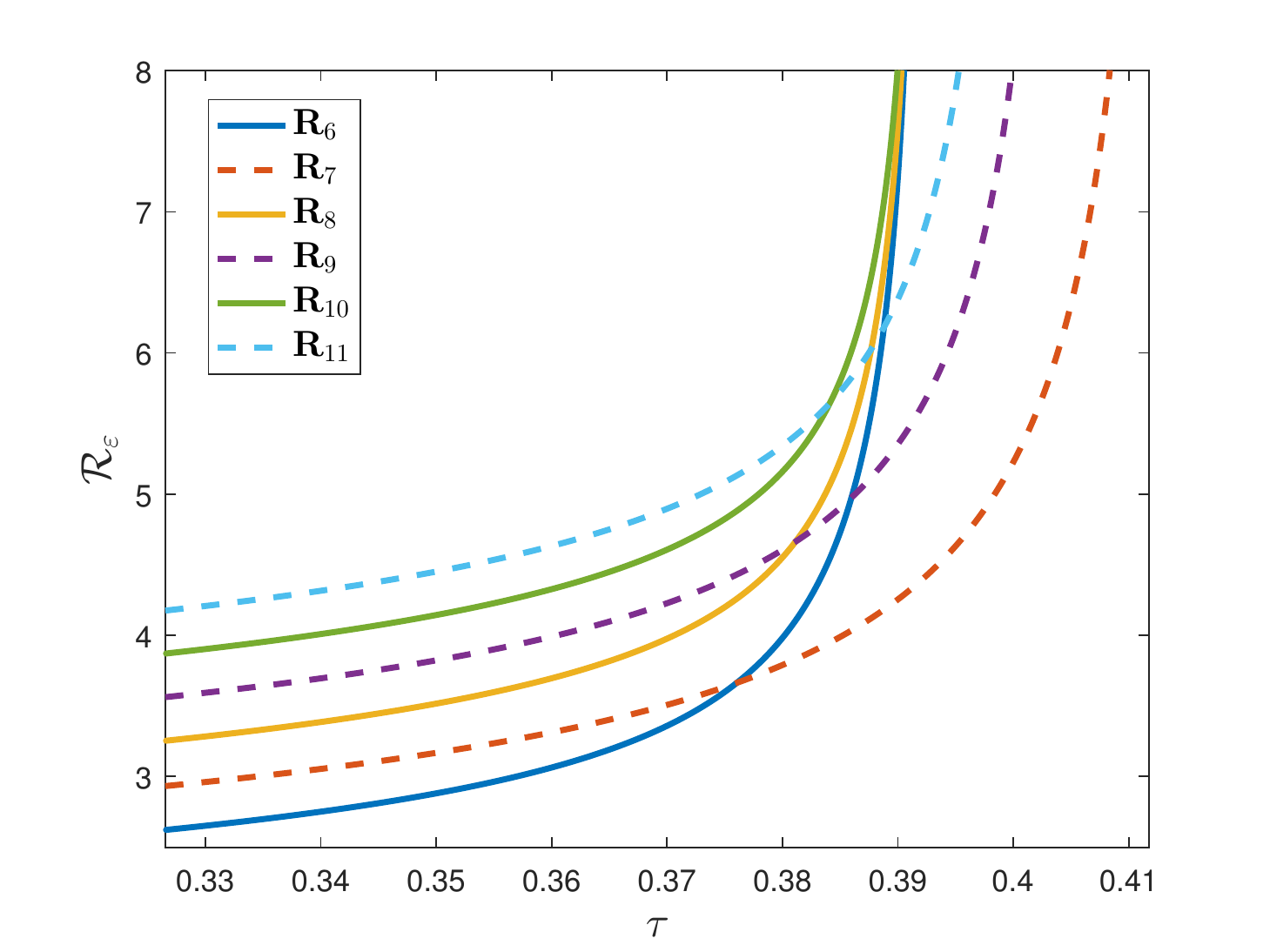}
\caption{(Left) The fluctuation risk of a complete graph as a function of the delay $\tau$. Here $\varepsilon=0.01$ and $\mathbf K_n$, $n=2,\dots,50$. (Right) Fluctuation risk on ring graphs in terms of network
size. A remarkable qualitative distinction appears between
graphs with odd and graphs with even number of agents. } \label{fig: riskcomplete}
\vspace{-0.5cm}
\end{figure}

\vspace{0.1in}

\noindent {\it Complete Bipartite graph.} Let us denote each side of the complete bipartite graph $\mathbf K_{n_1,n_2}$ by $\G_1$ and $\G_2$ that each has $n_1$ and $n_2$ nodes, respectively. Nodes at each side are only connected to nodes on the other side. The total number of nodes is $n=n_1+n_2$ and let us assume that $n_1\leq n_2$.  In this case, the network \eqref{first-order} is marginally stable if and only if $\tau<\frac{\pi}{2n}$.  The symmetric topology of $\mathbf{K}_{n_1,n_2}$  entails that all $\mathcal R_{\varepsilon}(|\overline{y}|_i)$ belonging to one side (either $\G_1$ or $\G_2$) are identical. We may categorize the set of complete bipartite graphs into two subclasses: graphs with $n_1=1$ and the ones with $n_1 >1$. The first class contains all star graphs $\mathbf K_{1,n-1}$, where $\G_1$ consists of only node with label $1$ and $\G_2$ contains all nodes with labels $2,\ldots,n$. From Table \ref{table: flucrisktopgraph}, it is straightforward to verify that \[\mathcal R_{\varepsilon}(|\overline{y}|_1)>\mathcal R_{\varepsilon}(|\overline{y}|_j)\]
for all nodes $j \in \{2,\ldots,n\}$ in subgraph $\G_2$.  Similar to the wheel graph, the central node in a star graph attains the highest risk value among all nodes in the network, regardless of the time delay value. When $n_1 > 1$, from Table \ref{table: flucrisktopgraph} and some elementary analysis, it can be shown that there exists a unique $\tau^*=\tau^*(n_1,n_2)$ such that for all $i$ in subgraph $\G_1$ and all $j$ in subgraph $\G_2$ the following inequalities hold 
\begin{equation*}\begin{split}
\mathcal R_{\varepsilon}(|\overline{y}|_i) &\leq \mathcal R_{\varepsilon}(|\overline{y}|_j)~~~\textrm{for}~~~0 \leq \tau \leq \tau^*\\
\mathcal R_{\varepsilon}(|\overline{y}|_j) & < \mathcal R_{\varepsilon}(|\overline{y}|_i)~~~\textrm{for}~~~\tau^* < \tau < \frac{\pi}{2n}.
\end{split}
\end{equation*} 
For $n_1>1$, we conclude that, unlike graphs with star topology, time delay plays an active role in determining which of the subgraphs $\G_1$ or $\G_2$ contains riskier nodes.  This point is illustrated with two examples in Figure \ref{fig: bipartite}.

\begin{figure}\center
\includegraphics[scale=0.5]{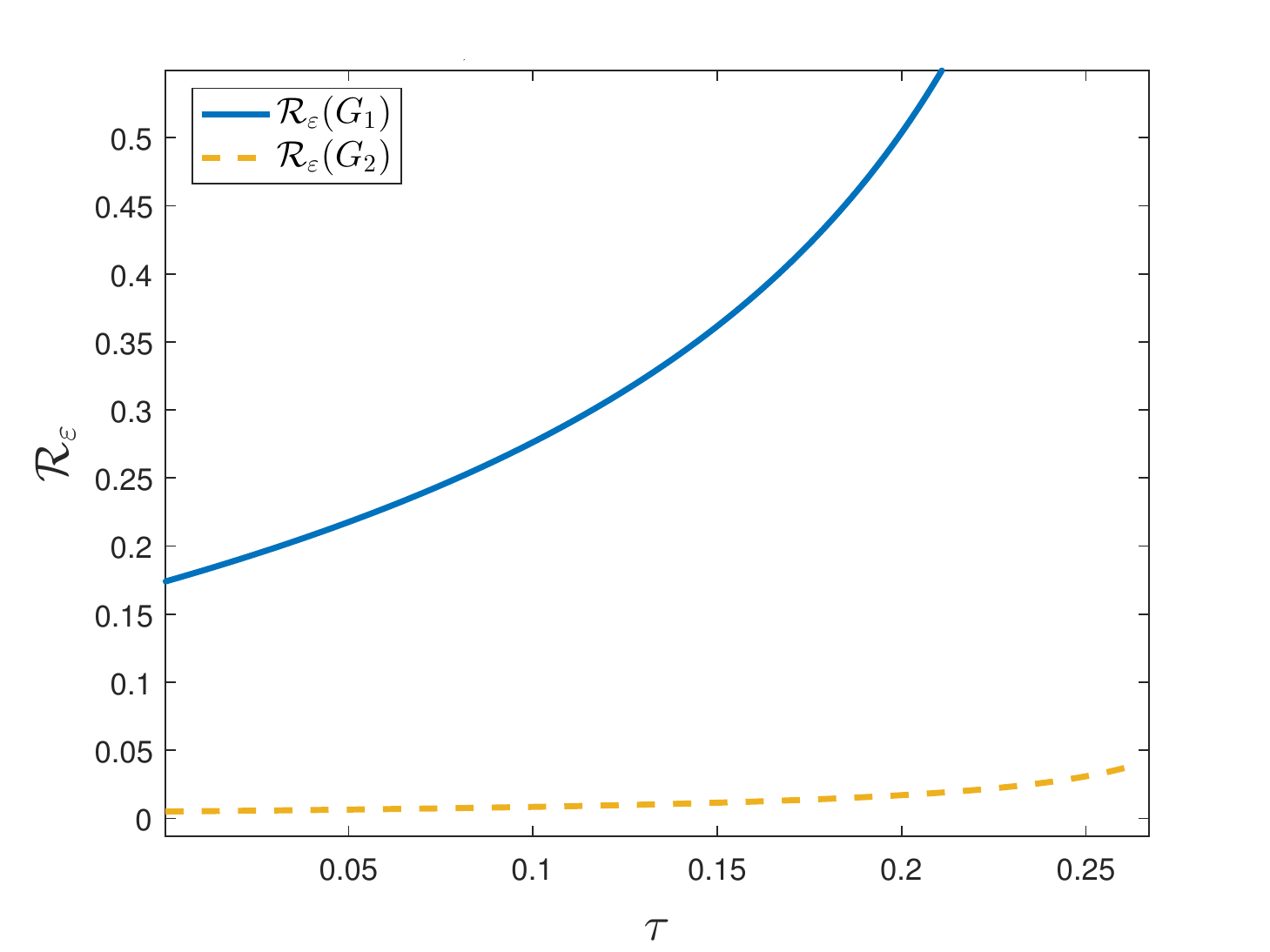}
\includegraphics[scale=0.5]{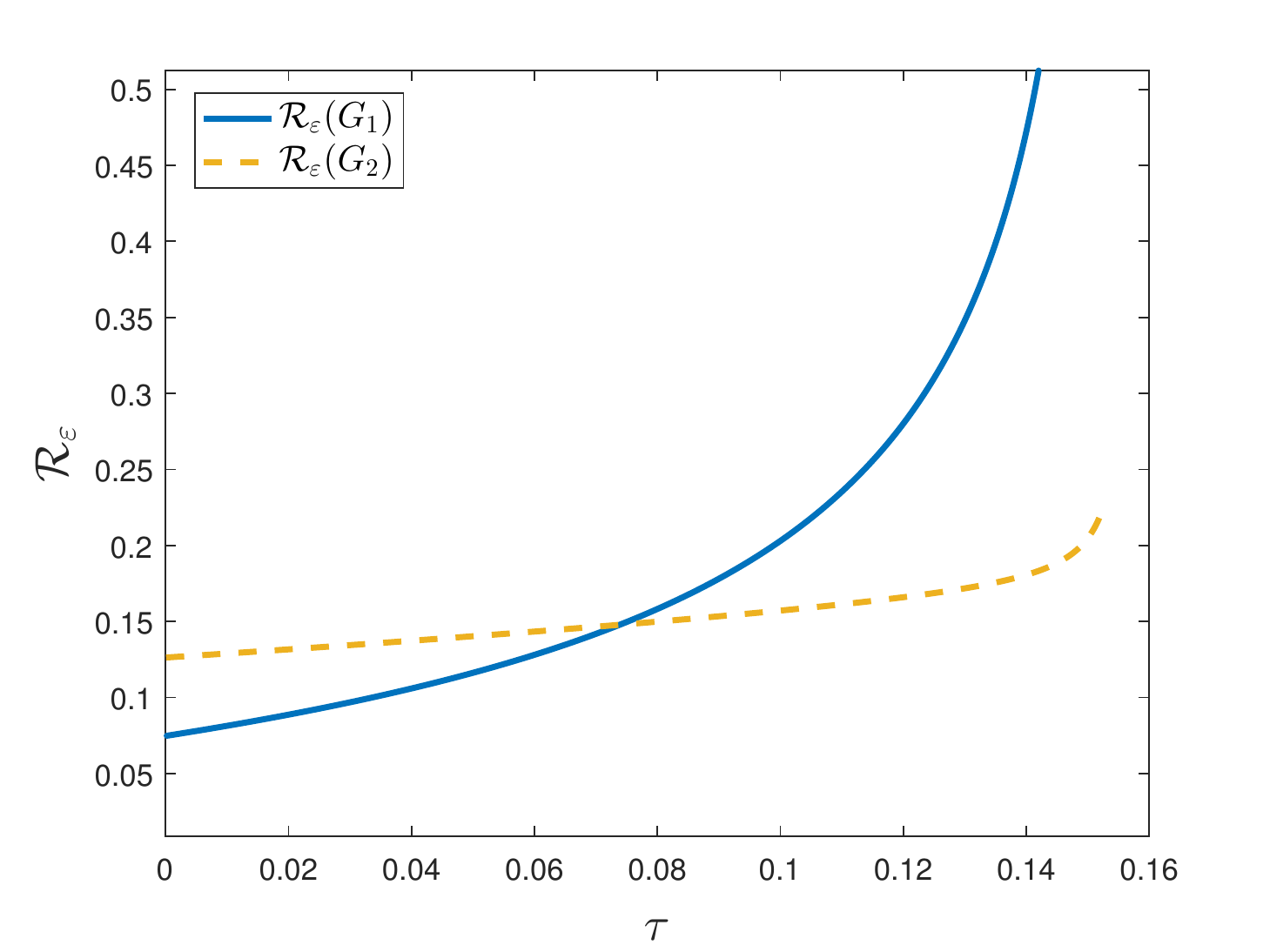}
        \caption{ Illustration of the systemic risk on $\mathbf{K}_{1,4}$ and $\mathbf{K}_{2,8}$. The first example is the star graph. The central node is the riskier one. The second scheme has a group of two and eight nodes, respectively. In the second scheme, we observe that the larger group is safer for small delay and riskier for large delay.}\label{fig: bipartite}
\end{figure}

\begin{figure}\center
 \includegraphics[scale=0.5]{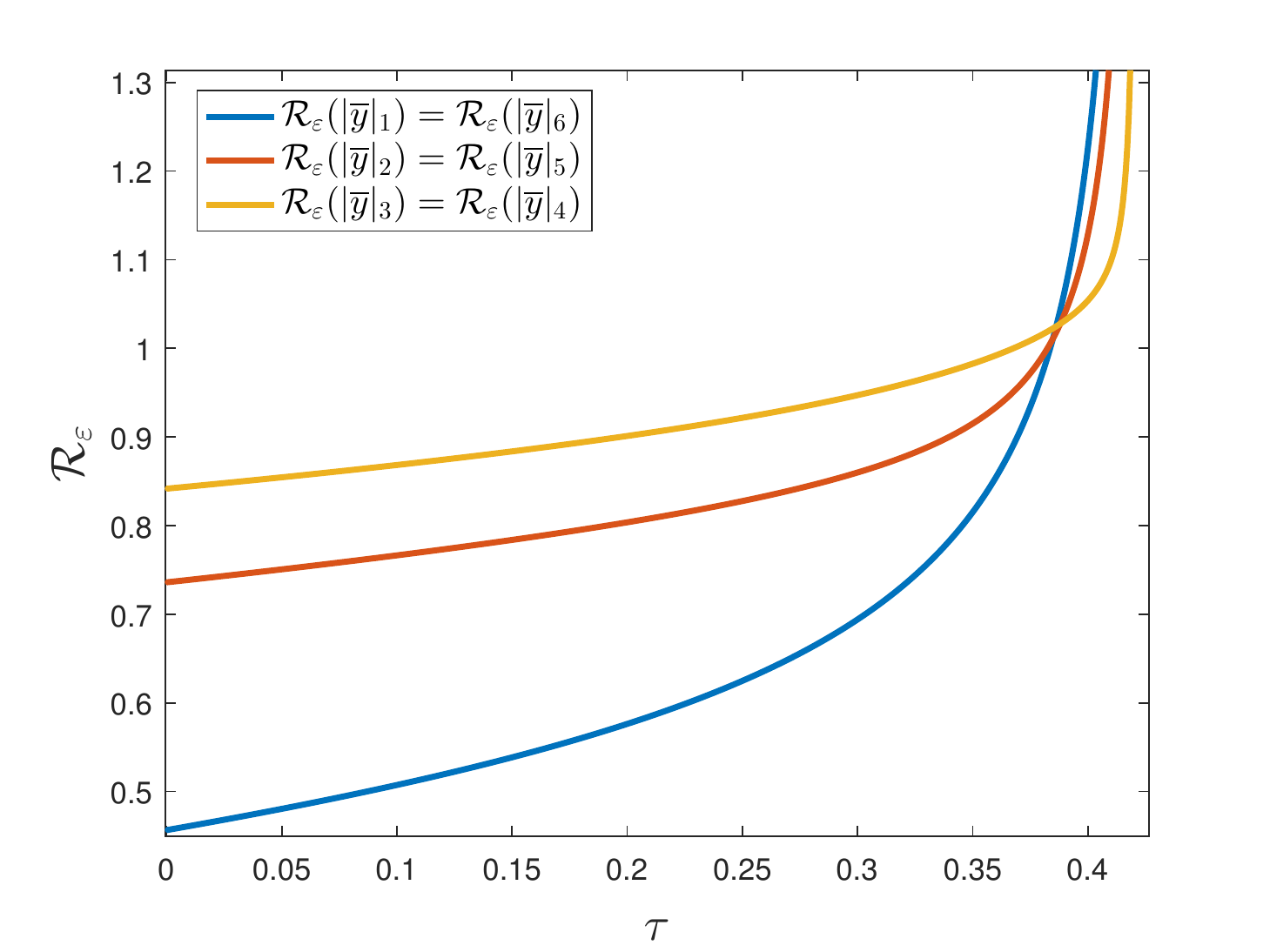}
\includegraphics[scale=0.5]{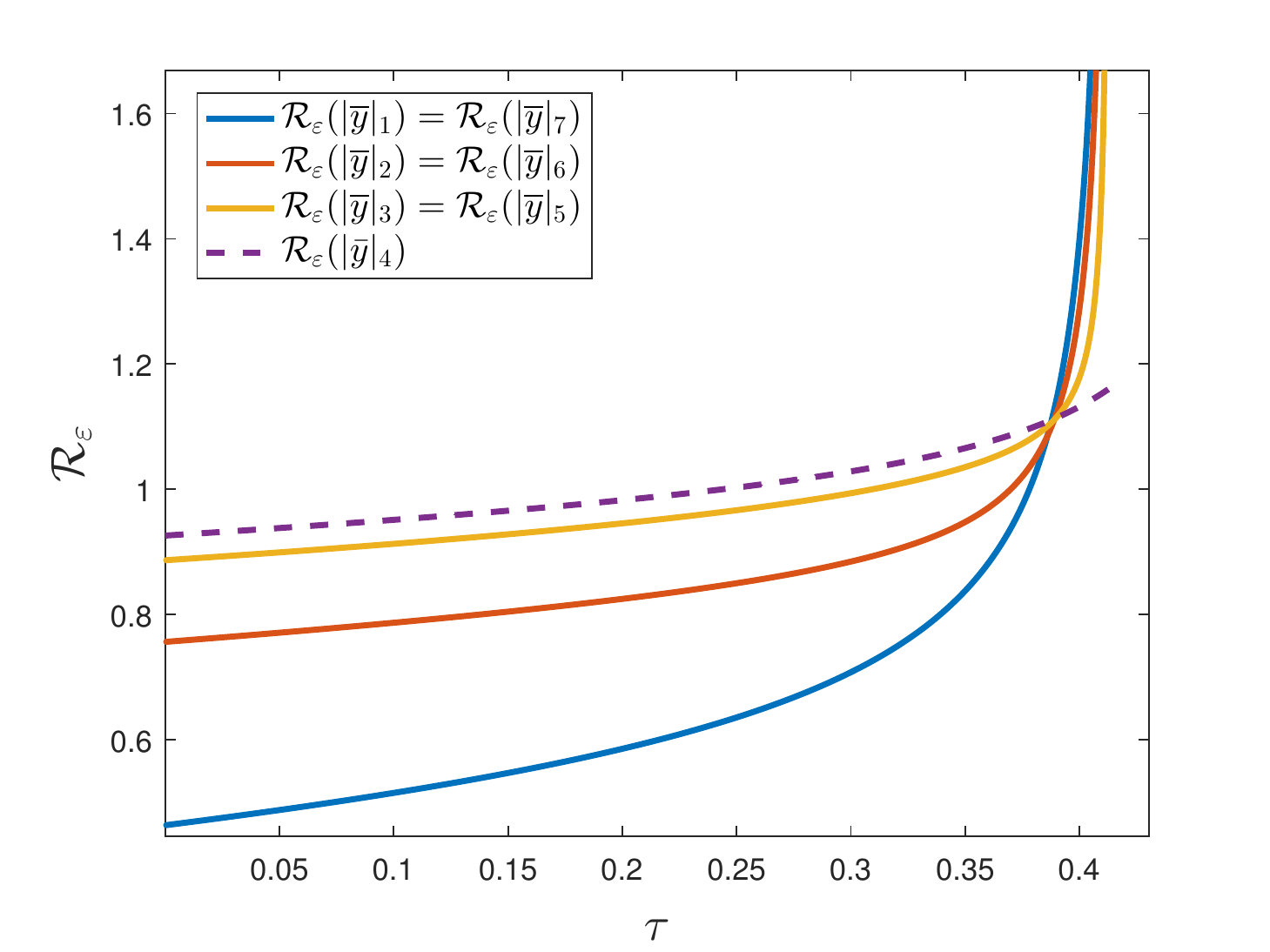}\caption{Systemic risk for $\mathbf{P}_{6}$ and $\mathbf{P}_{7}$. For small delays, the risk increases monotonically for as nodes lie closer to the center, where as for large delays, the risk monotonically decreases. When $n$ is odd and $\tau$ approaches the critical delay value, the risk of the central node increases towards infinity at a significantly lower rate than the rest of the network.}\label{fig: pathrisk}
\end{figure}

\vspace{0.1in}

\noindent {\it Path graph.} In networks for which agents are connected to a uniformly bounded number of neighbors, such as path and ring graphs, the marginal stability range for time delay  does not shrink to zero as the network size grows. A network with path graph topology $\mathbf{P}_n$ is marginally stable iff  $\tau < \frac{\pi}{4 \left(1+\cos(\frac{\pi}{n})\right)}$. A simple calculation reveals that the network is marginally stable for all  $n$ iff $\tau < \frac{\pi}{8}$. From Table \ref{table: flucrisktopgraph},  it can be verified that 
\[\mathcal R_{\varepsilon}(|\overline{y}|_{i})=\mathcal R_{\varepsilon}(|\overline{y}|_{n-i+1})\] 
for all  $i=1,\ldots,n$.

There exists a critical time delay $\tau^*$ in the marginal stability range that for all time delays $\tau < \tau^*$ agents residing toward the tails of the path graph are safer than the ones located closer to the center. In this case, central nodes are more prone to large fluctuations. When $\tau > \tau^*$, agents closer to the center  become safer and the ones that are closer to the tails become riskier.  In conclusion, agents residing in the boundaries are the safest for small time delays, and in contrary, agents located in the center are the safest for large time delays.  When $n$ is odd, there is a central agent who becomes significantly safer for large time delays. Figure \ref{fig: pathrisk} illustrates these analytic design rules of thumb graphically. 
%

%
%

\vspace{0.1in}

\noindent {\it Ring graph.}  The dynamical network \eqref{first-order} with ring graph topology $\mathbf{R}_n$ is marginally stable iff  
$$\tau < \frac{\pi}{4} \left( 1- \cos \left(2 \pi \left(1-\frac{1}{n}\right)\right)\right)^{-1}.$$ The symmetric structure of $\mathbf{R}_n$ necessitates that all agents to have identical risk values. In the right plot of Figure \ref{fig: riskcomplete}, the value of the systemic risk measure for one of the agents, which is given in Table \ref{table: flucrisktopgraph}, is drawn for different values of $n$. One observes that the systemic risk measure monotonically increases by $n$ for small time delays.  However, when time delay increases and approaches the marginal stability border,  networks with odd number of agents appear safer than networks with even number of agents.

%
%


\section{Simulations}
We verify our theoretical findings through two numerical examples. 

\begin{exmp} (Transient Behavior of Risk Measures)  Consider the graph over $n=5$ nodes as in Figure \ref{fig: example1}. 
\begin{figure}\center
\begin{tikzpicture}[auto, node distance=2cm, every loop/.style={},
                    main node/.style={circle,draw,font=\sffamily\large\bfseries,scale=0.8}]

  \node[main node] (1) {1};
  \node[main node] (2) [below left of=1] {2};
  \node[main node] (3) [below right of=1] {3};
  \node[main node] (4) [below of=3] {4};
  \node[main node] (5) [below of=2] {5};

  \path[every node/.style={font=\sffamily\small,scale=0.8}]
    (1) edge node [left] {2} (2)
        edge node  {3.2} (3)
    (2) edge node {0.1} (5)
        edge node  {5} (3)
    (4) edge node {0.2} (3)
    	edge node  {0.3} (5);
\end{tikzpicture} 
\caption{The graph for Example 1. The assigned weights accompany the edges between nodes, allow a maximum delay of $\tau_{\max}=0.1211$.}\label{fig: example1}
\end{figure}
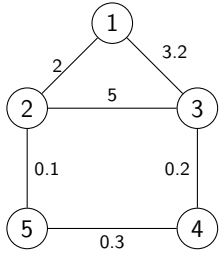

\begin{figure}\center
 \includegraphics[trim = 6mm 0mm 5mm 0mm, clip, scale=0.5]{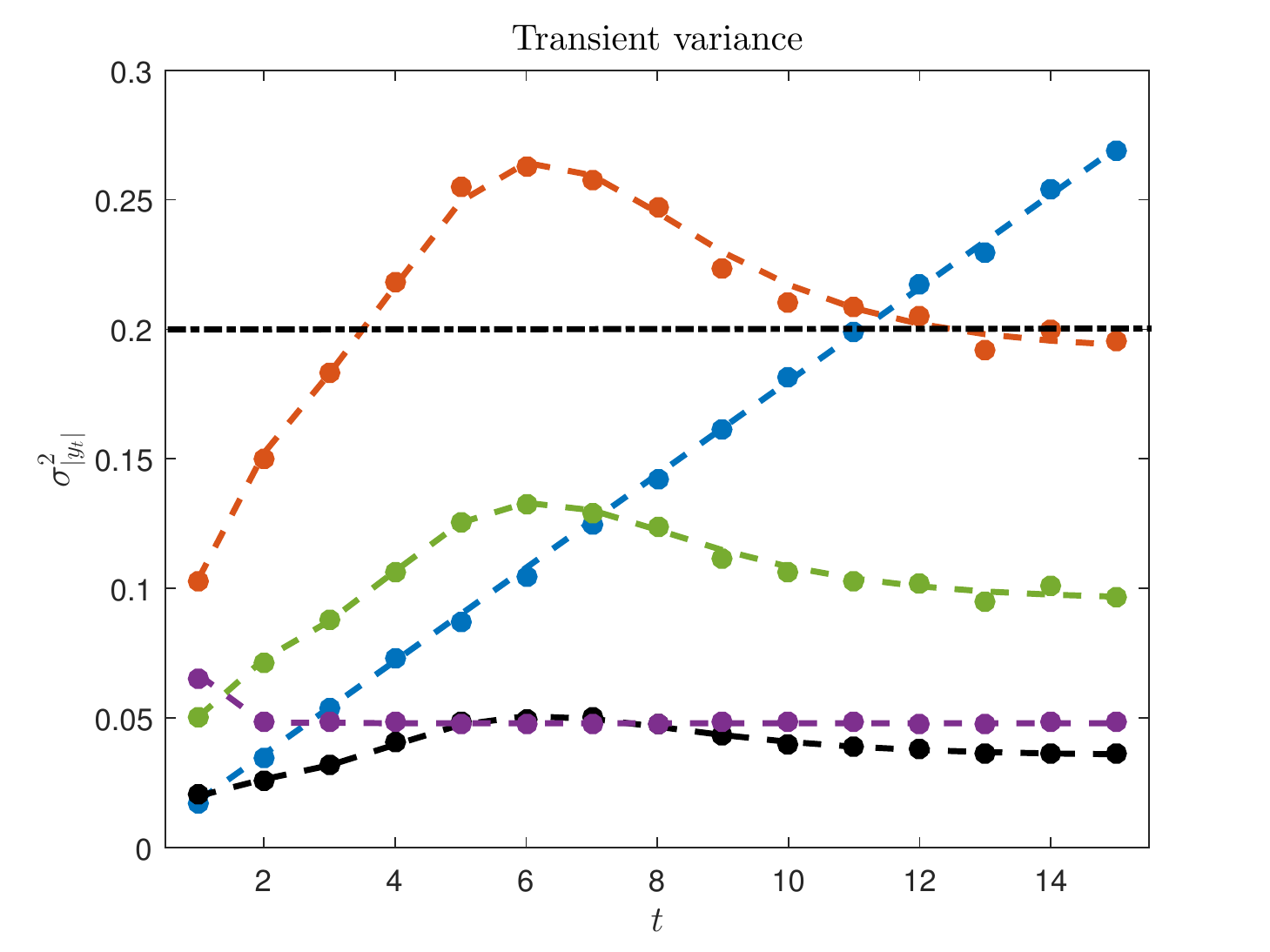}
\includegraphics[trim = 6mm 0mm 5mm 0mm, clip, scale=0.5]{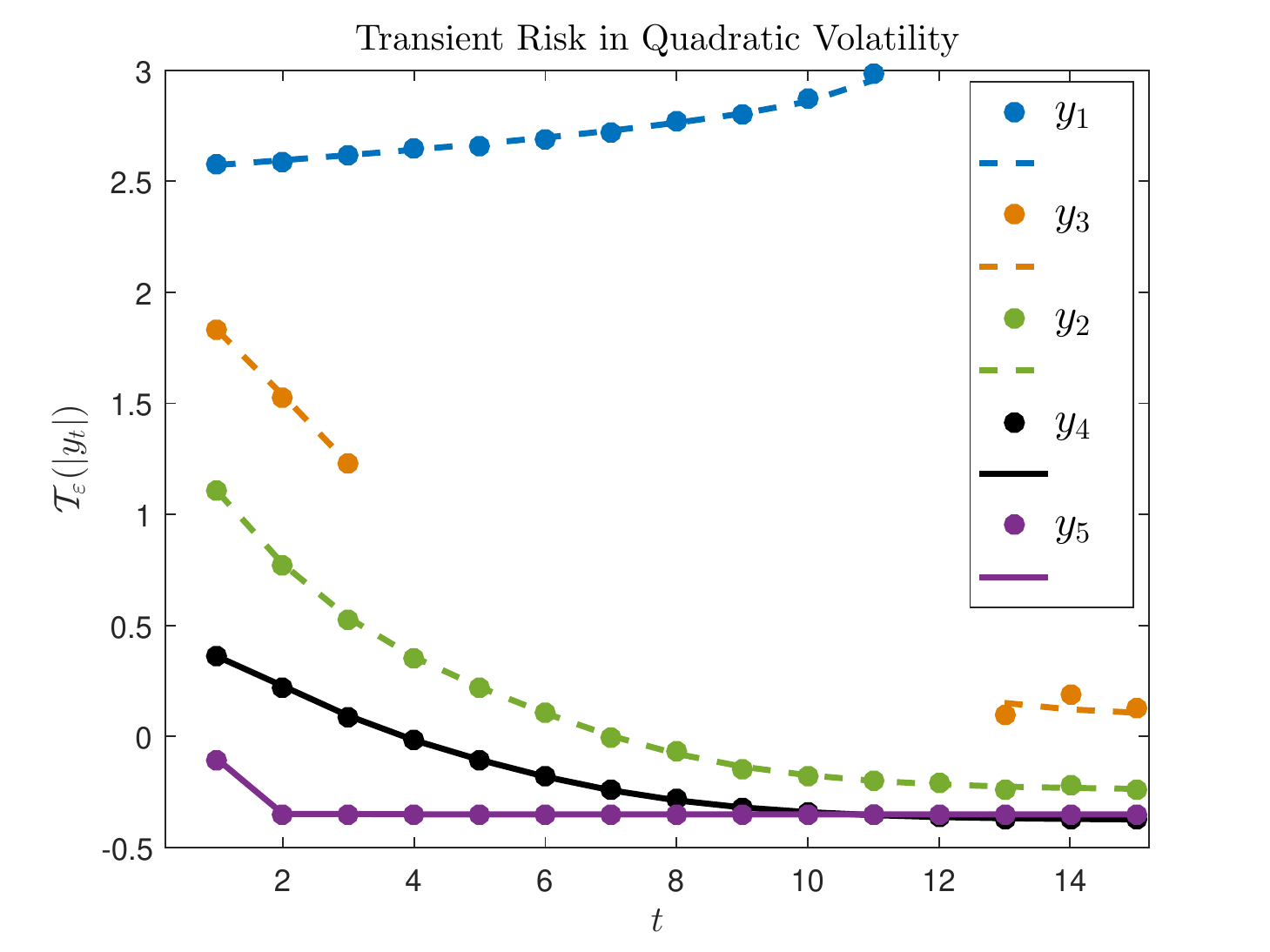}\caption{Simulated (dot) and analytic (solid line) results for $\mathbf y_t=C \mathbf x_t$. Figure on the left illustrates the variances and figure on the right the volatility risk with quadratic utility. For $t=4$ to $t=12$, the variance of the orange observable lies above the permissible cut-off value. According to Theorem \ref{thm: main0}, its risk in quadratic volatility is infinite. Similarly for the motion of the average (the line) for $t \geq 11$.}\label{fig: example1r1}
\end{figure}

\begin{figure}\center
\includegraphics[trim = 6mm 0mm 5mm 0mm, clip, scale=0.5]{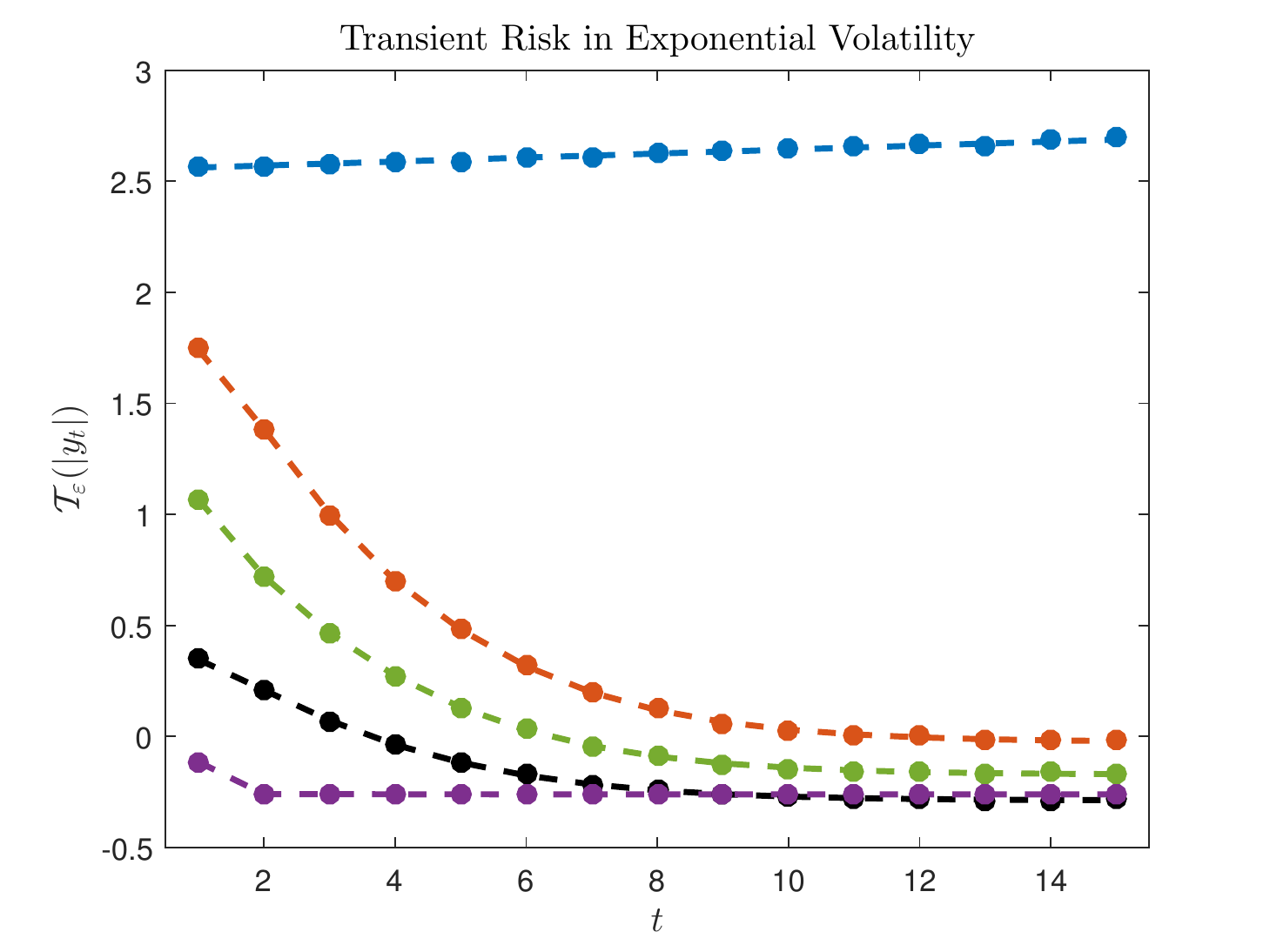}
\includegraphics[trim = 6mm 0mm 5mm 0mm, clip, scale=0.5]{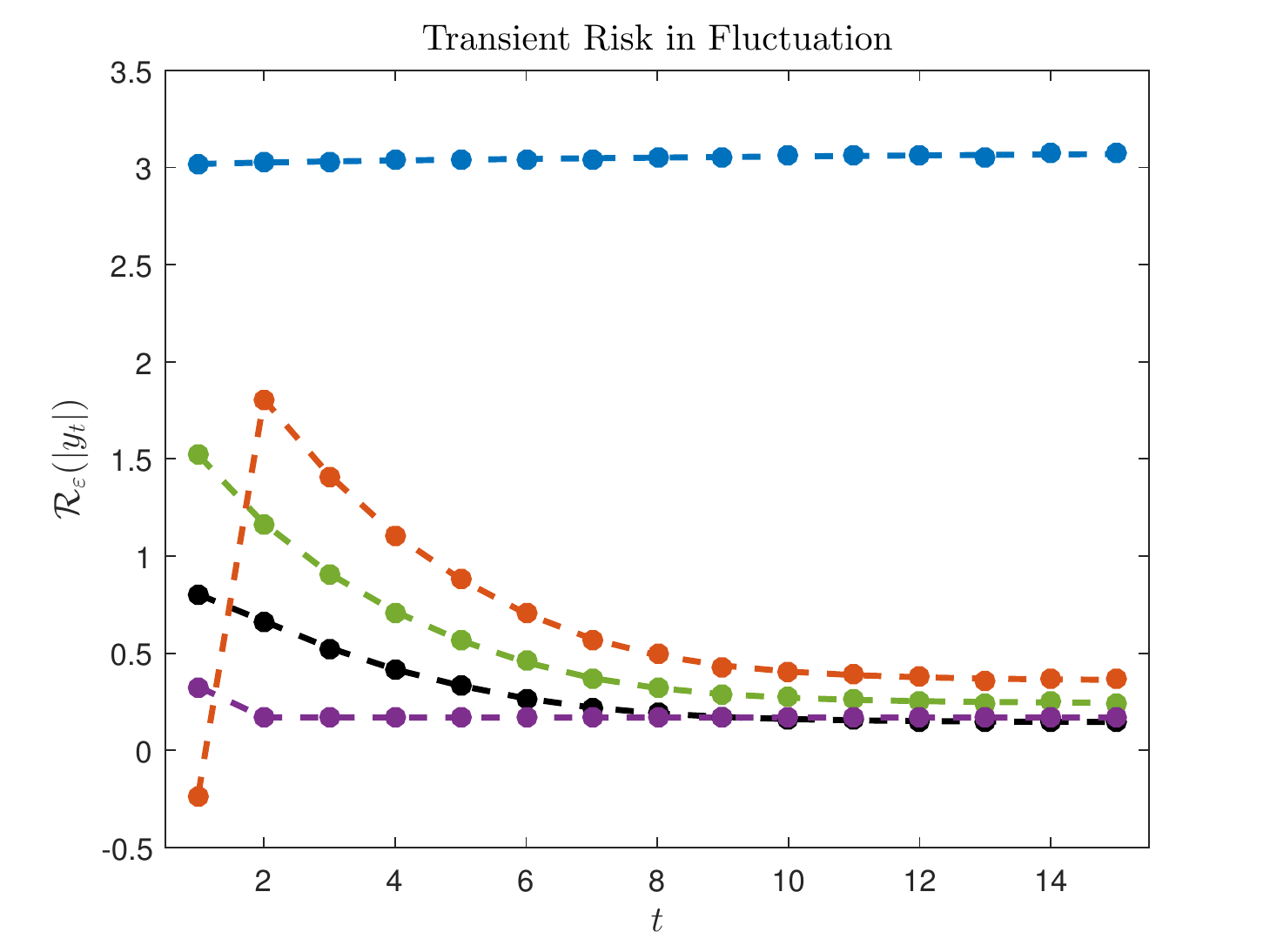}\caption{ Simulated (dot) and analytic (solid line) results for $\mathbf y_t=C \mathbf x_t$. Figure on the left illustrates the volatility risk with exponential utility and figure on the right illustrates the fluctuation risk.}
\label{fig: example1r2}
\end{figure}

The parametrization of \eqref{first-order} is set as $\tau=0.1$, $b=0.3$ and the initial conditions are chosen $\phi_i(t)=i$ for $t\in [-0.1,0]$. The cutoff value is selected to be $\varepsilon=\sqrt{0.2}$. 
%
We will examine the transient risk with the output matrix $C=[\frac{1}{5}\mathbf 1_5^T ~|~  \mathbf m_1^T  ~|~ \mathbf m_5^T~ |~ \mathbf m_2^T-\mathbf m_3^T ~|~ \mathbf m_2^T-\mathbf m_5^T ]$. The corresponding rows are the motion of the average, the deviation of nodes $i=1$ and $i=5$ from the average as well as the difference between nodes $2$ and $3$ and the difference between nodes $2$ and $5$. The results are presented in Figures \ref{fig: example1r1} and \ref{fig: example1r2}.  The risk of the average of motion becomes unbounded, as predicted. Finally, we notice a consistent pattern on behavior of risk among the observables $\mathbf c_1, \ldots, \mathbf c_5$: the risk involving nodes with stronger coupling is systematically smaller than the risk that involves nodes with weaker coupling, whereas the risk on the pairwise disagreement between the strongly coupled  node $2$ and the loosely coupled node $5$ appears to be the highest.
\end{exmp}

%
%

\begin{exmp}\label{exmp: M100} 
We consider a network with $100$ agents and an  arbitrary graph topology  that satisfies  Assumption \ref{assum0}. We will numerically calculate the value of the steady-state systemic risk in expectation  \eqref{eq: riskmeanquadss}  with respect to output matrix $C=M_{100}$ as a function of time delay in its marginal stability region. This way, an observable output is assigned to every agents.  According to  Theorem \ref{thm: main1}, the value of the systemic risk measure corresponding to each agent can be either negative, positive bounded or infinite depending on the value of parameter. For a fixed parameter $\varepsilon$ and depending on the value of each agent's systemic risk measure,  agents can be classified as being safe (negative), marginally safe (positive), and unsafe (infinite). For this example, we generate a network whose marginal stability range is $ \tau < \frac{\pi}{2\lambda_n}=0.3524$. We consider risk w.r.t. quadratic utility function  with cutoff threshold $\varepsilon=\sqrt{0.05}$.  For each class, we also record average of (weighted) node degrees as a connectivity measure. The left sub-plot in Figure \ref{fig: M100}  pictures the number of agents in each class  for every $\tau \in [0,0.35)$.  The right sub-plot in Figure \ref{fig: M100} depicts average of (weighted) node degrees. For small time delays, the transition from safe mode to marginally safe mode starts with those nodes that have smaller than average degrees (which is about $3.65$). For larger time delays, the transition from marginally safe mode to the unsafe mode starts with those nodes that have larger than average degrees. This remarkable phenomenon has been repeatedly observed for almost all networks with similar characteristics.  
\end{exmp}

\begin{figure}\center
\includegraphics[trim = 6mm 0mm 5mm 0mm, clip, scale=0.55]{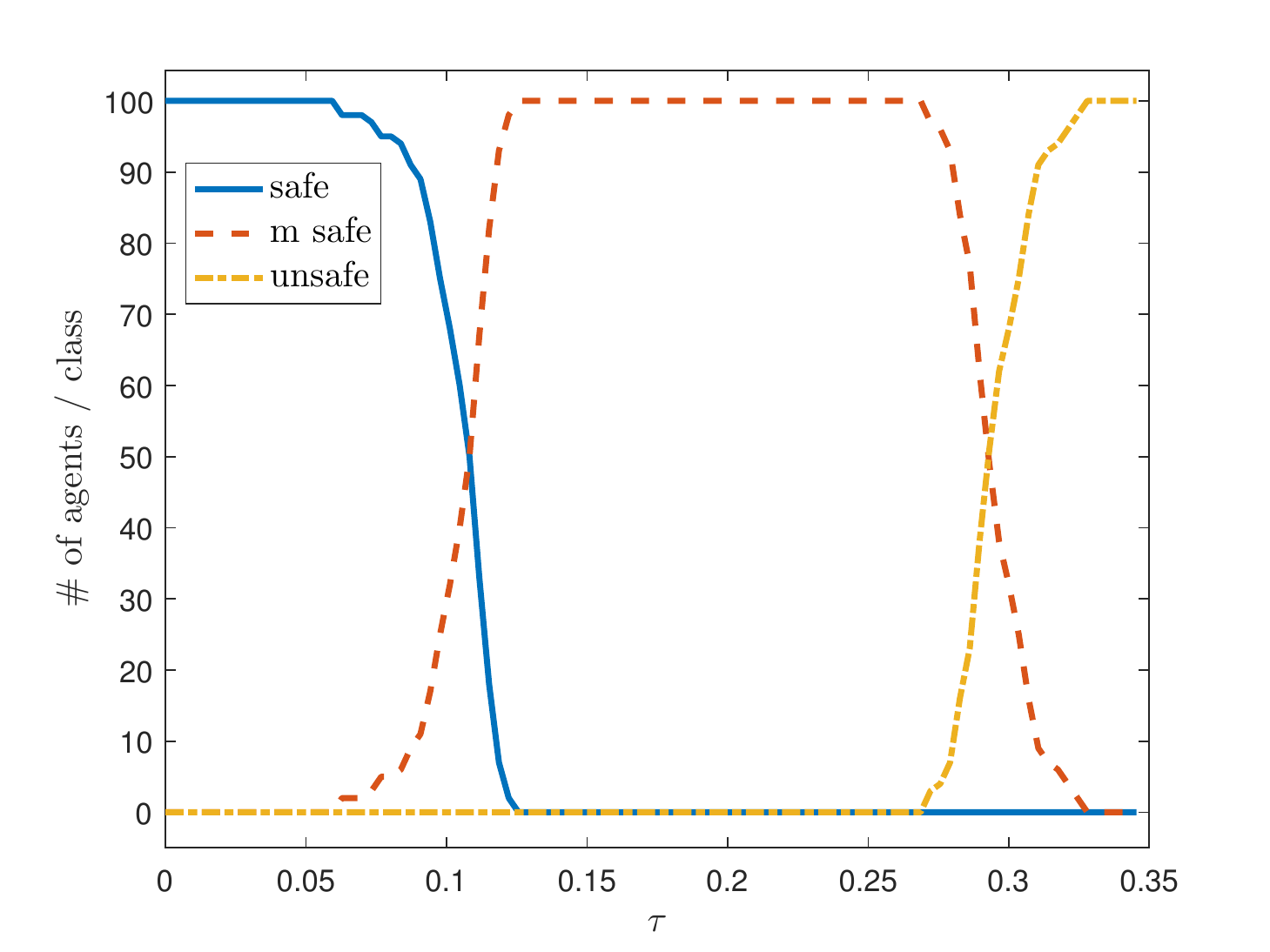}
\includegraphics[trim = 6mm 0mm 5mm 0mm, clip, scale=0.55]{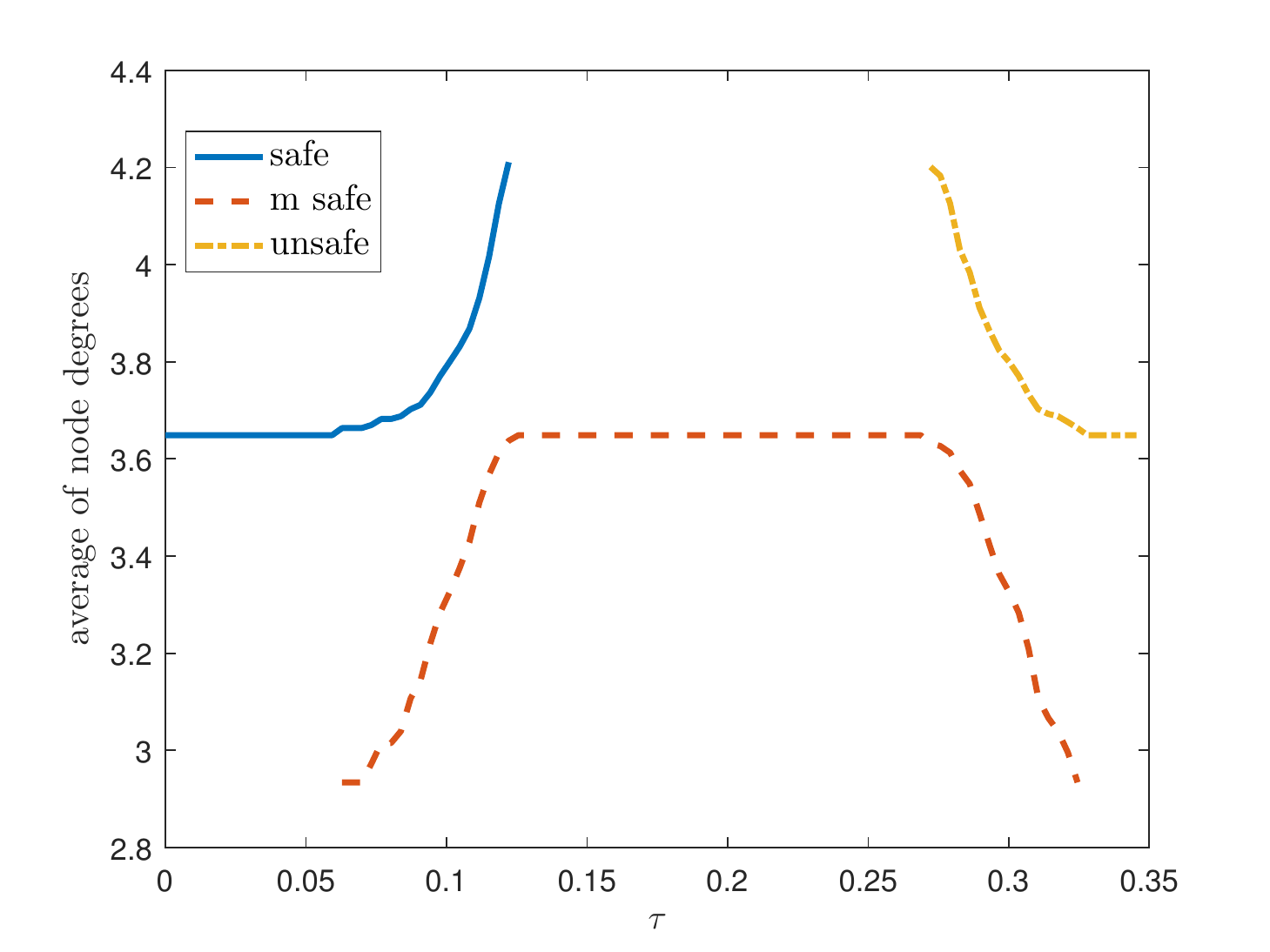}\caption{The network of Example \ref{exmp: M100}, where $\tau\in[0,0.35]$.  }\label{fig: M100}
\end{figure}

\begin{exmp}\label{exmp: B100}
We consider a network with $100$ agents and an  arbitrary graph topology  that satisfies  Assumption \ref{assum0}, where the output matrix is incidence matrix of the graph, i.e., $C=B_{100}$. The generated network is marginally stable for all time delay $\tau < 0.3415$. The number of observable outputs is $4950$ and the results of our simulations are illustrated in Figure \ref{fig: B100}. The effective graph resistance between nodes $i$ and $j$,
\begin{equation*}
\ER(i,j)=[B_{100} L^{\dag} B_{100}^T]_{ij}\footnote{ Here $L^{\dag}$ is the Moore-Penrose pseudo-inverse of  $L$ and the sum over all $i$ and $j$ equals $\Xi_{\mathcal{G}}$ as in \eqref{eq: graphr}. More details can be found in \cite{Mieghem11}.}
\end{equation*}  is used as a connectivity measure. For two distinct agents, the larger the effective resistance between them, the weaker these agents are connected to each other. Depending on the value of the risk measure with quadratic volatility for each individual event $|\overline{y}_i| > \delta_i$ for $i=1,\ldots, 4950$, these events can be labeled as being safe (negative), marginally safe (positive bounded), and unsafe (infinity).  The sub-plots in Figure \ref{fig: B100} asserts that: for small time delays, the transition from safe mode to marginally safe mode happens between those pair of nodes that maintain a larger than average effective resistance (i.e., lower than average coupling strength),  where the average effective resistance is about $0.544$. For larger time delay, the transition from marginally safe mode to unsafe mode occurs for those pair of nodes that have a lower than average effective resistance (i.e., higher than average coupling strength). 
\end{exmp}

\begin{figure}\center
\includegraphics[trim = 0mm 0mm 5mm 0mm, clip, scale=0.55]{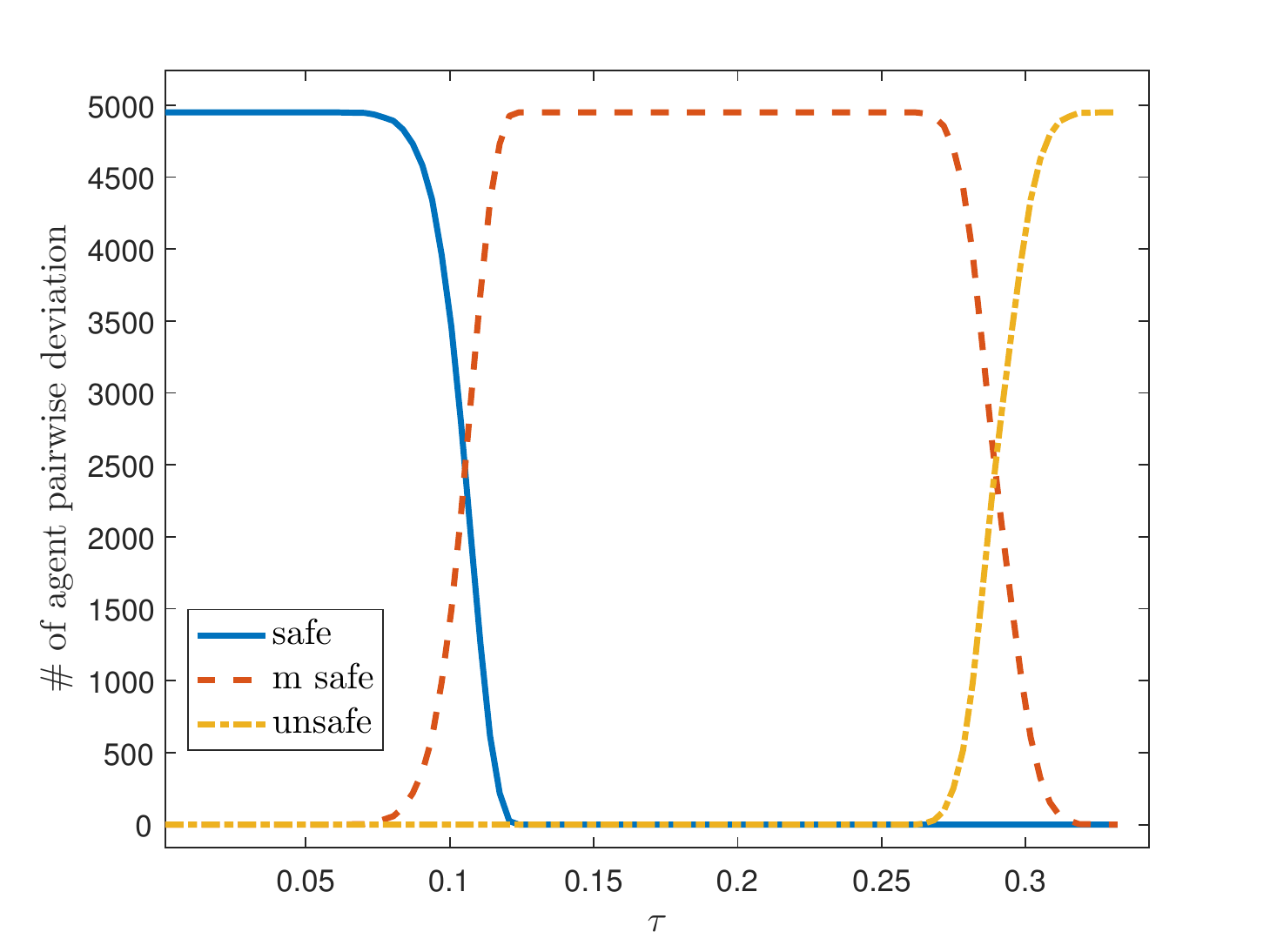}
\includegraphics[trim = 6mm 0mm 5mm 0mm, clip, scale=0.55]{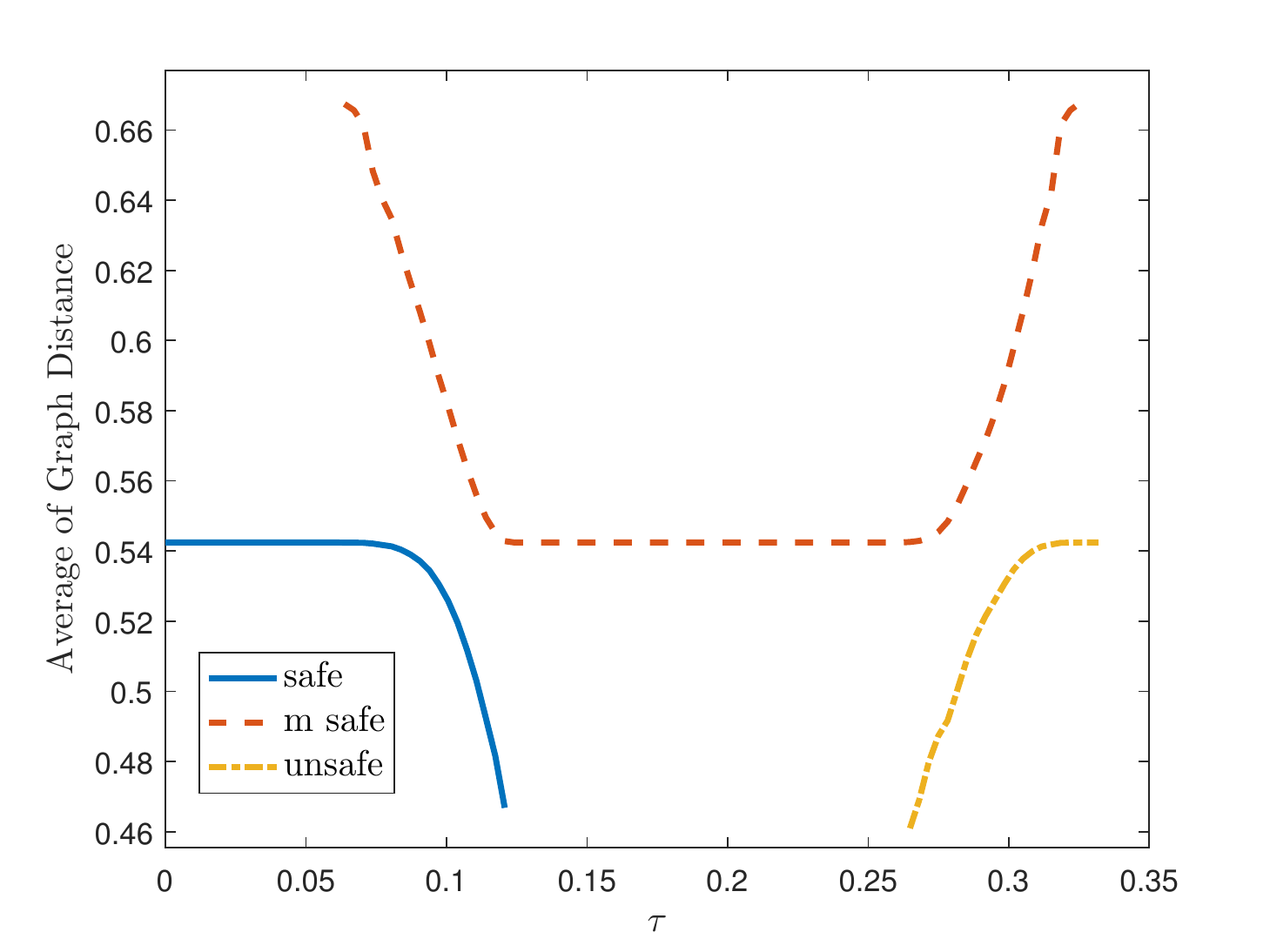}\caption{The network of Example \ref{exmp: B100}, where $\tau\in[0,0.34]$. }\label{fig: B100}
\end{figure}

\section{Discussion}\label{sec: discussion}

There are several points that need to be discussed. The focus of this paper is on the risk of large fluctuations in probability and in expectation. The value of these two types of risk measures capture behavior of relevant microscopic features of a network. Moreover, it turned out that their values depend on network topology through Laplacian spectrum and its eigenvectors. To investigate usefulness of other types of risk measures remains an open problem and is one of our future research directions. The results of this paper are particularly useful to develop design algorithms to optimize network connectivity, performance and risk in presence of external noise and time delay. The design of this class of networks involves only three main scenarios: growing by adding new feedback interconnections, sparsification (eliminating) of existing feedback loops, and reweighing feedback gains. Other design objectives can be obtained by combining these three design procedures. Our particular interest is to design networks that strike a balance among connectivity, performance, and risk. This is one of our current research directions.   



\appendix
\section*{A. Delay Differential Equations}\label{append: dde} Let the initial value problem
\begin{equation}\label{eq: nominaldiag}
\begin{split}
\dot \varphi(t)&=-a \varphi(t-\tau), \hspace{0.3in} t> 0\\
\varphi(t)&=\phi(t), \hspace{0.68in}  t\in [-\tau,0]
\end{split}
\end{equation} where $a>0,~\tau \geq 0$ and $\varphi= \varphi(t,0,\phi)\in \mathbb R,~t\geq 0$ its solution. It is well-known, \cite{lunel93}, that $z$ is exponentially stable with respect to zero if and only if $a\tau<\frac{\pi}{2}.$ Unfortunately, $\varphi$ attains a closed form expression only when $\tau=0$, i.e. $\varphi(t)=e^{-a t}\phi(0)$. For $\tau>0$, it is possible to derive a closed form expression of the energy of the fundamental solution of \eqref{eq: nominaldiag}; that is the integral $\int_0^{\infty}\varphi^2(t)\,dt$ with $\varphi$ the solution of \eqref{eq: nominaldiag} initiated as $\phi(t)=0$, $t\in [-\tau,0)$ and $\phi(0)=1$. 
\begin{lem}\footnote{The authors claim no novelty of this result (see for instance \cite{PhysRevE.86.056114,kuechler92}). The reason for providing a detailed proof is for the sake of consistency to the formalism of the paper but also because it is derived with alternative and simpler methodology.}\label{lem: vmatrix} If $a\tau<\frac{\pi}{2}$, the fundamental solution $\varphi$ of \eqref{eq: nominaldiag} satisfies
\begin{equation*}
\int_0^{\infty}\varphi^2(t)\,dt=\tau f(a\tau),
\end{equation*} where $f:\big(0,\frac{\pi}{2}\big)\rightarrow \mathbb R_+$ is defined through 
\begin{equation*}f(x)=\frac{1}{2x}\frac{\cos(x)}{1-\sin(x)}.\end{equation*} 
\end{lem}
\begin{proof} Given the fundamental solution $\varphi$, we define the quantity: \begin{equation}\label{eq: qfunctional}
V(\rho)=\int_{\rho}^{\infty}\varphi(t)\varphi(t-\rho)\,dt, \hspace{0.3in} \rho\in \mathbb R
\end{equation} 
\begin{claim}\label{claim} For $V$ as in \eqref{eq: qfunctional}, the following hold:
\begin{enumerate}
\item $V(\rho)=V(-\rho)$,
\item $\dot V(\rho)=-aV(\tau-\rho)$ for $\rho>0$, and
\item $V(\tau)=\frac{1}{2a}$.
\end{enumerate}
\end{claim}
Our objective is to calculate $V(0)$. Set $R(\rho)=V(\tau-\rho)$. The pair $(V,R)$ satisfies the system \begin{equation*}
\begin{split}
\dot V=-a R, ~ \hspace{0.4in}  \dot R= a V
\end{split}
\end{equation*} where $V(0)=R(\tau)$, $V(\tau)=R(0)=\frac{1}{2a}$, and with solution
\begin{equation*}
\begin{pmatrix}
V(\rho) \\ R(\rho)
\end{pmatrix}=\begin{pmatrix}
\cos(a\rho) & -\sin(a\rho) \\ \sin(a\rho) & \cos(a\rho)
\end{pmatrix}\begin{pmatrix}
V(0) \\ R(0)
\end{pmatrix}.
\end{equation*} The second row at $\rho=\tau$ reads $$R(\tau)=\sin(a\tau)V(0)+\cos(a\tau)R(0)=V(0).$$ Solving for $V(0)$, we have
\begin{equation*}
V(0)=\frac{\cos(a\tau)}{1-\sin(a\tau)}R(0)=\frac{\cos(a\tau)}{1-\sin(a\tau)}V(\tau),
\end{equation*} and from $(3.)$ we conclude $$\int_0^{\infty}\varphi^2(t)\,dt=\frac{1}{2a}\frac{\cos(a\tau)}{1-\sin(a\tau)}.$$
\end{proof}
\begin{proof}[Proof of Claim \ref{claim}]
For $(1.)$ we see that \begin{equation*}\begin{split}V(-\rho)&=\int_{-\rho}^{\infty}\varphi(t)\varphi(t+\rho)\,dt=\int_0^{\infty}\varphi(u-\rho)\varphi(u)\,du\\
&=\int_{\rho}^{\infty}\varphi(u-\rho)\varphi(u)\,du=V(\rho)\end{split}\end{equation*} because $\varphi$ is a fundamental solution. For $(2.)$ we calculate
\begin{equation}\label{eq: firstodeq}
\begin{split}
\frac{d}{d\rho}V(\rho)&=-\varphi(\rho)\varphi(0)-\int_{\rho}^{\infty}\varphi(t)\frac{\partial}{\partial\rho}\varphi(t-\rho)\,dt\\
&=-\varphi(\rho)+a\int_{\rho}^{\infty}\varphi(t)\varphi(t-\rho-\tau)\,dt\\
&=-\varphi(\rho)+aV(\rho+\tau).
\end{split}
\end{equation} From the symmetry of $V$,
\begin{equation*}\begin{split}
\frac{d}{d\rho}V(\rho)&=\frac{d}{d\rho}\big[V(-\rho)\big]=-\frac{d}{d\rho}V(-\rho)\\
&=-\varphi(-\rho)-aV(-\rho+\tau)=-aV(\tau-\rho)
\end{split}
\end{equation*}so long as $\rho>0$. Finally, while \eqref{eq: firstodeq} yields $$\dot V(0)=-\varphi(0)+aV(\tau)=-1+aV(\tau),$$ from property $(2.)$ we obtain $\dot V(0)=-a V(\tau).$ The result occurs after equating the right hand sides of these two expressions.
\end{proof} 
\section*{B. Proofs Of Technical Results}

 \begin{proof}[Proof of Lemma \ref{lem: boundofcov}]
Direct calculations yield 
\begin{equation*}\begin{split}
C\Sigma_tC&=\int_0^{t}CQ\Phi(t-s)Q^T B B^T Q \Phi(t-s) \Sp Q^T C\,ds\\
&=b^2 \Sp CQ\int_0^t\Phi^2(t-s)\,ds ~(CQ)^T,
\end{split}
\end{equation*} 
where $\Phi_L(t) = Q \Phi(t) Q^T$ in which  $\Phi(t)=\mathrm{diag}\big(\varphi_{1}(t),\ldots,\varphi_{n}(t)\big)$ and $\varphi_k$'s are eigensolutions of the unperturbed network.  By letting $C=[\mathbf c_1,\dots,\mathbf c_N]^T$, the elements of $C\Sigma_tC^T$ are given by \begin{equation*}\begin{split}
&[C\Sigma_t C^T]_{ij}=b^2 \bigg[\int_{0}^t\sum_{k}(\mathbf c_i^T\mathbf q_{k})\varphi_{k}^2(t-s)(\mathbf c_{j}^T\mathbf q_k)\,ds\bigg]_{ij}.
\end{split}
\end{equation*} According to our assumptions, it follows that $\varphi_{1}(t)\equiv 1$ and all $\varphi_k(t)$ for $k=2,\ldots,n$ are exponentially fast vanishing functions. Therefore, matrix $C\Sigma_t C$ is bounded if and only if $\sum_{k}(\mathbf c_i^T\mathbf q_{1})(\mathbf c_{j}^T\mathbf q_1)=0$ for all $i,j \in \V$. But for $i=j$ we have $(\mathbf c_i^T\mathbf q_{1})^2=0$ for all $i$. The reverse part is trivial.
\end{proof}

\begin{proof}[Proof of Proposition \ref{prop: cohrisk}] \textit{(i.)}  If $z$ is independent of $\omega\in \Omega$, then 
$\mathbb P(z>\delta)=0$ for any $\delta$ arbitrarily close to $z$ from the right. In view of $\varepsilon\in [0,1)$, the definition of $\mathcal R_{\varepsilon}(z)$ imposes the infimum of the sum that is achieved when all $\delta$ attain the infimum value which is $z$. 

\textit{(ii.)} Note that $ \mathcal R_\varepsilon(z+m)=\inf\{ \delta~|~\mathbb P\big(z(\omega)+m>\delta\big)<\varepsilon\} $ is equivalent to   $ \inf\{ \delta~|~\mathbb P\big(z(\omega)>\delta-m\big)<\varepsilon\} $. For $\delta':=\delta+m$ we obtain $\mathcal R_\varepsilon(z+m)=\inf\{ \delta'+m~|~\mathbb P\big(z(\omega)>\delta'\big)<\varepsilon\} =m+ \inf\{ \delta'~|~\mathbb P\big(z(\omega)>\delta'\big)<\varepsilon\}=\mathcal R_\varepsilon(z)$.

\textit{(iii.)} Let $z_1, z_2\in \mathbb L^2$ with $z_1\leq z_2$ almost surely. Then 
$\{z_1>\delta\}\subseteq \{z_2 >\delta\}$ that implies
$\mathbb P \big(z_1>\delta\big)\leq \mathbb P\big(z_2>\delta\big)$. The decrease of $\delta$ on $z_1$ to achieve $\varepsilon$ is smaller than the corresponding decrease on $z_2$. So the risk of $z_1$ cannot be larger than the risk of $z_2$.

\textit{(iv.)} For any $\lambda>0$, $\mathbb P\big(\lambda y>\delta\big)=\mathbb P\big(y>\delta/\lambda\big)$ and simple substitutions yield the result.

\textit{(v.)} Let $z=\alpha x+(1-\alpha)y$ with $x,y \in \mathbb L^2$ and $\mathbb P\big(x\leq y\big)=p$, $\mathbb P\big(x> y\big)=1-p$ and $\alpha\in [0,1]$.
Then 
\begin{equation*}
\mathbb P\big(z>\delta\big)=\mathbb P\big(z>\delta|x\leq y\big)p+\mathbb P\big(z>\delta|x> y\big)(1-p).
\end{equation*}If $x\leq y$, then
$\mathbb P\big(z>\delta\big)\leq \mathbb P(y>\delta\big) $
while, if $x\geq y$, then 
$\mathbb P\big(z>\delta\big)\leq \mathbb P(x>\delta\big)$
Using the last two inequalities and property $(iii.)$ the result follows. \end{proof}

\begin{proof}[Proof of Proposition \ref{prop: riskmom}] \textit{(i.)} Follows similarly with $(ii.)$ of Proposition \ref{prop: cohrisk}.

\textit{(ii.)} Due to the convexity and monotonicity properties of $v$, it is easy to verify that $\mathbb A$ is convex. Now let $\{z_i\}_{i=1,2}\in \mathbb A$ and $\{\delta_i\}_{i=1,2}\leq 0$ the corresponding risk values. In view of $(i.)$,  it holds that $\mathcal T_{\varepsilon}(z_i- \delta_i)= 0$, i.e. $z_i-\delta_i \in \mathbb A$ and by convexity $$\alpha(z_1- \delta_1)+(1-\alpha)(z_2-\delta_2)\in \mathbb A.$$ Again from $(i.)$ we get
\begin{equation*}\begin{split}0&\geq \mathcal T_{\varepsilon}\big(\alpha(z_1-\delta_1)+(1-\alpha)(z_2-\delta_2)\big)\\
&= \mathcal T_{\varepsilon}\big((\alpha z_1+(1-\alpha)z_2)-(\alpha \delta_1+(1-\alpha)\delta_2)\big).\end{split}\end{equation*} which yields the result.

\textit{(iii.)} If $z_2 \leq z_1$ then $z_2-\delta \leq z_1-\delta$, so that $z_2- \delta\leq z_1-\delta$
and $v(z_2-\delta)\leq u(z_1- \delta)$. Consequently the cost $\delta$ to bring $v(z_2-\delta)$ closer to $v(\varepsilon)$ is smaller than the respective cost on $ u(z_1-\delta)$. 

\textit{(iv.)} For every $z\in \mathbb B$, \textit{(iii.)} shows that $\lambda z \in \mathbb A$ as well for $\lambda\geq 0$. We take cases on $\lambda$:

$\lambda\in [0,1]$ : Then $\lambda z+\mathbf 1 (1-\lambda) \varepsilon\in \mathbb A$ as one can verify that $\mathcal T_{\varepsilon}(\mathbf 1 \varepsilon )=0$, i.e. $\varepsilon\in \mathbb A$. From \textit{(i.)}, $\mathcal T_{\varepsilon}(\lambda z)+(1-\lambda)\varepsilon=\mathcal T_{\varepsilon}(\lambda z+(1-\lambda)\mathbf 1 \varepsilon)\leq \lambda \mathcal T_{\varepsilon}(z)$, and solving for $\mathcal T_{\varepsilon}(\lambda z)$ we obtain the result.
 
$\lambda\geq 1$ : Then it follows from $\lambda z, \mathbf 1 \varepsilon \in \mathbb A$ and properties \textit{(i.)} and \textit{(ii.)}:  $\mathcal T_{\varepsilon}(z)+(1-\lambda^{-1})\varepsilon=\mathcal T_{\varepsilon}(z+(1-\lambda^{-1})\varepsilon)=\mathcal T_{\varepsilon}(\lambda^{-1}(\lambda z)+(1-\lambda^{-1})\varepsilon)\leq \lambda^{-1}\mathcal T_{\varepsilon}(\lambda z)$. Solving for $\mathcal T_{\varepsilon}(\lambda z)$ we obtain the result.
\end{proof}

\begin{proof}[Proof of Lemma \ref{prop: uniqueness}]We write  $\sigma_t^2(Q)=b^2\mathbf c^T Q \int_0^t \Phi^2(t-s)\,ds Q^T\mathbf c$ and with little abuse of notation our focus is on the matrix $\Sigma_t=b^2 Q \int_0^t \Phi^2(t-s)\,ds\,Q^T$ . This is a diagonal form for every $t$ and characterizes the algebraic multiplicity of the eigenvalues $\zeta_k=\zeta_k(t)=\int_0^t\varphi_k^2(t-s)\,ds,~i\in \V$. Fix $t$ and $i\in \V$. Note that we can write $\sigma_t^2(Q)=b^2\sum_{k}\zeta_k(t)\sum_{j=1}^{m_i}(c_j^Q)^2$ where $m_k=m_k(t)$ is the algebraic multiplicity of $\zeta_k$. We take cases.

\noindent $m_k=1$: the corresponding normal (real) eigenvector is unique up to a sign. Consequently, $c_k^{Q_1}=\pm c_k^{Q_2}$ therefore $(c_k^{Q_1})^2=(c_k^{Q_2})^2$.

\noindent $m_k>1$: $\zeta_k(t)$ is repeated $m_k$ times while it produces $m_k$ linearly independent normal eigenvectors that can be chosen mutually orthogonal. Then the collection of the eigenvectors of $Q_1$ that span the subspace of $\zeta_k(t)$, say $Q^{\zeta_k}_1$ and  the corresponding collection of $Q^{\zeta_k}_2$ of $Q_2$ are associated with an orthogonal matrix $P$ such that $Q_1^{\zeta_k}=Q_{2}^{\zeta_k} P$. Then \begin{equation*} 
\begin{split}
&\sum_{j=1}^{m_i}(c_j^{Q_1^{\zeta_k}})^2=\mathbf c^T Q_1^{\zeta_k} (\mathbf c^TQ_1^{\zeta_k})^T=\mathbf c^T Q_{2}^{\zeta_k} P (\mathbf c^T Q_{2}^{\zeta_k} P)^T\\
&=\mathbf c^T Q_{2}^{\zeta_k} P P^T (Q_{2}^{\zeta_k})^T\mathbf c =\mathbf c^T Q_{2}^{\zeta_k} (\mathbf c^T Q_{2}^{\zeta_k})^T= \sum_{j=1}^{m_i}(c_j^{Q_2^{\zeta_k}})^2.
\end{split}
\end{equation*} For $t$ and $k$ being arbitrary, the result follows.     
\end{proof}

\begin{proof}[Proof of Theorem \ref{thm: main0}]
From the discussion in \S \ref{subs: statisticsofoutput}, $y_t\sim \mathcal N(\mu_t,\sigma_t^2)$ with $\mu_t$ and $\sigma_t^2$ as in \eqref{eq: meanvaluescalar} and \eqref{eq: standarddevscalar}. Then $|y_t|$ follows a folded normal distribution with mean $\mu_{|y_t|}$ and variance $\sigma_{|y_t|}^2$ defined in the statement of the theorem. For the volatility risks we distinguish between $v(x)=x^2$ and $v(x)=e^{\beta x}$.

For $v(x)=x^2$,
$
\mathcal T_{\varepsilon}(|y_t|)=\inf\big\{\delta\in \mathbb R: \mathbb E[(|y_t|-\delta)^2]\leq \varepsilon^2\big\}$
 Solving for $\delta$ the inequality $\mathbb E[(|y_t|-\delta)^2]\leq \varepsilon^2$ we obtain that the infimum is achieved at the root
$\delta=\mathbb E[|y_t|]-\sqrt{\varepsilon^2-\mathbb E\big[\big(|y_t|-\mathbb E[|y_t|]\big)^2\big]}$. For $v(x)=e^{\beta x}$,
\begin{equation*}
\mathbb E\big[e^{\beta (|y_t|-\delta)}\big]\leq e^{\beta \varepsilon} \Leftrightarrow \delta \geq \frac{1}{\beta} \ln \big(\mathbb E[e^{\beta (|y_t|-\varepsilon)}]\big)
\end{equation*}
so $\mathcal T_{\varepsilon}(|y_t|)=\frac{1}{\beta} \ln \big(\mathbb E[e^{\beta (|y_t|-\varepsilon)}]\big)$ and the result follows by straightforward algebra on $$\mathbb E[e^{\beta |y_t|}]=\int_{-\infty}^{\infty}e^{\beta|u|}e^{-\frac{(u-\mu_t)^2}{2\sigma_t^2}}\,du.$$
\noindent Finaly from Eq. \eqref{eq: risk} we calculate the probability
\begin{equation*}\begin{split}
&\mathcal R_\varepsilon\big(|y_t|\big)=\inf\big\{\delta>0:\mathbb P(|y_t|>\delta)\leq \varepsilon \big\}\\
&= \inf\big\{\delta>0:\mathbb P(|y_t|\leq \delta)\geq 1-\varepsilon \big\} \\
&=\inf\bigg\{\delta>0:\frac{1}{\sqrt{2\pi}\sigma_t}\int_{-\delta}^{\delta}e^{-\frac{(x-\mu_t)^2}{2\sigma_t^2}}\,dx\geq (1-\varepsilon)\bigg\}\\
&=\inf\bigg\{\sqrt{2}\sigma_t\delta+\mu_t>0:\frac{1}{\sqrt{\pi}}\int_{-\delta-2\frac{\mu_t}{\sqrt{2}\sigma_t}}^{\delta}e^{-u^2}\,du\geq (1-\varepsilon)\bigg\}\\
&=\sqrt{2}\sigma_t S_{\varepsilon}\bigg(\frac{\mu_t}{\sqrt{2}\sigma_t}\bigg)+\mu_{t},
\end{split}
\end{equation*}
where $S_{\varepsilon}(\alpha)$ as in \eqref{eq: sgaussiant}.
\end{proof}
\begin{proof}[Proof of Theorem \ref{thm: main1}]
The condition $\mathbf c^T\mathbf q_1=0$ implies $c_1^Q=0$ or equivalently the marginally stable mode to be unobservable at the output. This implies $\lim_{t\rightarrow +\infty} \mu_{t}=0$ for $\mu_t$ as in \eqref{eq: meanvaluescalar} and $\lim_{t\rightarrow +\infty}\mu_{|y_t|}=\sqrt{\frac{2}{\pi}}\overline{\sigma}$ for $\overline{\sigma}=\lim_{t\rightarrow +\infty} \sigma_t$ and $\sigma_t$ as in \eqref{eq: standarddevscalar}. Of course, $
\overline{\sigma}^2 = \sum_{k> 1}(c_k^Q)^2\int_{0}^{\infty}\varphi_k^2(s)\,ds $ and we can combine Theorem \ref{thm: main0} together with Lemma \ref{lem: vmatrix} to derive the corresponding expressions \eqref{eq: riskmeanquadss}, \eqref{eq: expriskss} and \eqref{eq: riskprob} for all three risk measures.
\end{proof}

\begin{proof}[Proof of Theorem \ref{thm: monotonicity}]
Direct calculations, yield \begin{equation*}
\frac{\partial}{\partial \tau}\overline{\sigma}^2=\frac{1}{2}b^2\sum_{k=2}^n\frac{1}{1-\sin(\lambda_k\tau)}(c_k^Q)^2>0.
\end{equation*}
From this and the fact the systemic risk measures are increasing functions of $\overline{\sigma}$, it immediately follows  that $\mathcal R_\varepsilon$ or $\mathcal T_\varepsilon$ are strictly increasing functions of time delay as well.
\end{proof}

\begin{proof}[Proof of Theorem \ref{thm: doubleinequality2}]
We observe that $\boldsymbol{\mathcal R}_\varepsilon(|\overline{\mathbf y}|)$ is a set-valued function whose value is equivalent to the solution of the following multi-dimensional joint chance-constrained optimization problem 
\begin{eqnarray}
& & \hspace{-5.1cm}\underset{\boldsymbol \delta}{\textrm{minimize}}  ~~\boldsymbol \delta\label{eq: problem2}\\
& & \hspace{-5.1cm} \mbox{subject to:} ~~ \mathbb P\left(\Sp|\overline{\mathbf y} | \preceq \boldsymbol \delta \Sp \right) \geq 1-\varepsilon.\label{eq: problem2-1}
\end{eqnarray}
In general, this problem is very difficult to solve even numerically, let alone find closed form solutions. To overcome computational complexity and calculate bounds, we decompose the joint probability constraint \eqref{eq: problem2-1} into several individual constraints. Let us consider positive numbers $\varepsilon_1, \ldots, \varepsilon_q$  such that \eqref{epsilon-q} holds. Suppose that $\delta_k^*$ is the optimal solution of the following optimization problem
\begin{eqnarray}
& & \hspace{-4.8cm}\underset{\delta_k}{\textrm{minimize}}  ~~\delta_k \label{problem3-1}\\
& & \hspace{-4.8cm} \mbox{subject to:} ~~ \mathbb P\left(\Sp|\overline{y}_k| \leq   \delta_k \Sp \right) \geq 1-\varepsilon_k\label{problem3-2}
\end{eqnarray}
for $k=1,\ldots,q$. is an upper bound for the optimal solution of \eqref{eq: problem2}-\eqref{eq: problem2-1}. Let us denote $\boldsymbol \delta^* = [\delta_1^*,\ldots, \delta_q^*]^T$. From Bonferroni  inequality \cite{Genz:2009:CMN:1695822}, we have that 
\[ 
\PP\left(\Sp|\overline{\mathbf y} | \preceq \boldsymbol \delta^* \Sp \right) ~\geq~ \sum_{k=1}^q \PP\left(\Sp|\overline{y}_k| \leq   \delta_k^* \Sp \right) - (q-1)
\]
\noindent From constraint \eqref{epsilon-q} and \eqref{problem3-2}, it follows that 
\begin{eqnarray*} 
\PP\left(\Sp|\overline{\mathbf y} | \preceq \boldsymbol \delta^* \Sp \right) &\geq & \sum_{k=1}^q (1-\varepsilon_k) - (q-1)  =  1 - \sum_{k=1}^q \varepsilon_k ~\geq ~ 1-\varepsilon.
\end{eqnarray*}
\noindent This implies that $\boldsymbol \delta^*$ is a feasible solution for optimization problem \eqref{eq: problem2}-\eqref{eq: problem2-1}. Thus, we conclude that 
\begin{eqnarray}
\boldsymbol{\mathcal R}_{\varepsilon}(|\overline{\mathbf y}|)  &\preceq&  [\Sp \delta_1^*,\ldots, \delta_q^* \Sp]^T \nonumber \preceq  \left[\Sp {\mathcal R}_{\varepsilon_1}(|\overline{y}_1|),~\ldots~, {\mathcal R}_{\varepsilon_q}(|\overline{y}_q|) \Sp \right]^T. \label{ineq_vec_risk_1} 
\end{eqnarray}
By using the fact that ${\mathcal R}_{\varepsilon_k}(|\overline{y}_k|) = S_{\varepsilon_k}(0) S_{\varepsilon}(0)^{-1} {\mathcal R}_{\varepsilon}(|\overline{y}_k|) $, we can rewrite \eqref{ineq_vec_risk_1} to conclude that the optimal solution set must satisfy the upper bound of $\mathbb W_{\mathfrak{R}_{\varepsilon}}$.  In order to prove the lower bound, we use an equivalent representation of  \eqref{eq: riskflucmmonec} in the following form
\begin{equation*}\label{risk-complement}
\boldsymbol{\mathcal R}_{\varepsilon}(|\overline{\mathbf y}|)   =  \inf\big\{\boldsymbol \delta \in \R^q~\big|~\mathbb P\big(\Sp  |\overline{y}_1| > \delta_1 \Sp \vee \Sp \ldots \Sp \vee \Sp |\overline{y}_q| > \delta_q\big)\leq \varepsilon\big\}. 
\end{equation*}  
Suppose that $\boldsymbol \delta^{**}$ is an optimal solution of \eqref{eq: problem2}-\eqref{eq: problem2-1}. Then, 
\[\mathbb P\big(\Sp  |\overline{y}_1| > \delta^{**}_1 \Sp \vee \Sp \ldots \Sp \vee \Sp |\overline{y}_q| > \delta^{**}_q\big)\leq \varepsilon. \]
From Fr\'echet inequality \cite{Ruschendorf1991},
\[ \max \Big\{ \mathbb P\big(\Sp  |\overline{y}_k| > \delta^{**}_k\big) ~\big|~k=1,\ldots,q\Big\} ~\leq~\mathbb P\left(\Sp \bigvee_{k=1}^q |\overline{y}_k| > \delta^{**}_k \Sp \right),\]
it follows that $\mathbb P\big(\Sp  |\overline{y}_k| > \delta^{**}_k\big) ~\leq~\varepsilon$ for all $k=1,\ldots,q$.  Thus, $\delta_k^{**}$  is a feasible solution of \eqref{problem3-1}-\eqref{problem3-2} and its optimal solution is upper bounded by $\delta_k^{**}$. This implies that the solution set $ \boldsymbol{\mathcal R}_{\varepsilon}(|\overline{\mathbf y}|)$ must honor the lower bound of $\mathbb W_{\mathfrak{R}_{\varepsilon}}$.  
\end{proof}
\begin{proof}[Proof of Theorem \ref{thm: multobjvolrisk}] We recall from Section \ref{sec:stat} that in steady-state $\overline{\mathbf y}\sim \mathcal N({\mathbf 0}, \overline{\Sigma})$   if and only if $\mathbf 1_n$ is in the kernel of output matrix $C$. Thus,  constraint $\mathbb{E}\big[v(|\overline{\mathbf y}|-\boldsymbol \delta) \big] \leq  v(\varepsilon)$ in \eqref{eq: riskmomv} with quadratic function is equivalent to 
\begin{equation*}
\sum_{i=1}^q \overline{\sigma}_i^2-2\sqrt{\frac{2}{\pi}}\sum_{i=1}^q\delta_i \overline{\sigma}_i+\sum_{i=1}^q\delta_i^2\leq \varepsilon^2,
\end{equation*} where $\overline{\sigma}_i^2$ is the $i$'th diagonal element of the covariance matrix $\overline{\Sigma}$. From completing the square,  we obtain
$$ \sum_{i=1}^q \left[\bigg(\delta_i-\sqrt{\frac{2}{\pi}} \overline{\sigma}_i\bigg)^2 + \left(1-\frac{2}{\pi} \right) \overline{\sigma}_i^2 \right]\Sp \leq  \Sp \varepsilon^2.$$ 
The constraint is clearly feasible if and only if $r>0$. The feasible set lies on the disk  $ \sum_{i=1}^q \big(\delta_i-\sqrt{\frac{2}{\pi}} \overline{\sigma}_i\big)^2 \leq r $. Taking the infimum over delta, we arrive at the subset on the sphere that satisfies $$\text{both}~\delta_i\leq 0~\text{and}~\sum_{i=1}^q \bigg(\delta_i-\sqrt{\frac{2}{\pi}} \overline{\sigma}_i\bigg)^2 = r $$ which identifies with the set in the second part in the statement of the Theorem, concluding the proof.
\end{proof}
\begin{proof}[Proof of Theorem \ref{thm: multobjvolrisk2}]It suffices to show that the individual margins $\delta_i$, stacked as a vector $(\delta_1,\dots,\delta_q)^T$ belong to the subset $\boldsymbol{\mathcal T}_{\varepsilon}(|\overline{\mathbf y}|)$ of  $\mathbb R^q$, as in Theorem \ref{thm: multobjvolrisk}. We recall from \eqref{eq: riskmeanquadss}
\begin{equation*}
\mathcal{T}_{\varepsilon_i}(|\overline{y}_i|)=\frac{2}{\pi}\overline{\sigma}_i-\sqrt{\varepsilon_i^2 - (1-\frac{2}{\pi})\sigma_i^2}=:\delta_i
\end{equation*} this is equivalent to 
\begin{equation*}
\bigg(\delta_i-\sqrt{\frac{2}{\pi}}\overline{\sigma}_i^2\bigg)^2=\varepsilon_i^2-\bigg(1-\frac{2}{\pi}\bigg)\overline{\sigma}_i^2
\end{equation*} summing over $i=1,\dots,q$ we have
\begin{equation*}
\sum_{i=1}^q \bigg(\delta_i-\sqrt{\frac{2}{\pi}}\overline{\sigma}_i^2\bigg)^2=\sum_{i=1}^q \varepsilon_i^2- \bigg(1-\frac{2}{\pi}\bigg)\sum_{i=1}^q\overline{\sigma}_i^2=r.
\end{equation*}
\end{proof}
\begin{proof}[Proof of Theorem \ref{thm: tradeoff}] It suffices to calculate the lower bound of the stead state variance, $\sigma^*$, as in \eqref{eq: fl}. The corresponding fundamental limits on risk measures are then directly derived. Recall the steady-state variance $\overline{\sigma}^2$ from Theorem \ref{thm: main1} and observe that:
\begin{equation*}\begin{split}
\overline{\sigma}^2&=b^2 \sum_{k>1} \frac{(c_k^Q)^2}{2\lambda_k}\frac{\cos(\lambda_k\tau)}{1-\sin(\lambda_k\tau)} \geq  b^2\tau\big[\inf_{z\in(0,\frac{\pi}{2})}f(z)\big] \sum_k (c_k^Q)^2\\
&=\frac{b^2\tau}{2} \|\mathbf c\|^2\frac{1}{1-\sin(z^+)}=:\sigma_*^2
\end{split}
\end{equation*} as the infimum of $f(z),~z\in (0,\pi/2)$ is achieved at $z$ that satisfies $z=\cos(z)$, and $\sum_i (c_i^Q)^2=\|\mathbf c\|^2$ for $c_1^Q=0$, $Q$ orthogonal basis. 
\end{proof}
\begin{proof} [Proof of Theorem \ref{thm: tradeoff2}]
For the proof of this result we will need the following Lemma:
\begin{lem}\label{lem: auxtradeoff} For arbitrary but fixed  $\alpha,~\beta\geq 0$, the function $$p_{\alpha,\beta}(x)=\frac{(\alpha+\beta x)\cos(x)}{x^2(1-\sin(x))}$$ attains a unique minimum at $(0,\pi/2)$.
\end{lem} 
We directly calculate,
\begin{equation*}\begin{split}
\overline{\sigma}^2\Xi_{\mathcal G}&=b^2\frac{n}{2}\bigg(\sum_{k>1} \frac{(c_k^Q)^2}{\lambda_k}\frac{\cos(\lambda_k\tau)}{1-\sin(\lambda_k \tau)}\bigg)\bigg(\sum_{k>1}\frac{1}{\lambda_k}\bigg)\\
&= b^2\frac{n}{2}  \sum_{k>1}\bigg[\frac{(c_k^Q)^2}{\lambda_k}\frac{\cos(\lambda_k\tau)}{1-\sin(\lambda_k \tau)}\sum_{k'=2}^{k}\bigg(\frac{1}{\lambda_{k'}}\bigg) +\\
&\hspace{1.1in}+\frac{(c_k^Q)^2}{\lambda_k}\frac{\cos(\lambda_k\tau)}{1-\sin(\lambda_k \tau)}\sum_{k'=k+1}^{n}\bigg(\frac{1}{\lambda_{k'}}\bigg)\bigg]\\
&> \frac{n b^2\tau^2}{2}\sum_{k>1}(c_k^Q)^2\bigg(\frac{(k-1)\cos(\lambda_k\tau)}{(\lambda_k\tau)^2\big(1-\sin(\lambda_k\tau)\big)}+\\
&\hspace{1.4in}+\frac{2(n-k)}{\pi}\frac{\cos(\lambda_k\tau)}{\lambda_k\tau\big(1-\sin(\lambda_k\tau)\big)}\bigg)\\
&=\frac{n b^2\tau^2}{2}\sum_{k>1}(c_k^Q)^2p_{\alpha_{k},\beta_{k}}(\lambda_{k}\tau)\geq \|\mathbf c\|^2\frac{nb^2\tau}{2}\min_{k> 1}p_{\alpha_{k},\beta_{k}}(\lambda_{k}\tau)\\
& \geq \|\mathbf c\|^2\frac{nb^2\tau}{2}\min_{k> 1} \min_{x\in (0,\pi/2)} p_{\alpha_{k},\beta_{k}}(x) 
\end{split}
\end{equation*} 
where $p_{\alpha_{k},\beta_{k}}(x)$ is as in Lemma with $\alpha_{k}=k-1$ and $\beta_k=\frac{2(n-k)}{\pi}$. Now all the limits in $(i.), (ii.)$ and $(iii.)$ occur after combining \eqref{eq: fl} with the results of Theorem \ref{thm: main1}.  For the proof for the trade-offs we work as follows:

$(i.)$ The mean value theorem yields:
\begin{equation*}\begin{split}
\mathcal T_\varepsilon + \varepsilon & =\sqrt{\frac{2}{\pi}}\overline{\sigma}-\sqrt{\varepsilon^2-\bigg(1-\frac{2}{\pi}\bigg)\overline{\sigma}^2}+\sqrt{\varepsilon^2}\\
& \geq \sqrt{\frac{2}{\pi}}\overline{\sigma}+\frac{1}{2\varepsilon}\bigg(1-\frac{2}{\pi}\bigg)\overline{\sigma}^2.
\end{split}
\end{equation*} By virtue of \eqref{eq: fl} and \eqref{eq: fll}
 then
\begin{equation*}
\big(\mathcal T_\varepsilon + \varepsilon\big)\cdot  \sqrt{\Xi_{\mathcal G}}>\bigg(\sqrt{\frac{2}{\pi}}+\frac{1}{2\varepsilon}\bigg(1-\frac{2}{\pi}\bigg)\sigma_*\bigg)\vartheta_*
\end{equation*}

For $(ii.)$ we recall the  inequality $\ln(1+x)\geq \frac{x}{x+1}$ for $x>-1$ and we calculate:
\begin{equation*}
\frac{\ln\big(1+\mathrm{erf}(\beta\overline{\sigma}/2)\big)}{\beta}\geq \frac{1}{\beta}\cdot\frac{\mathrm{erf}(\beta\overline{\sigma}/2)}{1+\mathrm{erf}(\beta\overline{\sigma}/2)}\geq \frac{\mathrm{erf}(\beta\overline{\sigma}/2)}{2\beta}
\end{equation*} 
The best bound we can now get is 
\begin{equation*}
\frac{\ln\big(1+\mathrm{erf}(\beta\overline{\sigma}/2)\big)}{\beta}\sqrt{\Xi_{\mathcal G}}> \frac{\mathrm{erf}(\beta\sigma_*/2)}{2\beta} \frac{2n(n-1)\tau}{\pi}
\end{equation*} So
\begin{equation*}
\big(\mathcal T_{\varepsilon}+\varepsilon\big)\sqrt{\Xi_{\mathcal G}}>\frac{\beta}{2}\sigma_*\vartheta_*+ \frac{\mathrm{erf}(\beta\sigma_*/2)}{2\beta} \frac{2n(n-1)\tau}{\pi}
\end{equation*}

For $(iii.)$ the result follows immediately from \eqref{eq: fll}.
\end{proof}

\begin{proof}[Proof of Corollary \ref{cor: multilimitstradeoffs}] $(i)$ let $\delta_i$ the $i^{th}$ element of $ \boldsymbol{\mathcal T}_{\varepsilon}(|\mathbf{\overline{y}}|)$. From Theorem \ref{thm: multobjvolrisk}, it follows that $\delta_i\in \big(\sqrt{\frac{2}{\pi}}\overline{\sigma}_i-r,\sqrt{\frac{2}{\pi}}\overline{\sigma}_i+r\big)$. It is straightforward to verify that  $$\delta_i\geq \sqrt{\frac{2}{\pi}}\sigma_i^*-\sqrt{\varepsilon^2-\bigg(1-\frac{2}{\pi}\bigg)\sum_{i=1}^q(\overline{\sigma}_i^*)^2},$$ where  the result follows similarly to the steps in the proof of Theorem \ref{thm: tradeoff}. $(ii)$ Let us denote $\boldsymbol{\mathcal T}_{\varepsilon}(|\mathbf{\overline{y}}|)=(\delta_1,\dots,\delta_q)$ and recall  Eq. \eqref{eq: riskvectorexponential2}.  
From the Jensen's inequality:
\begin{equation*}
\begin{split}
\ln\big( \mathbb E\big[e^{\beta(\sum_{i}(|\overline{y}_i|-\frac{1}{q}\varepsilon)}\big]\big)&\geq \ln\big(e^{\mathbb E[\beta(\sum_{i}(|\overline{y}_i|-\frac{1}{q}\varepsilon)]}\big)\\
&=\beta \sqrt{\frac{2}{\pi}}\sum_{i=1}^q \sigma_i-\beta \varepsilon
\end{split}
\end{equation*} From Theorem \ref{thm: tradeoff}, we have that $\sigma_i\geq b \|\mathbf c_i\|\sqrt{\frac{\tau}{2(1-\sin(z^+))}}$. Substitute this lower bound to the first inequality of the proof and apply this bound to Eq. \eqref{eq: riskvectorexponential2} to conclude.
\end{proof}

\begin{proof}[Proof of Lemma \ref{lem: auxtradeoff}]
It can be easily shown that $\lim_{x\uparrow \frac{\pi}{2}}p_{\alpha,\beta}(x)=\lim_{x\downarrow  0}p_{\alpha,\beta}(x)=+\infty$. This means that, there exists at least one extremum that is a minimum of $p_{\alpha,\beta}(x)$ in $(0,\pi/2)$. Now this extremum must satisfy 
\begin{equation*}
p_{\alpha,\beta}'(x)=-\frac{(\beta x+2\alpha)\cos(x)-\beta x^2-\alpha x}{x^2(1-\sin(x))}=0.
\end{equation*} The claim is concluded if we show that, $I(x)=(\beta x+2\alpha)\cos(x)-\beta x^2-\alpha x$, has a unique solution in any compact subset of $(0,\pi/2)$. At first note that
$I(0)=2\alpha>0$ and $I(\pi/2)=-\beta\big(\frac{\pi}{2}\big)^2-\alpha\frac{\pi}{2}<0$. By the Intermediate Value Theorem, $I(x)$ attains a root in $(0,\pi/2)$. If this is not unique then $I'(x)=0$ should have a solution in the interval between the two distinct roots of $I(x)$. Now,  \begin{equation*}
I'(x)=\beta\cos(x)-(\beta x+2\alpha)\sin(x)-2\beta x-\alpha
\end{equation*} does not change sign in $(0,\pi/2)$ if $\beta\leq \alpha$. If $\beta\geq \alpha$ then the root $\xi$ of $I'(x)$ satisfies $$\beta\cos(\xi)-(\beta \xi+2\alpha)\sin(\xi)-2\beta \xi-\alpha=0.$$ Plugging the above relation to $I(x)$ we have $$I(\xi)=\frac{1}{\beta}(\beta\xi+2\alpha)^2\sin(\xi)+\beta \xi^2+4\alpha\xi+\frac{2\alpha^2}{\beta}>0$$
 Consequently, there exists $x'>0$ such that $I(x')>0$ and $I'(x)<0$ for $x\in(x',\pi/2)$ proving that $I$ vanishes at a unique point in $(0,\pi/2)$ for every fixed pair of $\alpha$ and $\beta$.
\end{proof}

\section*{APPENDIX C: Derivation of Risk Formulas}

The output vector $\mathbf c$ is set as $\mathbf c^T=\mathbf m_i$ for the deviation of the average measure. We remark that, $\mathbf m_i^T\mathbf q_k=q_{ik}$, i.e. the $i^{th}$ element of the $k^{th}$ eigenvector of $L$. Finally, we will take $b=1$ for the diffusion coefficient.

\subsection*{The Complete Graph, $\mathbf K_n$} This is a simple undirected graph in which every pair of distinct vertices is connected by a unique edge. The eigenspace of $\mathbf K_n$ is $\lambda_1=0$ with $\mathbf q_1=\frac{1}{\sqrt{n}}\mathbbm 1$, and $\lambda_k=n$ for $k>1$ with linearly independent, unitary  eigenvectors $\frac{1}{\sqrt{2}}(\mathbf e_1-\mathbf e_k)$ that are normal, but not mutually orthogonal. For this topology we therefore assume $n\tau<\frac{\pi}{2}$. While a Gram-Schmidt process could produce an orthonormal basis, this is not necessary here. From Theorem 1 we calculate the variance
\begin{equation*}\begin{split}\bar\sigma^2&=b^2\tau\sum_{k>1} (c_k^Q)^2 f(\lambda_k\tau)=\tau f(n\tau)\sum_{k>1}(c_k^Q)^2=\tau f(n\tau) ||\mathbf c||^2=\tau\frac{n-1}{n} f(n\tau)
\end{split}
\end{equation*}
The symmetry of $K_n$ is particularly convenient for an explicit calculation of the the variance of all elements of $|\bar{\mathbf y}|$ and which is of course identical for all $|\bar y|_i$. The fluctuation risks are $$\mathcal R^{(i)}_\varepsilon(|\bar{y}|)=S_{\varepsilon}(0)\sqrt{\frac{2\tau(n-1)}{n}f(n\tau)}.$$ 

\subsection*{Wheel graph $\mathbf W_{n+1}$}
This graph occurs out of adding in the circulant graph of $n$ nodes an additional central node. We denote this central node with $i=1$. The next result establishes the eigen-structure of $W_{n+1}$.
\begin{lem}\label{lem: wheelgraph}
The eigen-structure of $W_{n+1}$ is $\lambda_1=0$ with $\mathbf q_1=\frac{1}{\sqrt{n+1}}\mathbbm 1_{n+1}$, $~\lambda_k=3-2\cos\big(\frac{2\pi(k-1)}{n}\big)$, with $\mathbf q_k=\frac{1}{\sqrt{n}}(0,1,e^{i\frac{2k\pi}{n}},\dots,e^{i\frac{2k(n-1)\pi}{n}})^T$ where $i^2=-1$, $k=2,\dots,n$ and  $\lambda_{n+1}=n+1,$ with $\mathbf q_{n+1}=\frac{1}{\sqrt{n^2+n}}(n,-1,\dots,-1)^T$.
\end{lem}
\begin{proof}[Proof of Lemma \ref{lem: wheelgraph}]
Let's assign $n=1$ to the center node. The $(n+1)\times (n+1)$ laplacian matrix is $$L=\begin{bmatrix}
n & -1 & -1 & -1 & \cdots & -1 & -1 \\
-1 & 3 & -1 & 0 & \cdots & 0 & -1 \\
-1 & -1 & 3& -1& \cdots & 0 & 0 \\
-1 & 0 & -1 & 3 & \cdots & 0 & 0 \\
\vdots & \vdots & \vdots & \vdots & \vdots &\vdots & \vdots \\
-1 & -1 & 0 & 0 & \cdots & -1 & 3 
\end{bmatrix}$$
The eigenvalues of $W_{n+1}$ are calculated in \cite{Mieghem:2011:GSC:1983675}. The eigenvector for $\lambda_{n+1}$ is directly shown to satisfy the eigenvalue equation $L\mathbf q_{n+1}=(n+1)\mathbf q_{n+1}$. For $\lambda_k\neq \lambda_{n+1}$, the first row 
\begin{equation*}
n q_1-\sum_{j=2}^{n+1}q_k=\lambda_{k}q_1 \Leftrightarrow  n q_1+q_1=\lambda_k q_1 \Leftrightarrow q_1=0
\end{equation*} by virtue of $\sum_{j=1}^{n+1}q_j=0$. Since $k$ is arbitrary, all the eigenvectors, other than $\mathbf q_1$ and $\mathbf q_{n+1}$, have their first element to be zero. As for the rest elements of $\mathbf q_k$, these are calculated from the $n\times n$ submatrix of $L$, after one removes the first row and column. This sub-matrix is circulant $\{\mathbf q_k\}_{k=2}^{n}$ as in the statement of the Lemma consist a set of orthonormal eigenvectors (see also \cite{davis94}).
\end{proof}
Note that $\mathbf q_k$ are taken to be complex but without loss of generality since $(\mathbf c^T \mathbf q_k)^2=(\mathbf c^T\mathbf q_k)(\mathbf c^T\mathbf q_k)^*$. This choice simplifies greatly the calculation of the steady state variance. The stability condition for the wheel graph scheme is  $(n+1)\tau<\frac{\pi}{2}$.

The deviation of the $i^{th}$ component from the average yields the steady-state variance:
\begin{equation*}
\bar{\sigma}_i^{2}=\begin{cases}
n\tau f\big((n+1)\tau\big), & i=1\\
\frac{\tau}{n}\sum_{k>1}f(\lambda_k\tau)+\frac{\tau}{n} f\big((n+1)\tau\big), & i = 2,\dots,n+1
\end{cases}
\end{equation*}

The symmetry of $W_{n+1}$ distinguishes between the variance of the central node $i=1$ and the rest $n$ nodes that surround it. We note that for $i\neq 1$
\begin{equation*}\begin{split}
\frac{\bar{\sigma}_i^2}{\bar{\sigma}_1^{2}}&=\frac{1}{n^2}+\sum_{k=2}^{n}\frac{(n+1)}{n^2}\frac{\cos(\lambda_k\tau)}{\lambda_k(1-\sin(\lambda_k\tau))}\frac{1-\sin((n+1)\tau)}{\cos((n+1)\tau)}\\
&<\frac{1}{n^2}+\frac{n^2-1}{n^2} \frac{\cos(\lambda^*\tau)}{1-\sin(\lambda^*\tau)}\frac{1-\sin((n+1)\tau)}{\cos((n+1)\tau)}
\end{split}
\end{equation*} where $\lambda^*$ is the maximizer over $\{\lambda_k\}_{k=2}^{n}$ of $\frac{\cos(\lambda\tau)}{1-\sin(\lambda\tau)}$ so that elementary algebra yields $\frac{\bar{\sigma}_i^2}{\bar{\sigma}_1^{2}}<1$ if $$\frac{\cos(\lambda^*\tau)}{1-\sin(\lambda^*\tau)}\frac{1-\sin((n+1)\tau)}{\cos((n+1)\tau)}\leq 1.$$ The latter inequality holds true in view of $\frac{\partial}{\partial \lambda}\frac{\cos(\lambda\tau)}{1-\sin(\lambda\tau)}>0.$ Consequently the variance of the central node is larger than the variance of the rest ones.


\subsection*{The Complete Bipartite Graph, $\mathbf K_{n_1,n_2}$.}
This graph consists of two sets of nodes denoted as $[n_1]$ and $[n_2]$, where, without loss of generality, $n_1\leq n_2$, and $n=n_1+n2$. Each node of one set is connected to all other nodes of the other set. There are no links between nodes of a same set. The following Lemma describes the eigenspaces of $K_{n_1,n_2}$.
\begin{lem}\label{lem: cbipeigen} The eigenstructure of $K_{n_1,n_2}$ is
\begin{itemize}
\item[a.] $\lambda_1=0$ with multiplicity one and eigenvector $\mathbf q_1=\frac{1}{\sqrt{n}}\mathbbm 1$
\item[b.] $\lambda_{k}=n_1$ with multiplicity $n_2-1$ and eigenvectors, $$\mathbf q_k=\frac{1}{||\mathbf w_k||}\mathbf w_k \hspace{0.2in} \text{for} \hspace{0.2in}\mathbf w_k=\frac{1}{i-1}\sum_{j=n_1+1}^{n_1+k-1}\mathbf e_j-\mathbf e_{n_1+k}, $$ with $k=2,\dots,n_2.$
\item[c.] $\lambda_k=n_2$ with multiplicity $n_1-1$ and eigenvectors, 
\begin{equation*}
\mathbf q_{k}=\frac{1}{||\mathbf w_k||}\mathbf w_k \hspace{0.2in} \text{for} \hspace{0.2in}\mathbf w_k=\frac{1}{k-n_2}\sum_{j=1}^{k-n_2}\mathbf e_j-\mathbf e_{k-n_2+1},\end{equation*} with $k=n_2+1,\dots,n-1.$

\item[d.] $\lambda_{n}=n_1+n_2$  with multiplicity 1 and eigenvector $\mathbf q_n=\big[\sqrt{\frac{n_2}{n n_1}}\mathbbm 1_{n_1}^T, \sqrt{\frac{n_1}{n n_2}}\mathbbm 1_{n_2}^T\big]^T$ .
\end{itemize}
\end{lem}
\begin{proof}[Proof of Lemma \ref{lem: cbipeigen}]
For the eigenvalues of $K_{n_1,n_2}$ we refer to \cite{Mieghem:2011:GSC:1983675}. For the eigenvectors we work as follows. From \cite{Mieghem:2011:GSC:1983675} (or from simple observation) the laplacian matrix of $K_{n_1,n_2}$ is of the form
\begin{equation*}L=\begin{bmatrix}
n_2 I_{n_1} & -J_{n_1, n_2} \\
-J_{n_2, n_1} & n_1 I_{n_2}
\end{bmatrix}
\end{equation*} where $J_{\alpha,\beta}$ is the $\alpha\times \beta$ matrix of ones. Now, $$L\mathbf q_n=\lambda_n \mathbf q_n $$ with the structure $\mathbf q_n=[\mathbf q_{n}^{(1)}, \mathbf q_{n}^{(2)}]^T$ where $\mathbf q_{n}^{(1)}\in \mathbb R^{n_1}$ and $\mathbf q_{n}^{(2)}\in \mathbb R^{n_2}$. This  implies that the the eigenvalue equation is equivalent to $$n_2\mathbf q_{n}^{(1)}-J_{n_1, n_2}\mathbf q_{n}^{(2)}=(n_1+n_2)\mathbf q_{n}^{(1)}$$ and $$n_1\mathbf q_{n}^{(2)}-J_{n_1,n_2}\mathbf q_{n}^{(1)}=(n_1+n_2)\mathbf q_{n}^{(2)},$$  that in turn breaks down to $\mathbf q_{n}^{(1)}=-\frac{1}{n_1}J_{n_1, n_2}\mathbf q_{n}^{(2)}$ and  $\mathbf q_{n}^{(2)}=-\frac{1}{n_2}J_{n_2, n_1}\mathbf q_{n}^{(1)}$. Either equation implies that the elements of $\mathbf q_{n}^{(1)}$ must be identical. Likewise for the elements of $\mathbf q_{n}^{(2)}$. Then if $\mathbf q_{n}^{(1)}$ can be written as  $\mathbf q_{n}^{(1)}=(q,\dots,q)^T,~q>0$ then $\mathbf q_{n}^{(2)}=-\frac{n_1}{n_2}(q,\dots,q)^T$ with the corresponding dimensions. The normality condition then implies $\mathbf q_n^{(1)}= \sqrt{\frac{n_2}{n_1 n}} \mathbbm 1_{n_1}$ and $\mathbf q_n^{(2)}= -\sqrt{\frac{n_1}{n_2 n}}\mathbbm 1_{n_2} $ 

When $\lambda_k=n_1$ for $k=2\dots,n_2$ the eigenvalue equation breaks down to 
\begin{equation*}
\begin{bmatrix}
(n_2-n_1)I_{n_1} & -J_{n_1, n_2} \\
-J_{n_2, n_1} & \mathbf 0_{n_2, n_2}
\end{bmatrix}\begin{bmatrix}
\mathbf q_{i}^{(1)} \\ \mathbf q_{i}^{(2)}
\end{bmatrix}=0
\end{equation*} so that a choice of linearly independent eigenvectors could be: $\mathbf e_{n_1+1}-\mathbf e_{n_1+k},~k=2,\dots,n_2$. The Gram-Schmidt process yields the orthonormal basis: \begin{equation*}
\begin{split}
\mathbf q_k&=\frac{1}{||\mathbf w_k||}\mathbf w_k \hspace{0.2in} \text{for} \hspace{0.2in}\mathbf w_k=\frac{1}{i-1}\sum_{j=n_1+1}^{n_1+k-1}\mathbf e_j-\mathbf e_{n_1+k},
\end{split}
\end{equation*}  with $k=2,\dots,n_2.$

When $\lambda_k=n_2$, $k=n_2+1,\dots,n-1$ a choice of linearly independent eigenvectors could be: $\mathbf e_{1}-\mathbf e_{k},~k=2,\dots,n_1$. The Gram-Schmidt process yields the orthonormal basis for the space spanned by the eigenvalue $n_2$
\begin{equation*}
\mathbf q_{k}=\frac{1}{||\mathbf w_k||}\mathbf w_k \hspace{0.2in} \text{for} \hspace{0.2in}\mathbf w_k=\frac{1}{k-n_2}\sum_{j=1}^{k-n_2}\mathbf e_j-\mathbf e_{k-n_2+1}, 
\end{equation*}  with $k=n_2+1,\dots,n-1$.
\end{proof}

Let now $\mathbf c=\mathbf m_i$ for $i\in [n]$ the output vector on the deviation of the $i^{th}$ agent from the average. We take cases.

For the eigenvalue $\lambda_k=n_1$ with multiplicity $n_2-1$:

\noindent[1.] $i=1,\dots,n_1:$ then $\mathbf m_i^T\mathbf q_k=0$ for $k=2,\dots,n_2$, hence $\sum_{k=2}^{n_2}(c^Q_k)^2=0$

\noindent[2.] $i=n_1+1,\dots,n:$ then $\mathbf c^T\mathbf q_k=\frac{w_k^{(i)}}{||\mathbf w_k||}$ and the sum $\sum_{k=2}^{n_2}(c^Q_k)^2$ yields
\begin{equation*}
\sum_{k=2}^{n_2}(c^Q_k)^2=\begin{cases}
 \sum_{j=1}^{n_1-1} \bigg(\sqrt{\frac{j}{j+1}}\frac{1}{j}\bigg)^2, i=n_1+1,n_1+2 \\
\bigg(\sqrt{\frac{l-1}{l}}\bigg)^2+ \sum_{j=l}^{n_1-1}\bigg( \sqrt{\frac{j}{j+1}}\frac{1}{j}\bigg)^2, i=n_1+l,\\  \hspace{1.8in} l=3,...,n_2-1
\end{cases}
\end{equation*}
Note that all the non-zero sums are equal to $1-\frac{1}{n_1}$ and independent of $i$. For the eigenvalue $\lambda_k=n_2$ with multiplicity $n_1-1$:

\noindent[1.]  $i=1,\dots,n_1:$ then $\mathbf c^T\mathbf q_k=\frac{w_k^{(i)}}{||\mathbf w_k||}$ and the sum $\sum_{k=n_2+1}^{n-1}(c^Q_k)^2$ yields
\begin{equation*}
\sum_{k=n_2+1}^{n-1}(c^Q_k)^2=\begin{cases}
 \sum_{j=1}^{n_2-1} \bigg(\sqrt{\frac{j}{j+1}}\frac{1}{j}\bigg)^2, \hspace{0.7in} i=1,2 \\
\bigg(\sqrt{\frac{l-1}{l}}\bigg)^2+ \sum_{j=l}^{n_2-1}\bigg( \sqrt{\frac{j}{j+1}}\frac{1}{j}\bigg)^2,  i=l,\\  \hspace{1.8in} ~l=3,...,n_1
\end{cases}
\end{equation*}

\noindent[2.] $i=n_1+1,\dots,n:$ 
then $\mathbf m_i^T\mathbf q_k=0$ for $k=n_2+1,\dots,n-1$, hence $\sum_{k=n_2+1}^{n-1}(c^Q_k)^2=0$

but note that all the non-zero sums are equal to $1-\frac{1}{n_2}$ and independent of $i$. Finally $(\mathbf m_i^T\mathbf q_n)^2=\frac{n_2}{n_1 n}$ for $i=1,\dots,n_1$ and  $(\mathbf m_i^T\mathbf q_n)^2=\frac{n_1}{n_2 n}$ for $i=n_1+1,\dots,n$. 
\begin{equation*}
\bar{\sigma}^2_i=\begin{cases}
\big(1-\frac{1}{n_2}\big)\tau f(n_2\tau)+\frac{n_2}{n_1 n}\tau f(n\tau), & i=1,\dots,n_1\\
\big(1-\frac{1}{n_1}\big)\tau f(n_1\tau)+\frac{n_1}{n_2 n}\tau f(n\tau), & i=n_1+1,\dots,n. \\
\end{cases}
\end{equation*}

Our objective now is to determine which of the two groups $G_1=\{1,\dots,n_1\}$, $G_2=\{n_1+1,\dots,n\}$ attains the largest steady state variance. Numerical explorations yield the following cases:

\noindent[1.] $n_1=1$. Then $G_1=\{1\}$ and this is star graph topology. In this case it is easy to verify that $\bar{\sigma}_1^2>\bar{\sigma}_j^2$ for $j\neq 1$. 
\noindent[2.]  $n_1>1$. Here $G_1$ seizes to be a singleton. It can be shown that there is $\tau^*$ such that 
\begin{equation*}
\bar{\sigma}_{G_1}^2(\tau)\leq \bar{\sigma}_{G_2}^{2}(\tau)~\hspace{0.1in}\text{for}~\tau\in[0,\tau^*]
\end{equation*} and 
\begin{equation*}
\bar{\sigma}_{G_1}^2(\tau)\geq \bar{\sigma}_{G_2}^{2}(\tau)~\hspace{0.1in}\text{for}~\tau\in\big(\tau^*,\frac{\pi}{2n}\big).
\end{equation*}
To prove the claim above, consider the expression $$I(x,y,\tau)=\bigg(1-\frac{1}{x}\bigg)\tau f((x+y)\tau)+\frac{x}{y(x+y)}\tau f((x+y)\tau).$$
Note that $I(x,y,0)=\big(1-\frac{1}{x}\big)\frac{1}{2x}+\frac{x}{2y(x+y)^2}$. For $x=y+k$ (for $x$ as in $n_2$ that is greater or equal to $n_1$), the expression
\begin{equation*}\begin{split}
&I(y+k,y,0)-I(y,y+k,0)=-\frac{1}{2}\frac{(2y^3+(3k-5)y^2+(k^2-5k)y-k^2)k}{(2y+k)y^2(y+k)^2}
\end{split}
\end{equation*} is strictly negative for $y>1$ and $k>1$. In addition it is easy to check that $\frac{\partial}{\partial \tau}I(x,y,\tau)>0$ for $\tau$ within the domain of definition (equivalently one can plot $\frac{\cos(z)}{1-\sin(z)}$ and observe that it is monotonically increasing for $z\in (0,\pi/2)$). Finally for $\tau\rightarrow \frac{\pi}{2n}$ we can see that $n_2\geq n_1$ implies that $I(n_2,n_1,\tau)>I(n_1,n_2,\tau)$. This concludes the proof of our claim.

\subsection*{Graph path of $n-1$ hops, $\mathbf P_n$} The path graph over $n$ nodes has $n-1$ edges (hops) and its eigenstructure is $\lambda_k=2\big(1-\cos\big(\frac{\pi (k-1)}{n}\big)\big),~k=1,\dots,n$ and $\mathbf q_1=\frac{1}{\sqrt{n}}\mathbbm 1$, $\mathbf q_k=(q_k^{(1)},\dots,q_k^{(n)})$ with $q_{k}^{(l)}=\sqrt{\frac{2}{n}}\cos\big(\frac{\pi k}{2n}(2l-1)\big)$ for $k=2,\dots,n$, consist an orthonormal basis, \cite{Mieghem:2011:GSC:1983675}. The $i^{th}$ component's deviation from the average, has a steady state variance 
\begin{equation*}
\bar{\sigma}_i^{2}=\frac{2\tau}{n}\sum_{k=2}^{n}\cos^2\bigg(\frac{\pi k}{2n}(2i-1)\bigg)f(\lambda_k\tau),
\end{equation*}
provided that $\tau<\frac{\pi}{4}$. In this case symmetry implies $\bar{\sigma}_i^2=\bar{\sigma}_{n-i+1}^2$ and $\frac{n}{2}$ or $\frac{n+1}{2}$ group of nodes with identical risk values.

\subsection*{Ring graph, $\mathbf R_n$}
The eigenstructure of the ring graph is $\lambda_1=0$ with eigenvector $\mathbf q_1=\frac{1}{\sqrt{n}}\mathbbm 1$ and $\lambda_k=2-2\cos\big(\frac{2\pi(k-1)}{n}\big)$ with eigenvector $\mathbf q_k=(1,e^{i\frac{2\pi(k-1)}{n}},e^{i\frac{2\pi(k-1)}{n}2},\dots,e^{i\frac{2\pi(k-1)}{n}(n-1)}\big)$ for $k=2,\dots,n$. The $i^{th}$ component's deviation from the average, has a steady state variance 
\begin{equation*}
\bar{\sigma}_i^{2}=\tau\sum_{k=2}^{n}f(\lambda_k\tau)
\end{equation*}
for every $i=1,\dots,n$.

\bibliographystyle{plain}    
\bibliography{bibliography} 
\vspace{-0.9cm} 

\end{document}